\documentclass[12pt]{article}
\usepackage[font=small]{caption}
\usepackage[english]{babel}
\usepackage{float}
\usepackage{rotating, graphicx}
\usepackage{amsmath}
\usepackage{amssymb, amsthm}
\usepackage{array}
\usepackage[T1]{fontenc}  
\usepackage[round]{natbib}
\usepackage{epstopdf}
\usepackage[table]{xcolor}    
\usepackage{framed}

\renewcommand{\arraystretch}{1.2}

\DeclareMathOperator{\E}{\mathsf{E}}   
\DeclareMathOperator{\Var}{\mathsf{Var}}   
\DeclareMathOperator{\gP}{\mathsf{P}}   

\newcolumntype{L}[1]{>{\raggedright\let\newline\\\arraybackslash\hspace{0pt}}m{#1}}
\newcolumntype{C}[1]{>{\centering\let\newline\\\arraybackslash\hspace{0pt}}m{#1}}
\newcolumntype{R}[1]{>{\raggedleft\let\newline\\\arraybackslash\hspace{0pt}}m{#1}}

\newcounter{problem}
\newenvironment{problem}[1][]{\refstepcounter{problem}\par\begin{framed}\textbf{Problem~\theproblem: #1}}{\end{framed}}

\newtheorem{theorem}{Theorem}[section]
\newtheorem{proposition}[theorem]{Proposition}
\newtheorem{corollary}[theorem]{Corollary}
\newtheorem{lemma}[theorem]{Lemma}

\begin{document}

\pagenumbering{arabic}

\title{Strategic mean-variance investing under mean-reverting stock returns}
\date{March, 2021}
\author{S{\o}ren Fiig Jarner \\[3mm] Danish Labour Market Supplementary Pension Fund (ATP)}
\maketitle

\begin{abstract}
In this report we derive the strategic (deterministic) allocation to bonds and stocks resulting in the optimal mean-variance trade-off on a given investment horizon. The underlying capital market features a mean-reverting process for equity returns, and the primary question of interest is how mean-reversion effects the optimal strategy and the resulting portfolio value at the horizon. In particular, we are interested in knowing under which assumptions and on which horizons, the risk-reward trade-off is so favourable that the value of the portfolio is effectively bounded from below on the horizon. In this case, we might think of the portfolio as providing a stochastic excess return on top of a ``guarantee'' (the lower bound).

Deriving optimal strategies is a well-known discipline in mathematical finance. The modern approach is to derive and solve the Hamilton-Jacobi-Bellman (HJB) differential equation characterizing the strategy leading to highest expected utility, for given utility function. However, for two reasons we approach the problem differently in this work. First, we wish to find the optimal strategy depending on time only, i.e., we do not allow for dependencies on capital market state variables, nor the value of the portfolio itself. This constraint characterizes the strategic allocation of long-term investors. Second, to gain insights on the role of mean-reversion, we wish to identify the entire family of extremal strategies, not only the optimal strategies. To derive the strategies we employ methods from calculus of variations, rather than the usual HJB approach.
\end{abstract}

\bigskip

\noindent\emph{Keywords:} Deterministic strategies, mean-variance optimization, mean-reverting stock returns, calculus of variations.

\bigskip

\thispagestyle{empty}
\newpage

\tableofcontents
\newpage

\section{Introduction}
This is a technical document providing the theoretical foundation for a study of the risk-reward trade-off on long horizons in a capital market under which equity returns are mean-reverting. The problem is stated and solved in continuous-time based on the model analyzed in \cite{jarpre17}, see also \cite{munetal04}.

We are particularly interested in studying how the risk-reward trade-off depends on the degree of mean-reversion and the horizon, and under which assumptions the trade-off is so favourable that the (optimal) portfolio is effectively bounded from below. In this case, one might argue that the portfolio can be used as a ``hedge'' for a guaranteed payment, equal to the lower bound of the portfolio.

The optimal strategy itself is of course also of interest. In a model without mean-reversion it is optimal to hold a constant share of equities over time. This is intuitively clear, since under this assumption a given equity return has the same probability and the same effect on the (final) portfolio value regardless of when it occurs. In other words, it is equally likely and equally bad to lose, say, 10\% of your investment on the first day as it is to lose it on the last day.

Under mean-reverting equity returns the situation is more complicated. Intuitively we would expect early losses to have a smaller effect than late losses, since an early loss has a higher chance of being ``reverted'' than a late loss. Qualitatively, this ought to imply a higher (optimal) equity exposure in the beginning of the period and a lower equity exposure at the end of the period. Quantitatively, however, we need to take account of the fact that we only benefit from mean-reversion if we maintain an equity exposure throughout. Loosely speaking, to benefit from the assumed mean-reversion we need to hold equities when they (hopefully) rebound later on. Essentially, ``early'' risk is better rewarded than ``late'' risk, but only if we take ``late'' risk. Mathematically, the optimal strategy needs to balance early and late risk taking account of the benefit late risk has on early risk.

We are looking for optimal strategies depending on time only. More precisely, the strategy is allowed to depend on the {\em initial value} of the state variables only. In particular, the strategy is not allowed to depend on the unobservable, stochastic excess return on equities as it changes over the period, nor on the value of the portfolio. This assumption makes the strategy less tailored to the specific capital market model used to derive it, and it is therefore reasonable to expect that the strategy is also close to optimal under other, less stylized, types of mean-reversion. The assumption also ensures that the strategies can be implemented and tested in practice, without the need to identify latent variables. \cite{chrste15, chrste18} find optimal, deterministic strategies motivated by the fact that marketed life-cycle products are deterministic, but apart from their work only few results exist in the literature.

Focusing on deterministic strategies implies that the portfolio is log-normally distributed on all horizons, in the model employed. The problem can thereby be stated as optimizing the mean for given variance of the (log) portfolio value on the horizon of interest. We focus primarily on optimal equity strategies, but we also consider optimal rate strategies and jointly optimal strategies. The optimal equity strategies are derived by methods from calculus of variations in three steps. First, we derive an integral equation characterizing the optimal strategy. Second, we use the integral equation to derive a second order differential equation, which gives us the general solution. Finally, we derive and solve a system of equations for determining the constants of the general solution. The solution is explicit up to the presence of a Lagrange multiplier, which when varied gives the entire family of extremal, including optimal, strategies for varying levels of variance. We also provide illustrations of the results, discussing in particular the role of mean-reversion.

\subsection{Outline}
The rest of the paper is organized as follows. In Section~\ref{sec:Prelim} we set up the model and provide the distributional results on the portfolio value on future horizons needed for the optimization, in Section~\ref{sec:OptimalMV} we formulate and solve three optimization problems, finding optimal rate and equities strategies, both separately and jointly. The results are illustrated and discussed in Section~\ref{sec:numeric}, and Section~\ref{sec:Conclusion} summarizes and concludes on the findings. All proofs and technical details are in the appendices.

\pagebreak
\section{Preliminaries} \label{sec:Prelim}
In this section we state the underlying capital market model. We refer to the model as a factor model due to its intended use as a sparse representation of generic ``rates'' and ''equities''. Accordingly, we interpret the resulting strategies as profiles for generic interest and equity risk, respectively. How to obtain this risk in practise, i.e., the selection of specific bonds and stocks, is outside the scope of the model.

After introducing the capital market model, we state a distributional result for portfolios resulting from time-dependent strategies of the kind we will be considering. The formulation and solution of the optimization problem rely on this result and the accompanying integral representation.

\subsection{Capital market model}
The capital market model is described in \cite{jarpre17}, see also \cite{munetal04}, but for ease of reference we restate it here. The model of \cite{jarpre17} also contains realized inflation and break-even inflation (BEI) curves for pricing inflation-indexed bonds, and inflation swaps. Here, however, we disregard inflation and include only nominal interest rates and equities in the model.

Hence, the model features
\begin{itemize}
  \item A stochastic short rate ($r_t$)
  \item A bond market of all maturities with stochastic risk premium ($\lambda_t^r$)
  \item A stochastic equity index ($S_t$)
  \item A mean-reverting equity risk premium ($x_t$)
\end{itemize}

Since we are interested in portfolio optimization we need the so-called 'real world' dynamics of the state variables of the model (sometimes referred to as the $P$-dynamics as opposed to the $Q$-dynamics used for pricing). In addition to these, we also need the evolution of the term structure of interest rates (the yield curve) for pricing bonds. 

We assume that the short (nominal) interest rate follows an Ornstein-Uhlenbeck process,
\begin{align}  \label{eq:shortrate}
   dr_t = \kappa(\bar{r} - r_t)dt + \sigma_r dW^r_t,
\end{align}
where $\bar{r}$ is the long-run mean of the short interest rate, $\kappa$ describes the degree of mean reversion, $\sigma_r$ is the interest rate volatility, and $W^r$ is a standard Brownian motion.

The stock index (total return index) is assumed to evolve according to the dynamics
\begin{align}  \label{eq:stockindex}
   \frac{dS_t}{S_t} = (r_t + x_t)dt + \sigma_S dW^S_t,
\end{align}
where $r_t$ is the short rate from (\ref{eq:shortrate}), $x_t$ is the time-varying risk premium (expected excess return) from investing in stocks, $\sigma_S$ is the stock index volatility, and $W^S$ is a standard Brownian motion. We further assume that the risk premium follows an Ornstein-Uhlenbeck process,
\begin{align}  \label{eq:excessreturn}
   dx_t = \alpha(\bar{x} - x_t)dt - \sigma_x dW^S_t,
\end{align}
where $\bar{x}$ denotes the long-run equity risk premium, $\alpha$ describes the degree of mean reversion towards this level, and $\sigma_x$ is the risk premium volatility. We assume joint normality of the two Brownian motions $W^r$ and $W^S$ with correlation coefficient $\rho$.

Note that the stock index and the risk premium processes are locally perfectly negatively correlated, i.e., a stock return above or below its expected value will ``cause'' a change in the (future) risk premium in the opposite direction. This interaction induces a mean-reversion in the stock returns over time.

Finally, we assume that the term structure of interest rates is of the form considered by \cite{vas77}. Specifically, we assume that the (arbitrage free) price at time $t$ of a zero-coupon bond maturing at time $T\geq t$ is given by
\begin{align}
    p_t(T) = \exp\left\{G(\Delta) - H(\Delta)r_t\right\},  \label{eq:VasZCBPrice}
\end{align}
with $\Delta = T-t$,
\begin{align}
  H(\Delta) & = \frac{1}{a}\left(1-\exp\{-a \Delta\}\right), \\
  G(\Delta) & = \left(b - \frac{\sigma_r^2}{2a^2}  \right)\left(H(\Delta)-\Delta \right) - \frac{\sigma_r^2}{4a}H^2(\Delta), \label{eq:VasGdef}
\end{align}
and where $a$ and $b$ are parameters controlling the slope and level of the yield curves.\footnote{The price of a zero-coupon bond can be obtained by the usual risk-neutral valuation formula, $p_t(T)=\E_t^Q\left[-\int_t^T r_sds\right]$, where  $dr_t = a(b-r_t)dt + \sigma_r d\bar{W}^r_t$ with $\bar{W}^r$ being a $Q$-Brownian motion.}
The specification corresponds to the market price of interest rate risk being equal to
\begin{equation}\label{eq:lambdar}
  \lambda^r_t = \{(a-\kappa)r_t + \kappa \bar{r} - a b\}/\sigma_r.
\end{equation}

As shown in Section~5.1 of \cite{jarpre17}, the price dynamics of $p_t(T)$ for fixed $T$ has the form
\begin{equation}  \label{eq:pTdyn}
  \frac{dp_t(T)}{p_t(T)} = (r_t - \lambda^r_t\Psi(a,T-t)\sigma_r) dt - \Psi(a,T-t)\sigma_r dW^r_t,
\end{equation}
where
\begin{align} \label{eq:psidef}
  \Psi(a,t) \equiv \int_0^t e^{-a u}du =
  \begin{cases}
  t & \mbox{for } a = 0, \\
  \frac{1}{a}\left(1-e^{-a t} \right) & \mbox{for } a \neq 0. \\
  \end{cases}
\end{align}

From (\ref{eq:pTdyn}) we see that the (negative) market price of interest rate risk, $-\lambda^r_t$, can be interpreted as the risk-reward trade-off at a given point in time, i.e., as the excess return per unit of volatility. It is the continuous-time analogue to the Sharpe ratio used in portfolio construction and benchmarking, cf.\ \cite{sha66, sha94}. In general, the market price of interest rate risk depends on the current short rate, in particular, the market price of interest rate risk is stochastic. For subsequent use, we define the market price of equity risk as the excess return per unit of volatility when investing in equities,
\begin{equation}\label{eq:lambdaS}
  \lambda^S_t = x_t/\sigma_S.
\end{equation}

\subsection{Portfolio dynamics}
We now consider the dynamics of a portfolio exposed to (interest) rate and equity risk. In general, the dynamics of the portfolio value, $V_t$, is given by
\begin{equation}\label{eq:Vdyn}
   \frac{dV_t}{V_t} = \left(r_t + f_t' \lambda_t \right)dt + f_t' dW_t,
\end{equation}
where $f_t=(f^r_t,f^S_t)'$ is the vector of (factor) exposures to rate and equity risk, respectively, $\lambda_t=(\lambda_t^r,\lambda_t^S)'$ is the vector of market prices of risk given by (\ref{eq:lambdar}) and (\ref{eq:lambdaS}), and $W_t = (W_t^r,W_t^S)'$ is the vector of driving Brownian motions.

The exposures are measured in volatility and formula (\ref{eq:Vdyn}) succinctly states that the excess return, i.e., the return in excess of the short rate, is the market prices of risk weighted by the exposure. As an example, if $f^r_t \equiv 0$ and $f^S_t \equiv \sigma_S$ we get
\begin{equation}
   \frac{dV_t}{V_t} = \left(r_t + \sigma_S \lambda^S_t \right)dt + \sigma_S dW^S_t = \left(r_t + x_t \right)dt + \sigma_S dW^S_t,
\end{equation}
which is equal to the dynamics of the equity index, as expected.

In general, assuming only that the strategies are adapted and sufficiently regular, we have the following integral representation of the portfolio process
\begin{equation}\label{eq:Vintrep}
   V_t = V_0\exp\left\{\int_0^t r_s ds + \int_0^t f_s'\lambda_s ds - \frac{1}{2}\int_0^t f_s' C f_sds + \int_0^t f_s'dW_s   \right\},
\end{equation}
where $C$ is the (instantaneous) correlation matrix given by
\begin{equation}\label{eq:Cdef}
   C = \begin{pmatrix}
         1 & \rho \\
         \rho & 1
       \end{pmatrix}.
\end{equation}
This follows by an application of the multi-dimensional version of It\^{o}'s lemma, see e.g.\ Proposition 4.18 of \cite{bjo09}. Let $X_t = \int_0^t r_s ds + \int_0^t f_s'\lambda_s ds - \frac{1}{2}\int_0^t f_s' C f_sds + \int_0^t f_s'dW_s$, and use It\^{o}'s lemma to calculate the differential of $V_t = G(X_t)$, where $G(x)= V_0 \exp(x)$,
\begin{eqnarray*}
  d V_t &=&  \frac{\partial G}{\partial x}dX_t + \frac{1}{2}\frac{\partial^2 G}{\partial x^2} dX_t dX_t \\[2mm]
        &=& V_t\left[\left(r_t + f_t'\lambda_t - \frac{1}{2}f_t' C f_t\right) dt + f_t'dW_t \right] + \frac{1}{2}V_t\left[ \left(f^r_t\right)^2 + \left(f^S_t\right)^2 +2\rho f^r_t f^S_t \right]dt \\[2mm]
        &=& V_t \left[\left(r_t + f_t'\lambda_t \right)dt + f_t'dW_t \right],
\end{eqnarray*}
which is the same as (\ref{eq:Vdyn}). Since the right-hand side of (\ref{eq:Vintrep}) equals $V_0$ for $t=0$  we conclude (under standard regularity conditions) that (\ref{eq:Vintrep}) is the solution to (\ref{eq:Vdyn}) as claimed. In the second equality of the above calculations we have used the formal multiplication rules: $dW^r_t dW^r_t = dW^S_t dW^S_t = dt$ and $dW^r_t dW^S_t = \rho dt$.

\subsection{Portfolio distribution}
As mentioned in the introduction, we are aiming at finding the optimal risk-reward trade-off for a deterministic strategy, i.e., a strategy that depends on time only. We are aided in this task by the fact that under this restriction, the portfolio is log-normally distributed at every horizon. In this section we derive the necessary distributional results.

We need the following integral representations for the risk premium, the short rate and the integrated short rate from Section~3 of \cite{jarpre17}:
\begin{align}
  x_t & = \bar{x} + e^{-\alpha t}(x_0 - \bar{x}) - \sigma_x \int_0^t e^{-\alpha(t-s)}dW^S_s, \label{eq:xsol} \\[2mm]
  r_t & = \bar{r} + e^{-\kappa t}(r_0 - \bar{r}) + \sigma_r \int_0^t e^{-\kappa(t-s)}dW^r_s, \label{eq:rsol} \\[2mm]
 \int_0^t r_s ds    & =  t \bar{r} + (r_0-\bar{r})\Psi(\kappa,t) + \sigma_r \int_0^t \Psi(\kappa,t-s) dW^r_s, \label{eq:intrsol}
\end{align}
where $\Psi$ is given by (\ref{eq:psidef}). Since $\lambda^r_t$ and $\lambda^S_t$ are linear in $r_t$ and $x_t$, respectively, it follows from the expressions above that the integrals at the right-hand side of (\ref{eq:Vintrep}) are either deterministic, or stochastic integrals with deterministic integrands and Brownian integrators. From this observation, the claimed log-normality follows; the mean and variance are given by the following theorem.

\begin{theorem} \label{thm:VTrep}
Assume that $f_t = (f^r_t,f^S_t)'$ is deterministic (depends on time only). For given $T\geq 0$, the portfolio value at time $T$ can be expressed as
\begin{equation}\label{eq:Vrep}
   V_T = V_0\exp\left\{m^0_T + m^r_T + m^S_T - \rho\int_0^T f^r_u f^S_u du + \int_0^T h_u^r dW^r_u + \int_0^T h^S_u dW^S_u \right\},
\end{equation}
where
\begin{align}
  m^0_T & = T \bar{r} + (r_0-\bar{r})\Psi(\kappa,T),  \label{eq:m0Tdef} \\[2mm]
  m^r_T & = \frac{a(\bar{r}-b)}{\sigma_r}\int_0^T f^r_s ds + \frac{a-\kappa}{\sigma_r}\int_0^T e^{-\kappa s}(r_0 - \bar{r})f^r_s ds - \frac{1}{2}\int_0^T \left(f^r_s\right)^2ds,  \\[2mm]
  m^S_T & = \frac{1}{\sigma_S}\int_0^T f^S_s\left(\bar{x} + e^{-\alpha s}(x_0 - \bar{x})\right)ds - \frac{1}{2}\int_0^T \left(f^S_s\right)^2ds, \label{eq:mSTdef} \\[2mm]
  h^r_u & = \sigma_r\Psi(\kappa,T-u) + f^r_u + (a-\kappa)\int_u^T f^r_s e^{-\kappa(s-u)}ds,  \\[2mm]
  h^S_u & = f^S_u - \frac{\sigma_x}{\sigma_S}\int_u^T f^S_s e^{-\alpha(s-u)}ds.  \label{eq:hSudef}
\end{align}
In particular, $V_T/V_0$ is log-normally distributed with mean and variance given by
\begin{align}
   \mu_T      & = m^0_T + m^r_T + m^S_T - \rho\int_0^T f^r_u f^S_u du,  \label{eq:VTmean} \\[2mm]
   \sigma_T^2 & = \int_0^T \left(h_u^r\right)^2 + \left(h^S_u\right)^2 + 2\rho h^r_u h^S_u du. \label{eq:VTvar}
\end{align}
\end{theorem}

Note that the weight functions $h^r$ and $h^S$, which measure the impact of the risk sources $W^r$ and $W^S$ on the horizon value of the portfolio, contain both a direct and an indirect effect. Due to mean-reversion, the risk (innovation) at time $u$ measured by $dW^r_u$ and $dW^S_u$, respectively, has both a direct effect which depends on the exposure at time $u$ (the terms $f^r_u$ and $f^S_u$ in $h^r_u$ and $h^S_u$, respectively) and an indirect effect which depends on the {\em subsequent} exposure to the risk source (the integrals from $u$ to $T$). The direct effect depends only on the current exposure, while the indirect effect depends over the (future) exposure on the remaining horizon.

We also note that the indirect effect for rates disappears if $a=\kappa$, which corresponds to the case where the market price of interest rate risk is constant, $\lambda^r_t \equiv \kappa(\bar{r}-b)/\sigma_r$. Thus, it is not the mean-reverting nature of the short rate process itself, but rather the state dependent market price of risk, that causes the indirect effect for rates. For equities, the indirect effect is caused by the feedback mechanism from equity risk to (future) equity risk premia. The size of this effect depends on the ``loading'' ratio $\sigma_x/\sigma_S$ and the degree of persistence, $\alpha$, in the equity risk premium process.

In principle, Theorem~\ref{thm:VTrep} can be used for joint optimization of the mean-variance trade-off at a given horizon. However, to make the computations easier to follow and to highlight the structure of the solution in special cases, we solve the problem in stages under various simplifying assumptions. We are primarily interested in the effect of mean-reverting equity returns and consequently we keep this feature, while we make the simplifying assumptions that the market price of interest rate risk is constant ($a=\kappa$) and that rate and equity factors are independent ($\rho=0$).

We state without proof a simplification of Theorem~\ref{thm:VTrep} to be used in the following.

\begin{corollary} \label{cor:VTratesonly}
Assume that $f^r_t$ is deterministic, $f^S_t \equiv 0$ and $a=\kappa$. For given $T > 0$, $V_T/V_0$ is log-normally distributed with mean and variance given by
\begin{align}
   \mu_T      & = m^0_T + \lambda^r \int_0^T f^r_s ds - \frac{1}{2}\int_0^T \left(f^r_s\right)^2ds,  \label{eq:VTmeanratesonly} \\[2mm]
   \sigma_T^2 & = \int_0^T \left[\sigma_r\Psi(\kappa,T-s) + f^r_s  \right]^2 ds,   \label{eq:VTvarratesonly}
\end{align}
where $m^0_T$ is given by (\ref{eq:m0Tdef}) of Theorem~\ref{thm:VTrep}, and $\lambda^r = \kappa(\bar{r}-b)/\sigma_r$.
\end{corollary}

\pagebreak
\section{Optimal mean-variance strategies} \label{sec:OptimalMV}
We are now ready to state the optimization problem(s) and derive the solution. As previously announced we proceed in stages: rates only, equities only, and finally combined rates and equities assuming independence of the two factors.

\subsection{Optimal rate strategies} \label{sec:OptimalBond}
We first want to find the deterministic, rates only investment strategies that lead to the optimal (log) mean-variance trade-off on a given horizon. We consider only the case of a constant market price of interest rate risk, corresponding to $a=\kappa$. Clearly, a constant market price of interest rate risk is a mathematical simplification. However, at this stage we are primarily interested in getting an intuition for the structure of the optimal strategies in a simple setup.\footnote{The solution technique developed in the next section for handling mean-reverting stock returns can be used to solve the rates only problem for general $a$, but the resulting strategies are less intuitive and harder to interpret.}

Corollary~\ref{cor:VTratesonly} gives the mean and variance of the log portfolio value and the problem can therefore be phrased as

\begin{problem} \label{problem:RatesOnly}
  Find $f^r$ such that (\ref{eq:VTmeanratesonly}) is maximized for given value of (\ref{eq:VTvarratesonly}).
\end{problem}

For given horizon, $T$, the solution to Problem~\ref{problem:RatesOnly} takes the form of a family of strategies, one for each value of the variance $\sigma_T^2$. Members of this family are referred to as {\em optimal bond strategies}. More generally, we refer to bond strategies maximizing, or minimizing, (\ref{eq:VTmeanratesonly}) for given value of (\ref{eq:VTvarratesonly}) as {\em extremal bond strategies}.

\begin{theorem} \label{thm:OptimalBond}
Assume $a=\kappa$ such that the market price of interest rate risk is constant, $\lambda^r = \kappa(\bar{r}-b)/\sigma_r$. For given horizon $T > 0$, let  $g^r_s = \sigma_r\Psi(\kappa,T-s)$ for $0\leq s\leq T$.
The extremal bond strategies on horizon $T$ are of the form
\begin{equation}
    f^r_s = \frac{\lambda^r+2\nu g^r_s}{1-2\nu} \quad (0\leq s\leq T), \label{eq:OptimalBond}
\end{equation}
where $\nu$ depends on the prescribed value of the variance $\sigma_T^2$.

The optimal bond strategies have $\nu < 1/2$, and the bond strategies minimizing the mean for given value of the variance have $\nu > 1/2$.
\end{theorem}

For $\nu=0$, Theorem~\ref{thm:OptimalBond} yields $f_s^r=\lambda^r$, which is the strategy (globally) maximising (\ref{eq:VTmeanratesonly}). Recall from (\ref{eq:pTdyn}) that long positions in bonds correspond to negative exposures to the underlying Brownian motion, $W^r$, and that positive excess returns for bonds correspond to $\lambda^r<0$. Thus, the maximizing strategy, $f_s^r=\lambda^r$, is long in bonds when bonds have positive excess returns, and short in bonds when bonds have negative excess returns. We typically expect positive excess returns for bonds, but we cannot rule out the possibility of negative excess returns, $\lambda^r>0$, neither mathematically, nor in practise. Also note, that the maximising strategy amounts to a constant exposure to interest rate risk over time; this is the only constant, optimal strategy.

The proof of Theorem~\ref{thm:OptimalBond} is a variational argument where $\nu$ plays the role of a Lagrange multiplier introduced to handle the variance constraint. At an extremal point (strategy) the gradients of the mean functional (the objective) and the variance functional (the constraint) are parallel with proportionality coefficient $-\nu$. If $\nu<0$ we are at a point where relaxing the constraint, i.e., allowing a larger variance, leads to a larger mean. Conversely, if $\nu > 0$ we are at a point where a larger variance leads to a smaller mean. With this in mind, we can categorize the optimal strategies of Theorem~\ref{thm:OptimalBond} as follows.

In the limit $\nu$ tending to minus (or plus) infinity, we have the strategy $f_s^r=-g_s^r$, which corresponds to the buy-and-hold strategy of a bond that matures at time $T$, cf.~(\ref{eq:pTdyn}). This strategy has $\sigma_T^2=0$, and it is the optimal (and only) strategy with no variance on the horizon. For $\nu<0$ we obtain optimal strategies where added risk is rewarded, i.e., as $\nu$ is increased from minus infinity to zero we get optimal strategies with higher variance and higher mean. For $\nu=0$ we obtain the maximal mean possible. For $0 < \nu < 1/2$, we have optimal strategies where added risk is penalized, i.e., strategies where the prescribed value of the variance can only be achieved by exposures so high that the quadratic term in (\ref{eq:VTmeanratesonly}) impairs the mean. Mathematically, these strategies are optimal in the sense that they achieve the highest possible mean given the variance, but in practice they are sub-optimal: for each of these strategies there exist strategies with lower variance and identical or higher mean. Hence, for practical purposes we are interested only in strategies of form (\ref{eq:OptimalBond}) with $\nu\leq 0$.

\subsubsection{Interpretation of the optimal bond strategy} \label{sec:OptimalBondInterpretation}
The optimal strategy of Theorem~\ref{thm:OptimalBond} can be written $f^r_s = a(\nu)(-g_s^r) + (1-a(\nu))\lambda^r$, where $a(\nu)=-2\nu/(1-2\nu)$. For $\nu \leq 0$, we have $0\leq a(\nu) <1$, and hence we can interpret the strategy as a convex combination of two optimal strategies: $T$-bonds and constant risk exposure of size $\lambda^r$. It is natural to interpret the $T$-bonds as a risk-free hedging component, and the constant risk exposure as a return-seeking component. This interpretation is supported by the fact that we always hold a long, unlevered position in $T$-bonds, while the constant risk exposure is obtained by either a long or short bond position depending on the sign of the market price of interest rate risk.

Strictly speaking, this is only one possible way to interpret the optimal strategy. Since we are using a single factor rate model, all bond returns are fully correlated and only the net interest rate exposure matters. Furthermore, the exposure can be taken using bonds of any maturity (with varying degrees of leverage). Still, interpreting (\ref{eq:OptimalBond}) as a combination of variance reducing $T$-bonds ("hedging") and additional long, or short, return-seeking bond positions ("investments") provides a useful intuition. The structure and the interpretation closely resemble the classic two-fund separation theorem of modern portfolio theory, see \cite{mar52, tob58, mer72}.

For $0<\nu <1/2$ we still have a linear, but no longer a convex, combination of "hedging" and "investment" strategies. More precisely, for $0<\nu <1/2$ we have $a(\nu)<0$ and $1-a(\nu)>1$. Hence, in this situation we have a short position in $T$-bonds and a long return-seeking position, larger than the optimal level of $\lambda^r$. Thus we are borrowing risk-free funds (as seen from time $T$) and investing these in the optimal return-seeking strategy. Finally, the extremal strategies minimizing the mean have $\nu>1/2$ which correspond to $a(\nu)>1$ and $1-a(\nu)<0$. Essentially, we are spending the variance budget on borrowing funds as expensively as possible and investing them in risk-free $T$-bonds.

\subsubsection{Closed form expressions for the extremal mean and variance} \label{sec:closedFormMeanVar}
For given $\nu$, we can find the portfolio mean and variance by computing the integrals of (\ref{eq:VTmeanratesonly}) and (\ref{eq:VTvarratesonly}), respectively, after substituting $f^r_s$  with the extremal strategy of (\ref{eq:OptimalBond}). Due to the simple form of the extremal strategy it is possible to evaluate these integrals analytically, and we will do so here. However, apart from special cases, e.g., $\nu=\pm\infty$, or $\nu=0$, the resulting expressions are hard to interpret. Instead, in Section~\ref{sec:NumOptimalRate} we will investigate numerically the risk-reward profile obtained by varying $\nu$.

To express the results we introduce the function $\Upsilon$ given by
\begin{align} \label{eq:upsilondef}
  \Upsilon(a,t) \equiv \int_0^t \Psi^2(a,s) ds =
  \begin{cases}
  \frac{t^3}{3} & \mbox{for } a = 0, \\
  \frac{1}{2a^3}\left(-3+2 at + 4e^{-at}- e^{-2at} \right) & \mbox{for } a \neq 0. \\
  \end{cases}
\end{align}
Using this notation we can express the zero-coupon bond price as
\begin{equation} \label{eq:ZCBpriceAltExp}
   p_t(T) = \exp\left\{-\Delta b - \Psi(a,\Delta)(r_t-b) + \frac{\sigma_r^2}{2}\Upsilon(a,\Delta) \right\},
\end{equation}
where $\Delta=T-t$, cf.\ Section~5.1 of \cite{jarpre17}. This expression is valid for all values of $a$. However, recall that the extremal strategies of Theorem~\ref{thm:OptimalBond} are derived under the simplifying assumption $a=\kappa$, such that $\lambda^r$ is constant.

We can now evaluate (\ref{eq:VTmeanratesonly}) with $f_s^r$ given by (\ref{eq:OptimalBond})
\begin{align}
  \mu_T & = \underbrace{T\bar{r} + \Psi(\kappa,T)(r_0-\bar{r})}_{m_T^0} + \underbrace{\frac{T(\lambda^r)^2+2\nu(T-\Psi(\kappa,T))(\bar{r}-b)}{1-2\nu}}_{\lambda^r\int_0^T f_s^rds} -
  \frac{1}{2}\int_0^T (f_s^r)^2ds \nonumber \\[2mm]
  & = T\left[\frac{\bar{r}}{1-2\nu} - \frac{2\nu b}{1-2\nu} + \frac{(\lambda^r)^2}{1-2\nu}- \frac{(\lambda^r)^2}{2(1-2\nu)^2} - \frac{2\nu(\bar{r}-b)}{(1-2\nu)^2}\right] +  \label{eq:muToptimalRateI} \\[2mm]
  & \quad \Psi(\kappa,T)\left[r_0 - \frac{\bar{r}}{1-2\nu} + \frac{2\nu b}{1-2\nu} + \frac{2\nu(\bar{r}-b)}{(1-2\nu)^2} \right] - \frac{\sigma_r^2}{2}\left(\frac{2\nu}{1-2\nu}\right)^2\Upsilon(\kappa,T), \label{eq:muToptimalRateII}
  \end{align}
where we have used that $\lambda^r \sigma_r/\kappa=\bar{r}-b$. Evaluation of the integral in (\ref{eq:VTvarratesonly}) with $f_s^r$ given by (\ref{eq:OptimalBond}) gives us the extremal variance
\begin{align}
  \sigma_T^2 & = \int_0^T [\sigma_r\Psi(\kappa,T-s) + f_s^r]^2ds = \frac{\int_0^T [\lambda^r+ \sigma_r\Psi(\kappa,T-s)]^2ds}{(1-2\nu)^2} \nonumber \\[2mm]
     & = \frac{T[(\lambda^r)^2 + 2(\bar{r}-b)] -\Psi(\kappa,T)2(\bar{r}-b) + \sigma_r^2\Upsilon(\kappa,T)}{(1-2\nu)^2}.  \label{eq:sigmaToptimalRate}
\end{align}
It follows from (\ref{eq:sigmaToptimalRate}) that there are two extremal strategies for any positive specification of the variance, with the pair of Lagrange multipliers characterising the strategies being of the form $\nu = 1/2 \pm \eta$ for $\eta>0$. The strategy with $\nu<1/2$ is mean optimizing, and the strategy with $\nu>1/2$ is mean minimizing. In the limit for $\nu$ tending to plus or minus infinity, we have $\sigma_T^2=0$, i.e., $V_T$ is constant, and
\begin{equation}
   \frac{V_T}{V_0} = \exp(\mu_T) = \exp\left\{Tb + \Psi(\kappa,T)(r_0-b)-\sigma_r^2\Upsilon(\kappa,T)/2 \right\} = \frac{1}{p_0(T)},
\end{equation}
where the first equality follows from Corollary~\ref{cor:VTratesonly}, the second equality follows from (\ref{eq:muToptimalRateI})---(\ref{eq:muToptimalRateII}), and the last equality follows from (\ref{eq:ZCBpriceAltExp}) with $a=\kappa$. Thus, as expected, the return on the risk free strategy equals the return on a zero-coupon bond maturing at the horizon.

\subsection{Optimal equity strategies} \label{sec:OptimalEquity}
In this section we set out to find the equity exposure strategy that leads to the optimal (log) mean-variance trade-off for {\em excess returns} on a given horizon. One might think of this as an overlay strategy to an underlying bond strategy, e.g., one of the strategies derived in Section~\ref{sec:OptimalBond}. In particular, the equity strategy can act as an overlay to the hedging $T$-bond strategy.

For reasons stated earlier, we are interested in deterministic, or strategic, equity exposure profiles only, i.e., we allow the strategy to depend on time and the initial equity risk premium only. In contrast to Problem~\ref{problem:RatesOnly}, we solve the equity problem in full generality taking account of the mean-reverting risk premia. Mathematically, this makes the present problem much harder. The solution is presented in three steps: an integral equation characterizing the strategy, a general solution, and explicit formulas for the coefficients in the general solution.

Assuming independence between the two risk sources ($\rho=0$), or alternatively no interest rate exposure ($f_s^r=0$), the excess return due to equity exposure is well-defined. The (excess) mean and variance due to equity exposure is given by the terms $m^S_T$ and $\int_0^T h^S_u dW^S_u$, respectively, of Theorem~\ref{thm:VTrep}.

For ease of notation and in this section only, we drop superscript $S$ on all quantities relating to equities. We also introduce a short-hand notation for the expected market price of equity risk at future times
\begin{equation}\label{eq:meanEquityRiskPremium}
   \xi_s = \E\left[\lambda_s \right] = \frac{1}{\sigma_S}\E\left[x_s \right] = \frac{1}{\sigma_S}\left(\bar{x} + e^{-\alpha s}(x_0-\bar{x})\right).
\end{equation}
Using this notation we obtain from Theorem~\ref{thm:VTrep} the following corollary, which serves to formulate the problem.

\begin{corollary} \label{cor:VTequityonly}
Assume $\rho=0$. For given $T > 0$, the {\em excess} return of $V_T/V_0$ due to equity exposure, $f$, is log-normally distributed with mean and variance given by
\begin{align}
   \mu_T      & = \int_0^T \xi_s f_s - \frac{1}{2}f_s^2 ds,  \label{eq:VTmeanequityonly} \\[2mm]
   \sigma_T^2 & = \int_0^T \left[f_u - \frac{\sigma_x}{\sigma_S}\int_u^T f_s e^{-\alpha(s-u)}ds \right]^2 du.  \label{eq:VTvarequityonly}
\end{align}
\end{corollary}

\noindent
The problem of this section can now be phrased as

\begin{problem}  \label{problem:EquitiesOnly}
 Find $f$ ($=f^S$) such that (\ref{eq:VTmeanequityonly}) is maximized for given value of (\ref{eq:VTvarequityonly}).
\end{problem}

Similarly to Problem~\ref{problem:RatesOnly}, for given horizon, $T$, the solution to Problem~\ref{problem:EquitiesOnly} takes the form of a family of strategies, one for each value of the variance $\sigma_T^2$. Members of this family are referred to as {\em optimal equity strategies}. More generally, we refer to equity strategies (locally) maximising, or minimizing, (\ref{eq:VTmeanequityonly}) for given value of (\ref{eq:VTvarequityonly}) as extremal equity strategies.

Before discussing the solution to Problem~\ref{problem:EquitiesOnly} in detail we start by noting that the global (unconstrained) maximum of $\mu_T$ is achieved for $f_s=\xi_s$ yielding $\hat{\mu}_T=\frac{1}{2}\int_0^T\xi_s^2 ds$; the maximum is obtained by optimizing the integrand in (\ref{eq:VTmeanequityonly}) at each $s$ individually. In particular, the maximal excess return is finite and generally also strictly positive (except in the degenerate cases $x_0=\bar{x}=0$, or $x_0=\alpha=0$). Let us denote by $\hat{\sigma}_T^2$ the variance associated with the optimal mean, i.e., (\ref{eq:VTvarequityonly}) evaluated with $f_s=\xi_s$. Since $\hat{\mu}_T$ is the global maximum it is also the constrained maximum under the variance constraint $\sigma_T^2=\hat{\sigma}_T^2$. In particular, $f=\xi$ is an optimal equity strategy.

We typically expect a positive risk-reward trade-off, meaning that the larger the risk (variance) the larger the reward (mean). Hence, we typically expect that as we increase the prescribed value of the variance the associated optimal mean also increases. This however is true only up to a certain point, namely $\hat{\sigma}_T^2$. As we vary the variance from $0$ to $\hat{\sigma}_T^2$ the optimal mean increases from $0$ to $\hat{\mu}_T$, but if we increase the prescribed value of the variance further the optimal mean starts to {\em decrease}. In other words, demanding a variance beyond $\hat{\sigma}_T^2$ is counter-productive in the sense that it can only be achieved by investing so much in equities that it impairs the mean log-return. Mathematically we can characterize the optimal strategies for any prescribed non-negative variance, but in practise we are interested only in optimal strategies with variance at most $\hat{\sigma}_T^2$.

\subsubsection{Characterization of extremal equity strategies}
The double integral in (\ref{eq:VTvarequityonly}) is a manifestation of the global nature of the problem by which later equity exposure reduces the variance arising from earlier exposure. The double integral also implies that Problem~\ref{problem:EquitiesOnly} does not immediately conform with problems solvable by the Euler-Lagrange equation of calculus of variations. However, we can still use a variational argument to arrive at the following integral characterization of the extremal equity strategies.

Extremal strategies are defined as stationary points (strategies) for the constrained optimization problem. These strategies are candidates for the optimal strategies, since optimal strategies are also extremal. However, extremal strategies can also be minimizing strategies, or local extremal strategies. 

\begin{lemma} \label{lemma:OptEquityIntegral}
For given $T > 0$, the extremal equity strategies satisfy
\begin{equation}\label{eq:OptEquityIntegral}
   \xi_s - f_s + 2\nu h_s - 2\nu \frac{\sigma_x}{\sigma_S}\int_0^s h_u e^{-\alpha(s-u)}du=0 \quad (0 \leq s \leq T),
\end{equation}
where
\begin{equation}\label{eq:hEquityDef}
   h_u=f_u - \frac{\sigma_x}{\sigma_S}\int_u^T f_s e^{-\alpha(s-u)}ds  \quad (0 \leq u \leq T),
\end{equation}
and $\nu$ is a Lagrange multiplier.
\end{lemma}
\noindent Note that $h$ defined in Lemma~\ref{lemma:OptEquityIntegral} is the same as $h^S$ of Theorem~\ref{thm:VTrep}.

The extremal strategies are characterized by having the same mean-variance trade-off for all (infinitesimal) perturbations of the strategy, i.e., if we alter an extremal strategy slightly such that the variance changes by $\epsilon$, say, then the mean changes by $\delta$, say, regardless of how we alter the strategy. The Lagrange multiplier, $\nu$, is equal to minus the (common) value of the mean-variance trade-off for the corresponding extremal strategy; loosely speaking, $\nu=-\delta/\epsilon$. In principle, there could be more than one solution to (\ref{eq:OptEquityIntegral}) for given $\nu$, but we will show later that the solution (if it exists) is unique. Hence, the extremal strategies are uniquely identified by their mean-variance trade-off.

Assuming sufficient regularity, differentiating (\ref{eq:OptEquityIntegral}) twice leads to a second-order differential equation in $f$ from which the general form of the solution can be inferred.

\begin{lemma} \label{lemma:OptEquityODE}
For given $T > 0$, the extremal equity strategies satisfy
\begin{equation}\label{eq:OptEquityODE}
    A f''_s + C f_s + D = 0 \quad (0 \leq s \leq T),
\end{equation}
where
\begin{equation}
  A  = 1-2\nu, \quad
  C  = 2\nu\left(\alpha-\frac{\sigma_x}{\sigma_S}\right)^2 -\alpha^2,  \quad
  D  = \alpha^2 \frac{\bar{x}}{\sigma_S}.  \label{eq:ODCcoef}
\end{equation}
\end{lemma}

Lemma~\ref{lemma:OptEquityODE} gives a surprisingly simple characterization of the extremal strategies as solutions to a differential equation with constant coefficients which do not depend on $T$, nor $x_0$. In general, the solution space to (\ref{eq:OptEquityODE}) is two-dimensional, and it remains to identify the specific solution also satisfying (\ref{eq:OptEquityIntegral}).

\subsubsection{Optimal equity strategies with positive mean-variance trade-off}  \label{sec:OptEqStrategyPosTradeOff}
As noted in the introductory remarks to this section, the unconstrained maximum to (\ref{eq:VTmeanequityonly}) is achieved by $f=\xi$ and this strategy is an optimal strategy. We note that $f=\xi$ is also the (unique) solution to (\ref{eq:OptEquityIntegral}) when $\nu=0$; this is to be expected, since a Lagrange multiplier vanishes at a global maximum.

In the limit $\nu$ tending to minus infinity, we get the strategy $f=0$ with $\mu_T=\sigma_T^2=0$. It is not obvious that this is the limiting solution to (\ref{eq:OptEquityIntegral}), but it follows from the explicit solution given in Theorem~\ref{thm:OptimalEquity} below. It also conforms with the intuition that the size of $\nu$ reflects the investor's risk aversion. Although, as we shall see in Section~\ref{sec:numWedge}, this intuition is only partly true.

Strategies satisfying (\ref{eq:OptEquityIntegral}) with $-\infty < \nu < 0$ correspond to extremal strategies with a positive mean-variance trade-off, i.e., strategies where added risk is rewarded in terms of higher mean. Varying $\nu$ in this range gives us all the strategies of practical relevance, ranging from the risk free strategy to the strategy achieving the global maximum. In this section we give an explicit expression for these strategies.\footnote{Strictly speaking, we only know that the strategies are extremal and with rewarded risk. However, we cannot rule out the possibility that for certain values of the variance constraint, there could exist both an extremal strategy with rewarded risk and a better (optimal) extremal strategy with unrewarded risk. This is conceivable because whether or not risk is rewarded is a local property applying only to infinitesimal changes in risk. A formal proof for the non-existence of this possibility would require an analysis of the global structure of the solution space, which is outside the scope of the present paper. In practise, however, it is easily verified that extremal strategies with positive mean-variance trade-off are indeed optimal; this can be observed from a mean-variance plot of the family of all extremal strategies for given model parameters, cf.\ Section~\ref{sec:numeric}}

Hence, assume $\nu < 0$. Under this assumption, $A > 1$ and $C < 0$. Note, that the strict negativity of $C$ holds even if $\alpha=0$ (assuming $\sigma_x>0$ and $\sigma_S>0$). From standard theory we know that the general solution to (\ref{eq:OptEquityODE}) takes the form of a specific solution plus the general solution to the associated, homogeneous differential equation, $A f''_s + C f_s = 0$. The characteristic polynomial of the latter is $Ar^2 + C$, which, under the current assumption, has two distinct, real roots $c_1 = \sqrt{-C/A}$ and $c_2=-\sqrt{-C/A}$, say, and it follows that the general solution to the homogeneous differential equation is of the form $f_s=b_1 e^{c_1 s} + b_2 e^{c_2 s}$. Regarding the specific solution, we note that $f_s=-D/C$ solves (\ref{eq:OptEquityODE}), since $C$ is non-zero. Thus, the general solution to (\ref{eq:OptEquityODE}) is given by
\begin{equation}\label{eq:OptEquityGenSol}
   f_s = b_0 + b_1 e^{c_1 s} + b_2 e^{c_2 s} \quad (0 \leq s \leq T),
\end{equation}
where $b_0=-D/C$. This leaves the determination of the two coefficients $b_1$ and $b_2$. To determine those we insert the general form of $f$ given by (\ref{eq:OptEquityGenSol}) into the left-hand side of (\ref{eq:OptEquityIntegral}). In general, this results in a linear combination of four exponential functions with different exponents, and a constant term. The requirement that all terms vanish gives us five equations for determining the five constants in (\ref{eq:OptEquityGenSol}), in particular $b_1$ and $b_2$. This programme is carried out in Appendix~\ref{app:coefmatch}. For ease of reference, the following theorem summarizes the results (with superscript $S$ reintroduced).

\begin{theorem} \label{thm:OptimalEquity}
Assume independent risk factors ($\rho=0$). For given horizon $T > 0$, the optimal equity strategies with positive mean-variance trade-off are of the form
\begin{equation} \label{eq:OptimalEquity}
     f^S_s = b_0 + b_1 e^{c_1 s} + b_2 e^{c_2 s} \quad (0 \leq s \leq T),
\end{equation}
where $c_1 = \sqrt{-C/A} > 0$, $c_2=-\sqrt{-C/A}$, $b_0=-D/C$ with $A$, $C$, and $D$ given by (\ref{eq:ODCcoef}), and $\nu<0$ depends on the prescribed value of the variance $\sigma_T^2$.

For $\alpha\neq \frac{1}{2}\frac{\sigma_x}{\sigma_S}$, the exponents satisfy $c_1 = |c_2| \neq |\alpha|$, and
\begin{equation}
   \begin{pmatrix}
    b_1 \\
    b_2
  \end{pmatrix}
  =
  \begin{pmatrix}
    \frac{e^{c_1 T}}{c_1-\alpha} & \frac{e^{c_2 T}}{c_2-\alpha} \\[2mm]
    \frac{\sigma_x}{\sigma_S(c_1+\alpha)-\sigma_x} & \frac{\sigma_x}{\sigma_S(c_2+\alpha)-\sigma_x}
  \end{pmatrix}^{-1}
  \begin{pmatrix}
    \frac{\alpha\bar{x}}{\sigma_S\left(\alpha^2 - 2\nu(\alpha-\sigma_x/\sigma_S)^2 \right)} \\[2mm]
    -\frac{x_0}{\sigma_S} + \frac{\alpha\bar{x}}{\sigma_S}\frac{\alpha-2\nu(\alpha-\sigma_x/\sigma_S)}{\alpha^2-2\nu(\alpha-\sigma_x/\sigma_S)^2}
  \end{pmatrix}. \label{eq:OptimalEquityGenB}
\end{equation}

For $\alpha = \frac{1}{2}\frac{\sigma_x}{\sigma_S}$, the exponents are $c_1=\alpha$ and $c_2=-\alpha$, with
\begin{equation}
  b_1=0, \quad b_2= \frac{-\frac{x_0}{\sigma_S}+ \frac{\bar{x}}{\sigma_S}\left(1-\frac{4\nu}{1-2\nu}\left(e^{-\alpha T}-1\right)\right)}{2\nu e^{-2\alpha T}-1}.
  \label{eq:OptimalEquitySpecialB}
\end{equation}
\end{theorem}

Note that the equity risk premium process is transient for $\alpha=0$, and exploding for $\alpha<0$, so for practical purposes we would consider using the model only for $\alpha>0$. Nevertheless, it is interesting to learn that the solution has the same structure regardless of the value of $\alpha$. Recall that the global maximum is achieved by $f=\xi$ with $\xi$ given by (\ref{eq:meanEquityRiskPremium}). This strategy is also of form (\ref{eq:OptimalEquity}) with either $c_1$ or $c_2$ equal to $\alpha$, and $b_0=\bar{x}/\sigma_S$.

For $\nu$ tending to minus infinity, the exponents approach $\pm (\alpha - \sigma_x/\sigma_S)$. In particular, the exponents have a finite limit. Further, for $\alpha \neq \sigma_x/\sigma_S$ the $b$-coefficients all tend to $0$ for $\nu$ tending to minus infinity, while for $\alpha=\sigma_x/\sigma_S$ the sum of the $b$-coefficients tends to $0$ in the same limit.\footnote{We here give an outline of the argument. Write (\ref{eq:OptimalEquityGenB}) as $F(b_1\ b_2)'=B$. Using the established limits of $c_1$ and $c_2$, we have that for $\alpha \neq \sigma_x/\sigma_S$, one of the bottom entries of $F$ diverges while the other three entries have finite, non-zero limits. Under the same condition on $\alpha$, the top entry of $B$ tends to $0$, while the bottom entry has a finite limit. From this we can conclude that $b_1$ and $b_2$ both tend to $0$.
For the special case, $\alpha=\sigma_x/\sigma_S$, both $c_1$ and $c_2$ tend to $0$ for $\nu$ tending to minus infinity. It follows that the top row of $F$ converges to $(-1/\alpha\ -1/\alpha) = (-\sigma_S/\sigma_x\ -\sigma_S/\sigma_x)$, while the top entry of $B$ equals $\bar{x}/\sigma_x (= b_0 \sigma_S/\sigma_x)$. From this we can conclude that $b_0+b_1+b_2$ tends to $0$. We leave it to the reader to handle the last special case, $\alpha = \frac{1}{2}\frac{\sigma_x}{\sigma_S}$.} Thus, in either case the limiting strategy is the risk-free strategy, $f=0$, as previously claimed.

The optimal strategy has a simple structure, but apart from the observations made already, the expressions in Theorem~\ref{thm:OptimalEquity} appear too complex for intuitive interpretation. Instead we will explore the optimal mean-variance trade-off and the optimal strategies numerically in Section~\ref{sec:NumOptimalEquity}.

\subsubsection{Extremal equity strategies with negative mean-variance trade-off} \label{sec:OptEqStrategyNegTradeOff}
From a practical point of view, we are only interested in pursuing optimal strategies with a positive mean-variance trade-off, and these are therefore the main focus of the paper. However, it is very instructive to also study extremal strategies with a negative mean-variance trade-off, corresponding to $\nu>0$. These strategies include both optimal strategies with "excessive" variance, minimizing strategies with lowest possible mean  for given variance, but also "interior" strategies, where the mean is strictly between the minimal and maximal value for given variance. The latter class of strategies displays an intriguing variety of equity profiles, exemplified in Section~\ref{sec:numWedge}. In this section we briefly discuss the different types of solution that can occur for $\nu>0$. The actual strategies can be found in Appendix~\ref{app:overviewExtremal}, with Theorem~\ref{thm:ExtremalOverview} giving an overview of the full set of extremal strategies.

Depending on the type of roots to the characteristic polynomial, $Ar^2+C$, the extremal strategy takes one of three forms. For distinct, real roots we get the exponential solution already covered, for complex roots we get a trigonometric solution of form
\begin{equation}
   f_s = b_0 + b_1\sin(cs) + b_2\cos(cs),
\end{equation}
where $b_0=-D/C$ and $c=\sqrt{C/A}$, and in the special case where zero is a double root we get a quadratic solution
\begin{equation}
   f_s = b_0 + b_1 s + b_2 s^2,
\end{equation}
where $b_2=-D/(2A)$.

The solution is exponential if $A$ and $C$ are of opposite signs, it is trigonometric if $A$ and $C$ are of the same sign, and it is quadratic in the special case $C=0$ and $A\neq 0$. In general, there is no extremal equity strategy for $\nu=1/2$ ($A=0$). The domains for the respective solution types are illustrated in Figure~\ref{fig:domain} of Appendix~\ref{app:overviewExtremal}. We know that the exponential solution applies for $\nu<0$, but we see from the figure that it is in fact also the "typical" solution for $\nu \geq 0$.

The quadratic solution defines the border between the exponential and trigonometric solutions. For given model parameters, there is (at most) one value of $\nu$ for which the extremal strategy is quadratic. Of course, one is unlikely to encounter this solution in applications, but it is interesting to note that quadratic solutions bridge the two main domains. This suggests that extremal strategies in general might be well approximated by quadratic strategies.

Finally, we note that the solution type depends only on $\nu$ and the mean reversion ratio $\tilde{\alpha} = \alpha/[\sigma_x/\sigma_S]$, where $\sigma_x/\sigma_S$ can be interpreted as the volatility of the market price of equity risk, $\lambda^S_t=x_t/\sigma_S$. Thus, apart from $\nu$, the determining factor for the type of solution is the ratio of mean reversion strength to noise in the market price of equity risk process,
\begin{equation}
    d\lambda^S_t = \alpha\left(\bar{\lambda}^S - \lambda^S_t\right)dt - \frac{\sigma_x}{\sigma_S}dW^S_t,
\end{equation}
where $\bar{\lambda}^S=\bar{x}/\sigma_S$ denotes the long-run market price of equity risk.

\subsection{Optimal joint strategies}
Assuming independent risk factors ($\rho=0$), the effects of rate and equity exposures can be separated and joint maximization can be performed based on the results of the preceding two sections.

To phrase the problem, let us first state a simplified version of Theorem~\ref{thm:VTrep}.

\begin{corollary} \label{cor:VTindependent}
Assume $\rho=0$ and $a=\kappa$. For given $T > 0$, $V_T/V_0$ is log--normally distributed with mean and variance given by
\begin{align}
   \mu_T      & = m^0_T +  \int_0^T \lambda^r f^r_s - \frac{1}{2} \left(f^r_s\right)^2ds + \int_0^T \xi_s f^S_s - \frac{1}{2}\left(f^S_s\right)^2 ds,  \label{eq:VTmeanIndependent} \\[2mm]
   \sigma_T^2 & = \int_0^T \left[\sigma_r\Psi(\kappa,T-u) + f^r_u  \right]^2 du + \int_0^T \left[f^S_u - \frac{\sigma_x}{\sigma_S}\int_u^T f^S_s e^{-\alpha(s-u)}ds \right]^2 du.  \label{eq:VTvarIndependent}
\end{align}
where $m^0_T$ is given by (\ref{eq:m0Tdef}) of Theorem~\ref{thm:VTrep}, $\xi_s$ is given by (\ref{eq:meanEquityRiskPremium}), and $\lambda^r = \kappa(\bar{r}-b)/\sigma_r$.
\end{corollary}

The problem of this section can now be phrased as

\begin{problem} \label{problem:IndependentRiskFactors}
 Find $(f^r, f^S)$ such that (\ref{eq:VTmeanIndependent}) is maximized for given value of (\ref{eq:VTvarIndependent}).
\end{problem}

For given horizon, $T$, the solution to Problem~\ref{problem:IndependentRiskFactors} takes the form of a family of pairs of rate and equity strategies, one for each value of the variance $\sigma_T^2$. Members of this family are referred to as {\em optimal pairs of strategies}. More generally, we refer to pairs of strategies (locally) maximising, or minimizing, (\ref{eq:VTmeanIndependent}) for given value of (\ref{eq:VTvarIndependent}) as extremal pairs of strategies.

We note that the mean and variance of Corollary~\ref{cor:VTindependent} are the sum of the means and variances, respectively, of Corollaries~\ref{cor:VTratesonly} and \ref{cor:VTequityonly}. Further, the interest rate strategy, $f^r$, and the equity strategy, $f^S$, affect separate terms. This implies that the solutions to Problem~\ref{problem:IndependentRiskFactors} consist of the previously derived, optimal strategies to Problems~\ref{problem:RatesOnly} and \ref{problem:EquitiesOnly}. A fortiori, the extremal pairs consist of extremal bond strategies and extremal equity strategies with the same Lagrange multiplier, $\nu$. Heuristically, the argument is as follows.

Assume $(f^r, f^S)$ is an extremal pair of strategies. This implies that the mean-variance trade-off is the same for all (infinitesimal) perturbations of the strategy. Since if this were not the case, it would be possible to combine two changes with different trade-offs to change the mean, while preserving the variance, contradicting the assumed extremity. In particular, the mean-variance trade-off is the same for all (infinitesimal) perturbations of only the rate strategy, or only the equity strategy. But since the mean-variance trade-off associated with $f^r$ is the same in Problem~\ref{problem:IndependentRiskFactors} as it is in Problem~\ref{problem:RatesOnly}, we conclude that $f^r$ is in fact an extremal bond strategy. Similarly, $f^S$ is in fact an extremal equity strategy. Moreover, the mean-variance trade-off of $f^r$ and $f^S$ must match, and since the trade-off equals (minus) the Lagrange multiplier, $\nu$, we conclude that the extremal strategies $f^r$ and $f^S$ must have the same value of $\nu$. The argument can be made formal, but it is useful to keep the heuristic argument in mind.

\begin{theorem} \label{thm:OptimalJoint}
Assume independent risk factors ($\rho=0$), and constant market price of interest rate risk ($a=\kappa$). For given horizon $T > 0$, the extremal pairs of strategies consist of pairs with $f^r$ given by Theorem~\ref{thm:OptimalBond}, and $f^S$ given by Theorem~\ref{thm:OptimalEquity} (Theorem~\ref{thm:ExtremalOverview}) for the same value of $\nu\neq 1/2$.

For $\nu<0$, the pairs are optimal strategies with positive mean-variance trade-off.
\end{theorem}

We note from Theorem~\ref{thm:OptimalJoint} that all extremal rate and equity strategies belong to one, and only one, extremal pair. For the rate strategies, we can compute a priori the variance contribution by formula (\ref{eq:sigmaToptimalRate}), but we have no similar (simple) formula for the variance contribution of the equity strategies. Thus, we cannot a priori say anything about the relative size of the variance contributions from rates and equities in the extremal pairs. In particular, we do not know in general whether the optimal pairs are balanced regarding rate and equity risk, or whether one of the risk sources dominates. Of course, we can answer this question numerically on a case by case basis by computing the family of extremal pairs and the variance contribution from each risk source.

In addition to the optimal pairs of strategies, it is also of interest to consider strategies with the same amount of rate and equity risk. Such risk parity, or balanced, strategies might be considered more robust as they do not rely on assumed differences in risk premia. In the present context, one could consider pairs of strategies with the same variance contribution on the horizon, i.e., pair up strategies solving Problems~\ref{problem:RatesOnly} and \ref{problem:EquitiesOnly} for the same value of $\sigma_T^2$ given by, respectively, (\ref{eq:VTvarratesonly}) and (\ref{eq:VTvarequityonly}). Again, this pairing has to be done numerically on a case by case basis. We do not pursue these ideas further in this paper.

\subsubsection{Dependent risk factors}
The case of dependent risk factors ($\rho\neq 0$) is mathematically rather more challenging. In full generality, the problem amounts to maximising (\ref{eq:VTmean}) for given value of (\ref{eq:VTvar}). We see from Theorem~\ref{thm:VTrep}, that for $\rho\neq 0$ the mean and variance of $\log(V_T/V_0)$ are affected by terms depending on both $f^r$ and $f^S$. Thus we cannot hope to solve the general problem by "stitching" together partial solutions. In fact, it is not even clear how to separate the mean and variance contributions arising from, respectively, the rate and equity exposures.

Apart from the mathematical difficulties, we also argue that the resulting strategies---if we were able to obtain them---are unlikely to be of much value in practice. An assumed correlation between rate and equity risk factors is very hard to verify in practice, and we therefore might be reluctant to pursue strategies which exploit this correlation. From that perspective, the independence assumption is the "neutral" assumption most often used in practice.

Rather than deriving the general solution, we can alternatively test the robustness of strategies derived under no correlation in an environment with correlation. We can, e.g., establish the correlation range for which the derived strategies are better than constant strategies. This in turn can be used as a guide for when the derived strategies are near-optimal without the need for an exact value of the correlation coefficient. We do not, however, pursue this idea further.

As an indication of the complexity of the optimal strategies under dependent risk factors we derive the optimal rate strategy for given equity exposure $f^S$. Assuming $a=\kappa$, this amounts to maximizing the mean
\begin{equation}
    \mu_T   = m^0_T + \int_0^T f^r_s \left[\lambda^r - \rho f_s^S \right] ds - \frac{1}{2}\int_0^T \left(f^r_s\right)^2ds ,
\end{equation}
for given value of the variance
\begin{equation}
     \sigma_T^2 = \int_0^T \left[\sigma_r\Psi(\kappa,T-s) + f^r_s\right]\left[\sigma_r\Psi(\kappa,T-s) + f^r_s +2\rho h_s^S\right] ds.
\end{equation}
These expressions follow from Theorem~\ref{thm:VTrep}, see also Corollary~\ref{cor:VTratesonly} and Problem~\ref{problem:RatesOnly} for the original optimization problem. By a straightforward extension of the proof of Theorem~\ref{thm:OptimalBond} we find that the optimal rate strategies on horizon $T$ are of the form
\begin{equation}
     f_s^r = \frac{\lambda^r + 2\nu g_s^r + \rho(2\nu h_s^S- f_s^S)}{1-2\nu} \quad(0\leq s\leq T),
\end{equation}
with $g_s^r=\sigma_r\Psi(\kappa,T-s)$ and $\nu<1/2$. We see that, although computable, the strategy is complicated and with explicit reference to the given equity exposure via both $f^S$ and $h^S$. This indicates that the optimal pais of strategies are presumably very complicated indeed, and we will not try to find them.

\pagebreak
\section{Numerical illustrations} \label{sec:numeric}
In this section we provide numerical illustrations of the optimal risk-reward profiles and the underlying optimal strategies. We cover optimal rate and equity strategies separately, with an emphasis on the latter. Assuming independent risk factors, the jointly optimal pais of strategies consist of optimal rate and equity strategies with the same value of $\nu$, cf.\ Theorem~\ref{thm:OptimalJoint}. Thus, joint optimization consists of pairing the illustrated optimal rate and equity strategies. We do not, however, explicitly consider joint optimization in this section.

\subsection{Optimal rate strategies} \label{sec:NumOptimalRate}
In the following we consider the optimal strategies of Section~\ref{sec:OptimalBond}, i.e., optimal strategies when we are allowed to invest only in the interest rate market. For illustrative purposes we consider two parameter sets, corresponding to moderate and low market prices of interest rate risk, cf.\ Table~\ref{tab:ratePar}. For these parameter sets we show the risk-reward profile for the entire family of extremal strategies and we give examples of optimal rate strategies and the resulting portfolio distributions. We also consider alternative risk and reward statistics connected to the portfolio distributions.

\subsubsection{Yield curves}
First, we visualize the interest rate assumptions. The continuously compounded zero-coupon yield for the period $[t,T]$, $r_t(T)$, is defined by the relation $p_t(T)=\exp\{-(T-t)r_t(T)\}$. From (\ref{eq:ZCBpriceAltExp}) it follows
\begin{equation} \label{eq:yieldDef}
   r_t(T) = b + \frac{\Psi(a,\Delta)}{\Delta}(r_t-b) - \frac{\sigma_r^2}{2}\frac{\Upsilon(a,\Delta)}{\Delta},
\end{equation}
where $\Delta=T-t$, and $\Psi(a,0)/0 = 1$ and $\Upsilon(a,0)/0=0$ are defined by continuity. The yield as a function of $T$ is referred to as the yield curve. The yield curve and the corresponding curve of zero-coupon bond prices are equivalent ways of representing the bond market, but we typically prefer the former due to its more intuitive interpretation.

\begin{table}[ht]
  \centering
\rowcolors{2}{gray!25}{white}
\bgroup
\def\arraystretch{1.3}
\begin{tabular}{C{2.6cm}C{1.2cm}C{1.2cm}C{1.2cm}C{1.2cm}C{1.2cm}C{1.2cm}} \hline
  \rowcolor{gray!50}
    Parameter set  &   $\kappa$   & $\bar{r}$  &  $\sigma_r$  & $a$  &  $b$  & $-\lambda^r$ \\ \hline
    Moderate       &    0.08      &   0.02     &   0.007      & 0.08 & 0.04  & 0.2286      \\
    Low            &    0.08      &   0.02     &   0.007      & 0.08 & 0.03  & 0.1143      \\  \hline
\end{tabular}
\egroup
  \caption{Parameters governing the short-rate $P$-dynamics (\ref{eq:shortrate}), $\kappa$, $\bar{r}$, and $\sigma_r$, and pricing parameters, $a$ and $b$.
  Since $a=\kappa$, the market price of interest rate risk is constant and equal to $\lambda^r = \kappa(\bar{r}-b)/\sigma_r$. The last column shows $-\lambda^r$ which is the continuous-time
  analogue to the Sharpe ratio, i.e., excess return per unit of volatility for holding bonds.}
  \label{tab:ratePar}
\end{table}

Figure~\ref{fig:yieldCurves} shows yield curves corresponding to different values of the short rate, $r_0$, for the two parameter sets in Table~\ref{tab:ratePar}. The yield curves represent the yield that can be locked in today (time $0$) by purchasing a zero-coupon bond. In the left plot interest rate risk is rewarded higher than in the right plot, and consequently the yields are higher. For example, if $r_0=0\%$ (green curves) we can lock in a return of $r_0(20)=1.89\%$ per year on a 20-year horizon when the market price of risk is moderate (left plot), and a return of $r_0(20)=1.39\%$ when the market price of risk is low (right plot). Note that since the short-rate $P$-dynamics are the same in the two cases, a money-market account will give rise to the {\em exact same} return distribution on any horizon in the two cases, while bond strategies will yield higher returns in the "moderate" parameter set than in the "low" parameter set.

\begin{figure}[ht]
\begin{center}
\includegraphics[height=7.5cm]{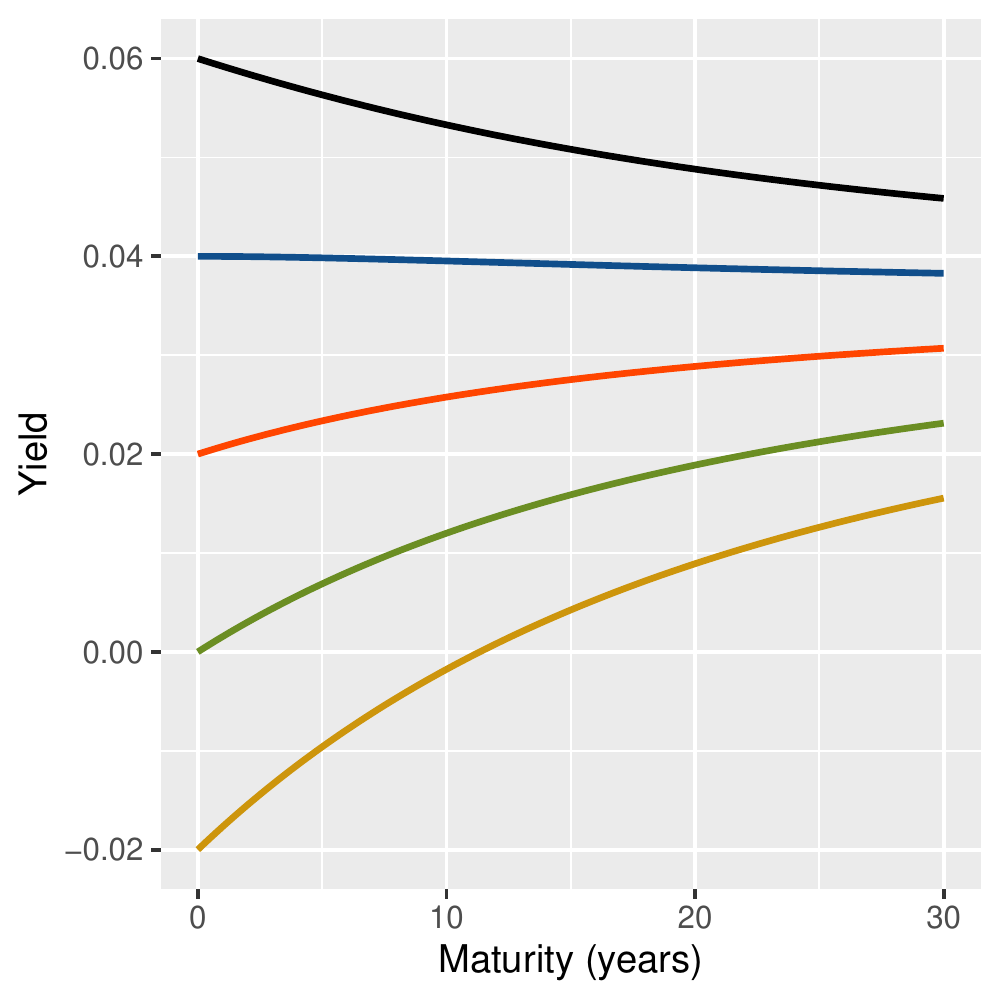}   
\hfill
\includegraphics[height=7.5cm]{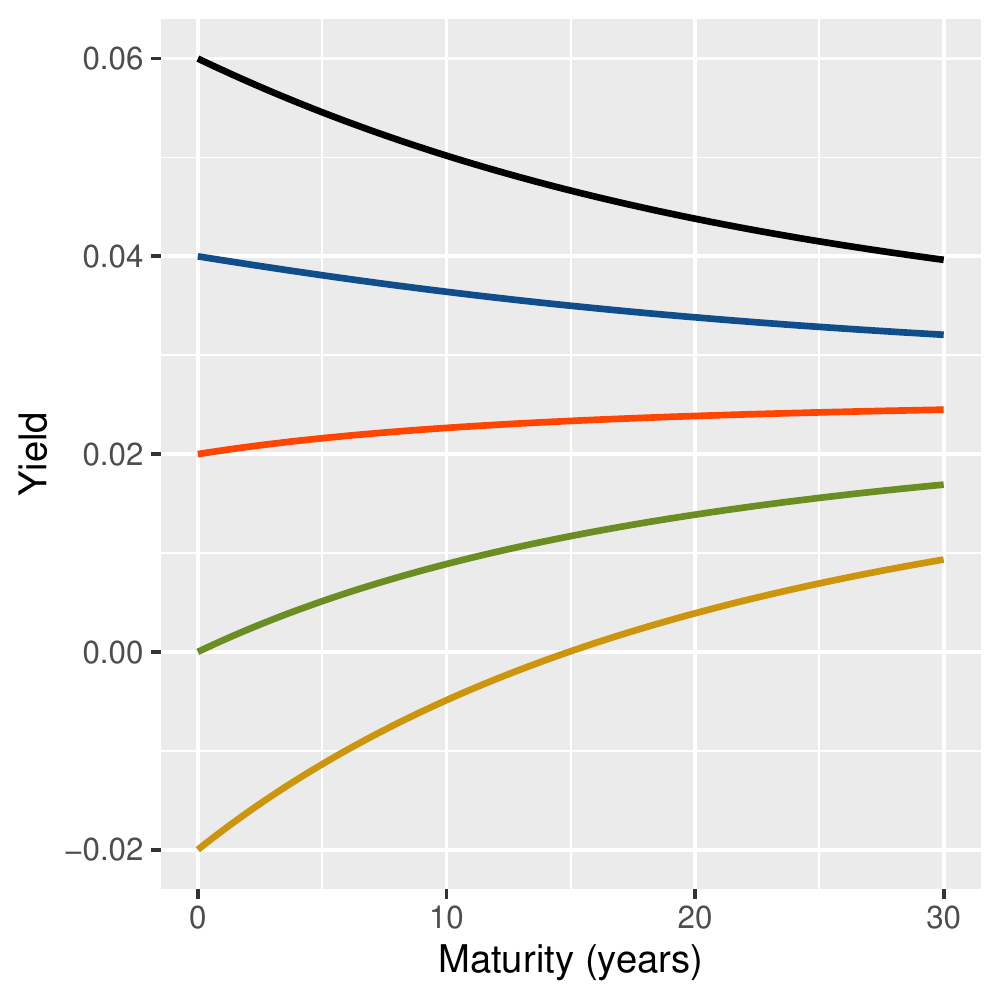}   
\end{center}
\vspace*{-5mm}
\caption{Illustration of yield curves with moderate (left) and low (right) market price of interest rate risk, cf.\ Table~\ref{tab:ratePar}. For each parameter set, the yield curve for five different values
of the short-rate is shown, $r_0=-2\%,\ 0\%,\ 2\%,\ 4\%,\ 6\%$.}
\label{fig:yieldCurves}
\end{figure}

\subsubsection{Risk-reward profiles} \label{sec:MVprofileOptimalRate}
Section~\ref{sec:closedFormMeanVar} provides formulas for the extremal mean and variance, i.e., the log-portfolio mean, $\mu_T$, and the log-portfolio variance, $\sigma_T^2$, arising from following the extremal rate strategies of Theorem~\ref{thm:OptimalBond}. Figure~\ref{fig:rateMVprofile} plots $\mu_T$ against $\sigma_T$ for $T=20$ for four different values of the initial short-rate, $r_0$, using the same colour scheme as in Figure~\ref{tab:ratePar}. We refer to the curves as risk-reward profiles. The upper part of the profiles represent maximizing strategies ($\nu<1/2$), and the lower part of the profiles represent minimizing strategies ($\nu>1/2$), where $\nu$ is the Lagrange multiplier of Theorem~\ref{thm:OptimalBond}. The upper part of the profiles consists of an initial upward-sloping part ($\nu<0$) where adding risk increases the mean until the global maximum is achieved ($\nu=0$), and a downward-sloping part ($0<\nu<1/2$) where risk is so high that adding further risk decreases the mean. Of course, for practical purposes only strategies with $\nu\leq 0$ are of interest.

It follows from Corollary~\ref{cor:VTratesonly} and expressions (\ref{eq:muToptimalRateI})---(\ref{eq:sigmaToptimalRate}) that,
\begin{equation}  \label{eq:VTrateFactorStructure}
   V_T \stackrel{\cal{D}}{=} V_0 \exp\left(\mu_T+\sigma_T U \right) = \frac{V_0}{p_0(T)}\underbrace{\exp\left(\mu_T + \log(p_0(T)) +\sigma_T U\right)}_{\mbox{independent of $r_0$}} = \frac{V_0}{p_0(T)}Y_T,
\end{equation}
where $U$ is a standard normal variate. Thus, the portfolio value on the horizon can be interpreted as the amount that can be locked in with certainty at time $0$ times a stochastic factor, $Y_T$. The stochastic factor depends on the $P$-dynamics of the short-rate process, the market price of interest rate risk,  and $\nu$, but not on the initial value of the short-rate, $r_0$. From this observation and (\ref{eq:VTrateFactorStructure}) it follows that the profiles in Figure~\ref{fig:rateMVprofile} are (vertical) translations of each other, with the (vertical) distance between two profiles being equal to the difference in log bond prices.

\begin{figure}[ht]
\begin{center}
\includegraphics[height=7.5cm]{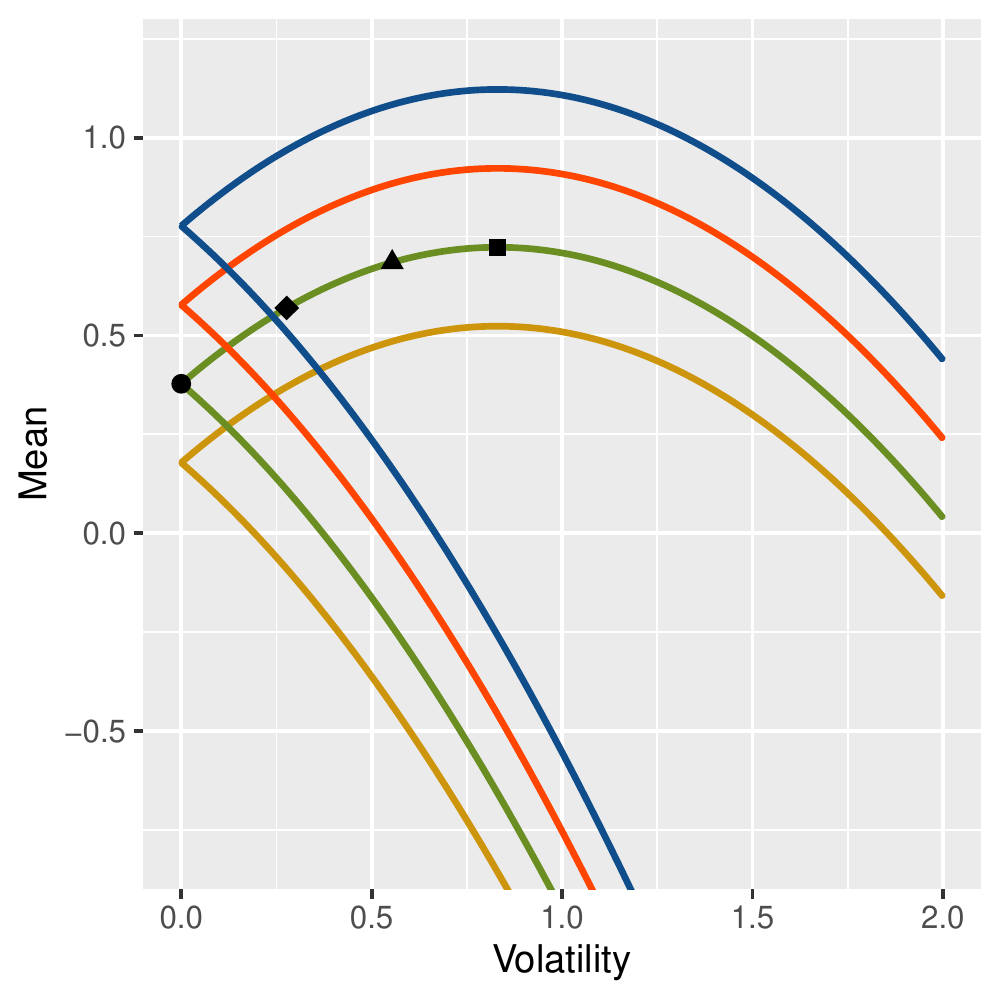}   
\hfill
\includegraphics[height=7.5cm]{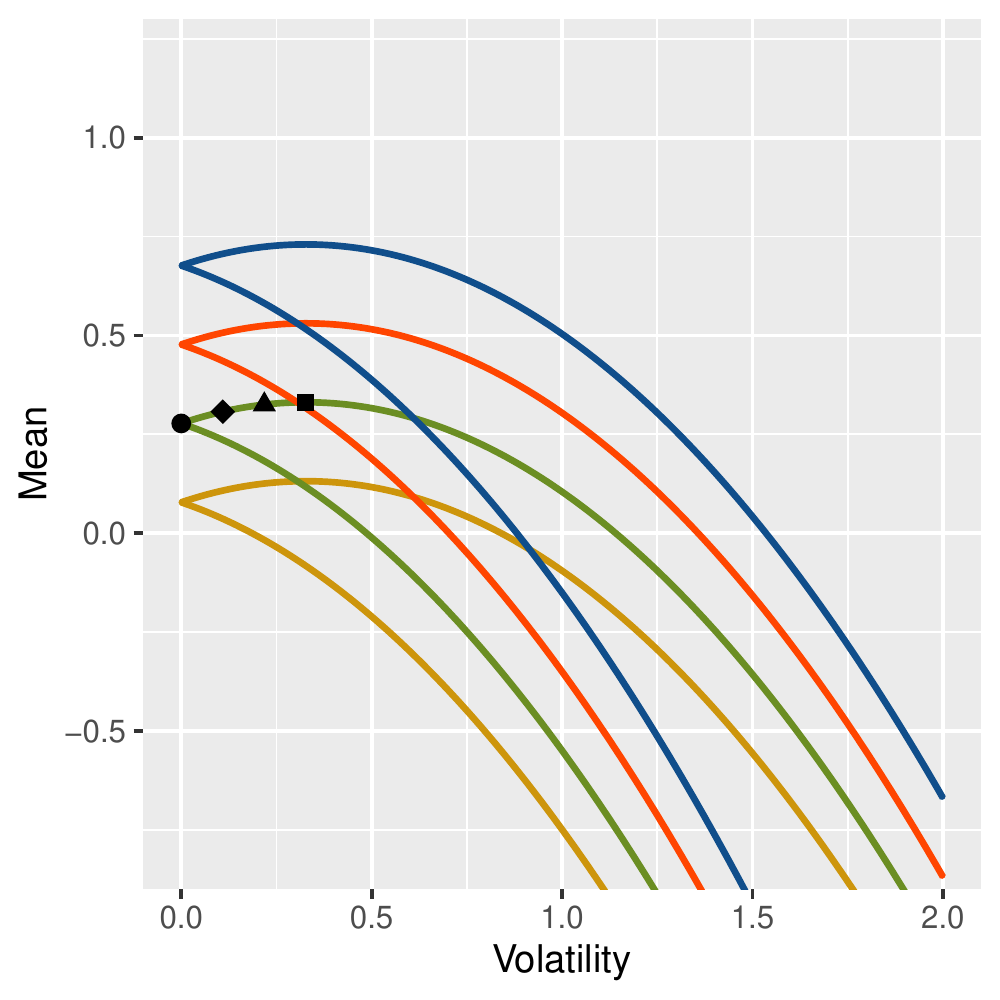}   
\end{center}
\vspace*{-5mm}
\caption{Risk-reward profiles, $(\sigma_T,\mu_T)$, for extremal rate strategies on a horizon of $T=20$ years for moderate (left) and low (right) market price of interest rate risk, cf.\ Table~\ref{tab:ratePar}. The different profiles correspond to $r_0$ of $-2\%$ (yellow), $0\%$ (green), $2\%$ (orange), and $4\%$ (blue). For $r_0=0\%$, four risk-reward pairs are marked, corresponding to $\nu=-\infty$ ({\large $\bullet$}), $\nu=-1$ ($\blacklozenge$), $\nu=-1/4$ ($\blacktriangle$), and $\nu=0$ ({\scriptsize $\blacksquare$}); the underlying portfolio distributions and strategies are shown in Figures~\ref{fig:rateDistrib} and \ref{fig:rateStrategy}, respectively.}
\label{fig:rateMVprofile}
\end{figure}

\subsubsection{Portfolio distributions}
We have marked four points on each of the profiles in Figure~\ref{fig:rateMVprofile} for $r_0=0\%$. Figure~\ref{fig:rateDistrib} illustrates the underlying distribution of $V_T$ (for $V_0=1$); alternatively, we can interpret the plots as showing $Y_T/p_0(T)$, i.e., the stochastic multiplier scaled by the 20-year bond return that can be obtained when $r_0=0\%$. The thin black lines show the median, $\exp(\mu_T)$, of each distribution.
As $\nu$ is increased from $-\infty$ to $0$ the median increases, but the volatility also increases markedly. In particular, when the market price of risk is low (right plot) the increase in median is very modest compared to the additional risk. It follows from (\ref{eq:VTrateFactorStructure}) that, apart from scaling, these plots will look the same for other values of $r_0$.

In Appendix~\ref{app:statistics}, Tables~\ref{tab:rateStatisticsMod} and \ref{tab:rateStatisticsLow} quantify various risk and reward statistics for $Y_T$ for a range of optimal rate strategies. The tables show the median as a measure of the reward, along with three different risk measures:
\begin{itemize}
  \item The probability that we earn less than the risk-free investment, $\gP(Y_T<1)$;
  \item The expected size of the loss given that we experience a loss, $\E[1-Y_T|Y_T<1]$;
  \item The expected size of the loss, $\E[1-Y_T]^+ = \E[1-Y_T|Y_T<1]\gP(Y_T<1)$.
\end{itemize}

Since $Y_T=1$ can be obtained with no risk by holding only $T$-bonds, we interpret $Y_T<1$ as a loss relative to this risk-free strategy. The risk measures therefore focus on the event $Y_Y<1$. We see that for given $\nu$, the probability of earning less than the risk-free investment is substantially higher when the market price of interest rate risk is low (Table~\ref{tab:rateStatisticsLow}) than when it is moderate (Table~\ref{tab:rateStatisticsMod}), but both the conditional and unconditional losses are smaller. This is due to the fact that the exposure depends on the market price of risk, see Figure~\ref{fig:rateStrategy}. Thus, when risk is less rewarded the exposure is decreased and therefore as the distribution of $Y_T$ moves "to the right" it also becomes more narrow.

\begin{figure}[ht]
\begin{center}
\includegraphics[height=7.5cm]{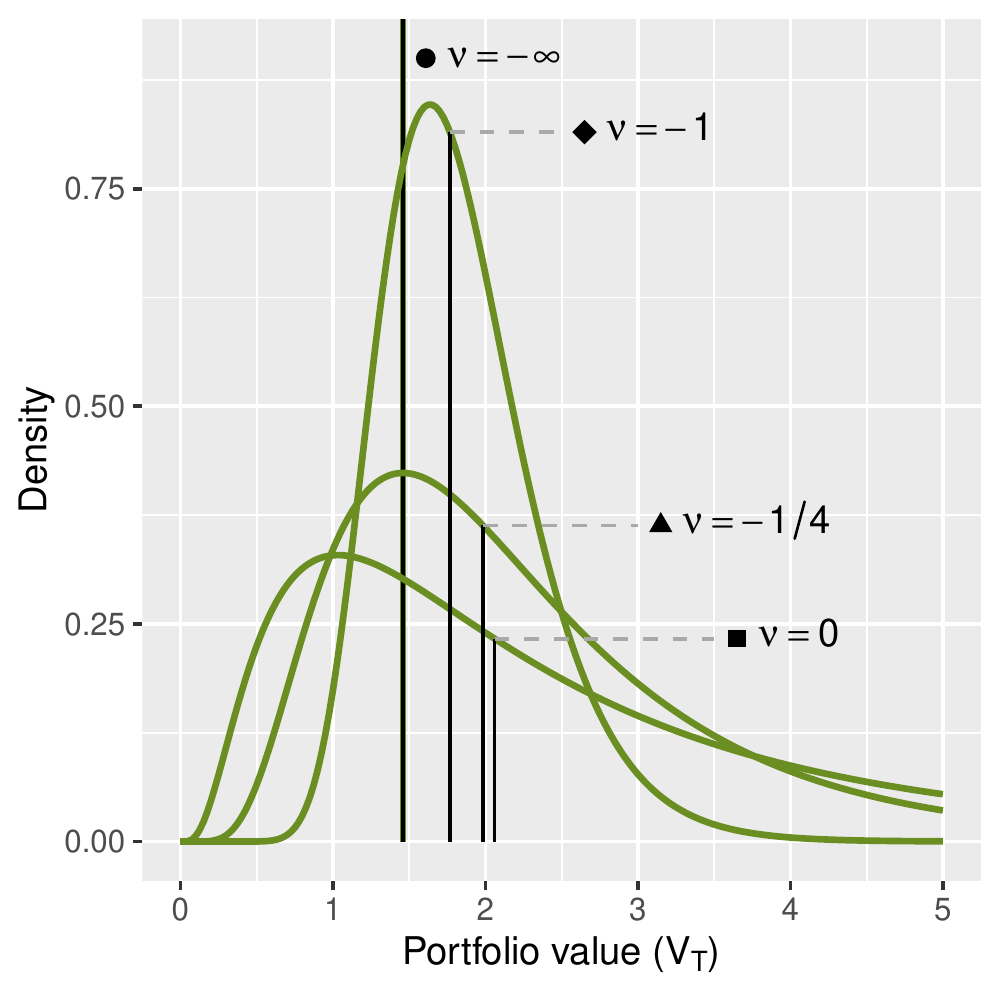}   
\hfill
\includegraphics[height=7.5cm]{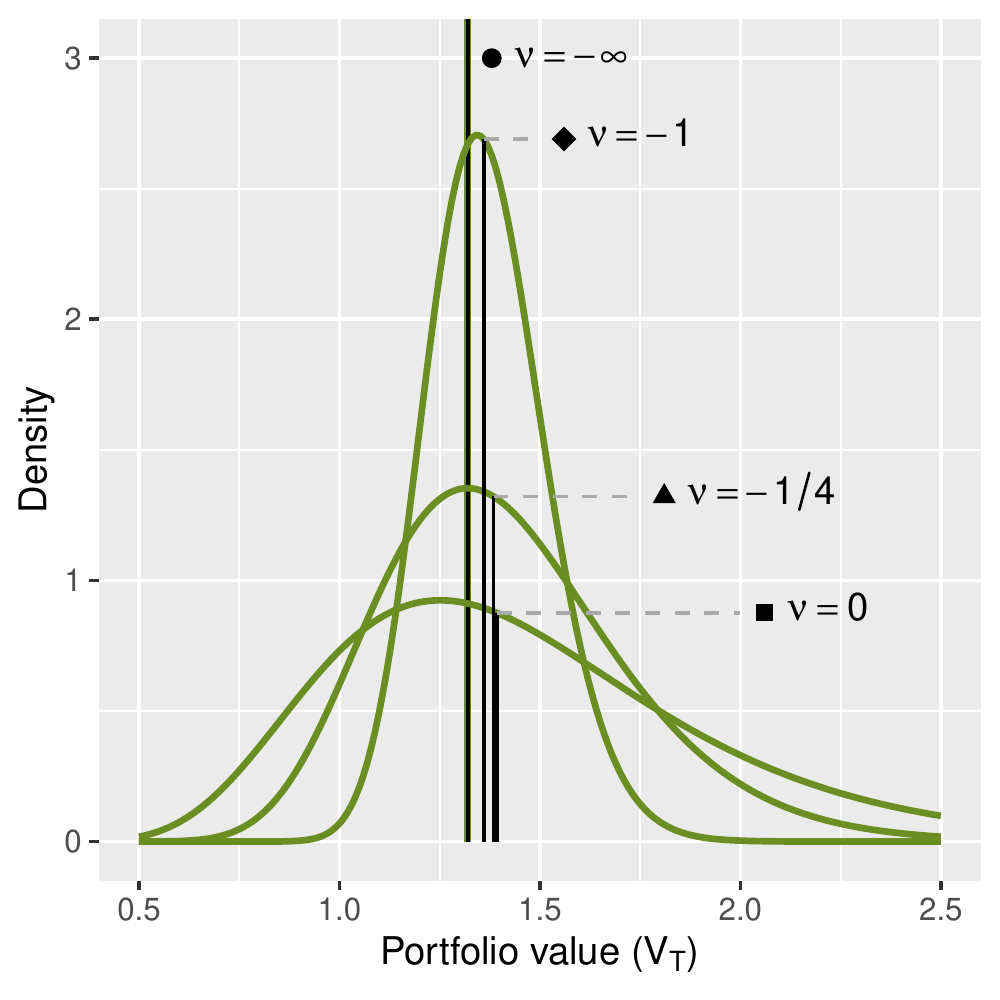}   
\end{center}
\vspace*{-5mm}
\caption{The distribution of $V_T$ for $T=20$ years for four optimal rate strategies when $V_0=1$, $r_0=0\%$ and the market price of interest rate risk is moderate (left) and low (right), cf.\ Table~\ref{tab:ratePar}. For $\nu=-\infty$, $V_T=1/p_0(T)$ is constant which is marked by a vertical green line. The thin black lines mark the median of each distribution. The underlying strategies are shown in Figure~\ref{fig:rateStrategy}.}
\label{fig:rateDistrib}
\end{figure}

\subsubsection{Illustration of optimal strategies}
Figure~\ref{fig:rateStrategy} shows the underlying strategies for the portfolio distributions of Figure~\ref{fig:rateDistrib}. The strategies are given by (\ref{eq:OptimalBond}) of Theorem~\ref{thm:OptimalBond}. The strategies give rise to the stochastic multiplier, $Y_T$, which acts on the risk-free bond return that can be locked in initially. The bond return depends on the initial yield curve which in turn depends on the initial short-rate, $r_0$. The strategies themselves, however, do not depend on $r_0$. In particular, Figure~\ref{fig:rateStrategy} is valid for all interest rate levels. In other words, an investor with a {\em relative} risk tolerance will choose the same strategy regardless of the interest rate level and thereby obtain the same (stochastic) multiplier on the risk-free investment. Conversely, an investor with an {\em absolute} risk tolerance will choose a strategy that depends on the interest rate level, but still within the family of optimal strategies.

In Figure~\ref{fig:rateStrategy} we plot the negated interest rate exposure ($-f_s^r$) as a function time, rather than the interest rate exposure itself. The optimal exposure is negative (corresponding to holding bonds) and the negated exposure therefore gives a more intuitive plot with higher exposures implying larger bond holdings. Also, for the optimal strategies the negated exposure is equal to the (instantaneous) volatility of the portfolio.\footnote{It follows from (\ref{eq:Vdyn}) that, in general, the (instantaneous) volatility of the portfolio at time $s$ is given by $\sqrt{[f_s^r]^2+[f_s^S]^2+2\rho f_s^r f_s^S}$. With no equity risk, $f_s^S=0$, this reduces to $|f_s^r|$, which equals $-f_s^r$ when $f_s^r$ is negative.}

Recall from Section~\ref{sec:OptimalBond} that $\nu=-\infty$ corresponds to the risk-free, buy-and-hold strategy where $V$ is fully invested in $T$-bonds, while $\nu=0$ corresponds to a constant risk exposure of $\lambda^r$ which is the strategy achieving the (unconstrained) maximal mean. For $-\infty < \nu <0$ we get a convex combination of these two strategies. Holding only $T$-bonds corresponds to a (negated) exposure of $-f_s^r=\sigma_r\Psi(\kappa,T-s)$ with an initial volatility of $0.007\Psi(0.08, 20)=7\%$, decreasing to zero over time. This (hedging) strategy does not depend on the market price of interest rate risk, and it is therefore the same for the two parameter sets. All other optimal strategies, however, depend on the market price of interest rate risk: The higher the market price of risk, the higher the optimal exposure (for given $\nu$).

Finally, Figure~\ref{fig:rateStrategy} also shows an example of an optimal strategy on the "wrong" side of the risk-reward profile ($\nu=1/8$) where risk is so excessive that the quadratic term in (\ref{eq:VTmeanratesonly}) impairs the mean. The exposure is still negative, however, i.e., we are long in bonds. In contrast, for mean minimizing strategies the exposure is fully or partly positive, corresponding to shorting bonds; if drawn these strategies would lie below the risk-free strategy.

\begin{figure}[ht]
\begin{center}
\includegraphics[height=7.5cm]{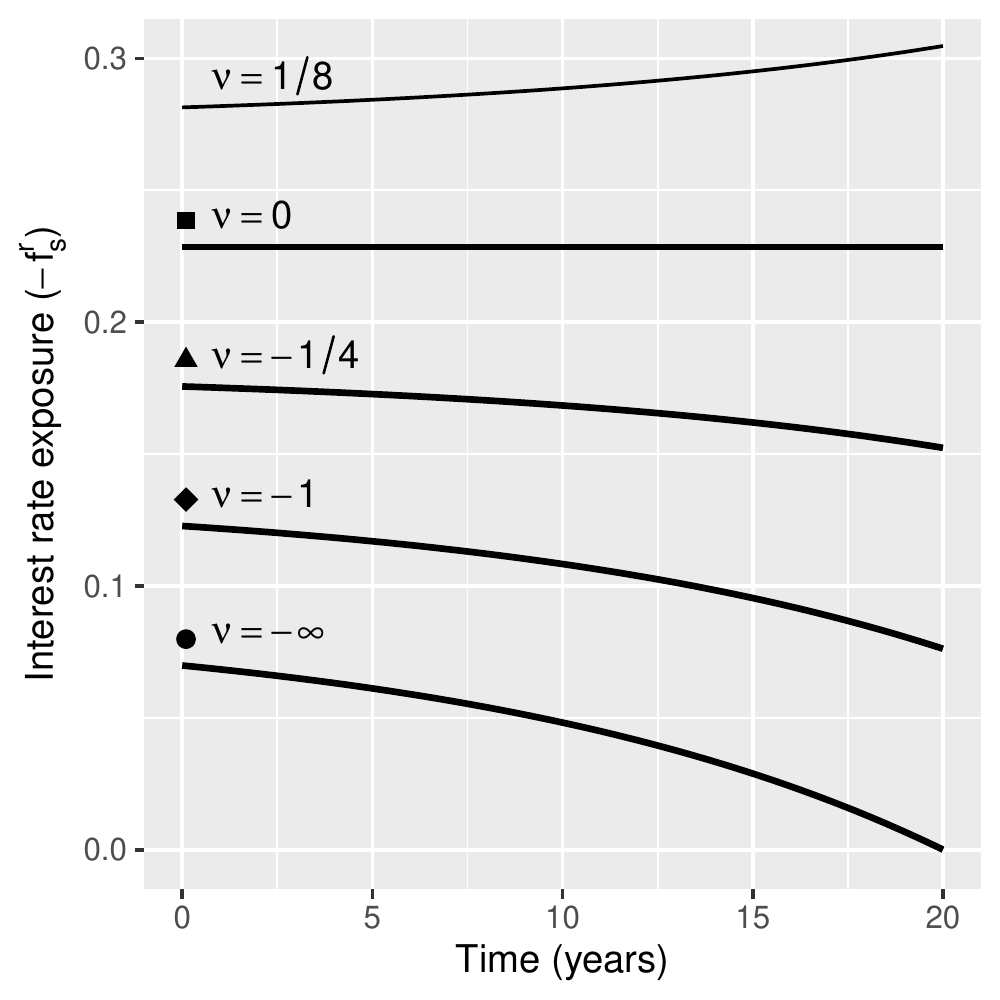}   
\hfill
\includegraphics[height=7.5cm]{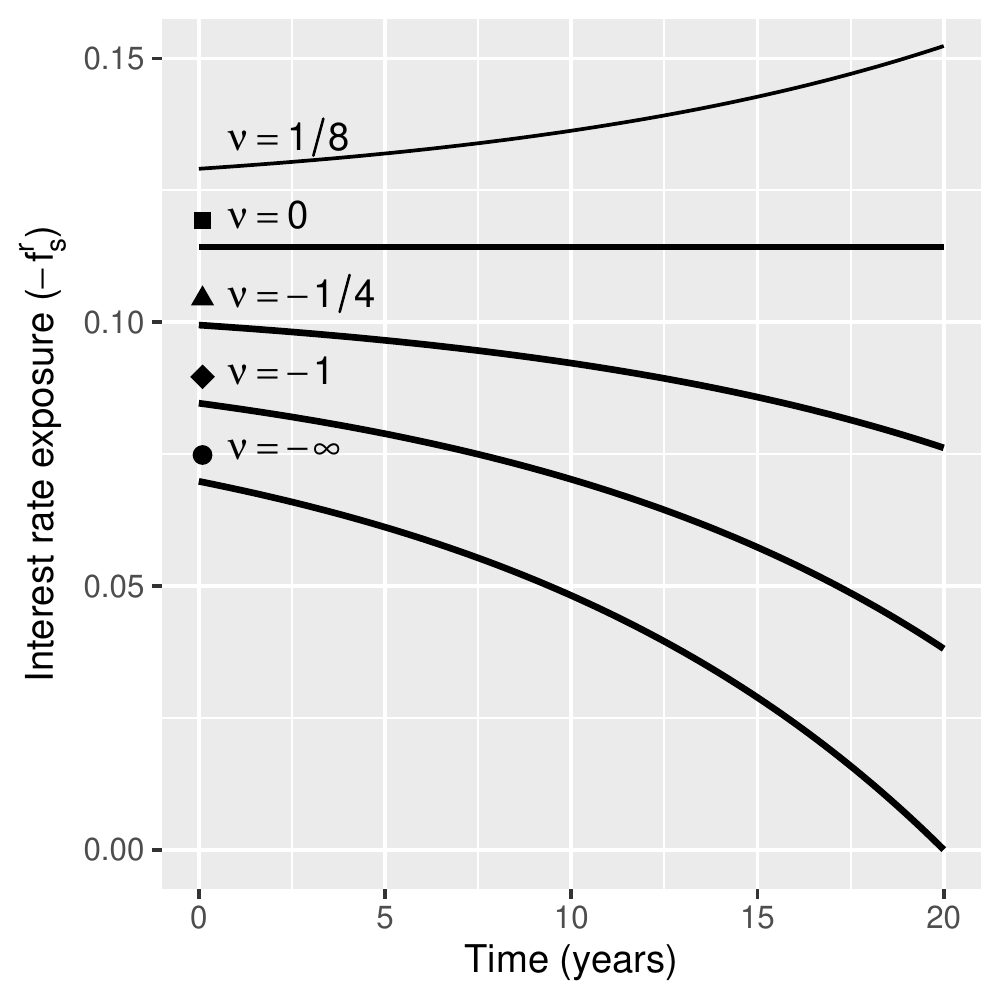}   
\end{center}
\vspace*{-5mm}
\caption{Optimal rate strategies for $T=20$ years when the market price of interest rate risk is moderate (left) and low (right), cf.\ Table~\ref{tab:ratePar}. The symbols distinguishing four of the strategies are also used in Figures~\ref{fig:rateMVprofile} and \ref{fig:rateDistrib} where they mark the risk-reward pair, ($\sigma_T,\mu_T$), and the portfolio value at the horizon, $V_T$, respectively, when the strategies are applied in a market with $r_0=0\%$.}
\label{fig:rateStrategy}
\end{figure}

\clearpage
\subsection{Optimal equity strategies} \label{sec:NumOptimalEquity}
We now turn attention to the optimal strategies of Section~\ref{sec:OptimalEquity}, i.e., optimal strategies where we are allowed to invest only in the equity market. Recall that we are optimizing {\em excess} returns on a given horizon and that we interpret the resulting strategies as overlay strategies, i.e, strategies where the exposure is financed by a corresponding short money-market position. In our setup, an equity strategy acts as a stochastic multiplier to (the portfolio value of) the underlying rate strategy, which in turn acts as a stochastic multiplier to the risk-free, $T$-bond hedging strategy.

\subsubsection{Volatility profiles and calibration of mean-reversion}
We begin with an investigation of the model for equity returns (\ref{eq:stockindex})---(\ref{eq:excessreturn}). The rationale for the model is that equity returns above/below the expected return leads to a decrease/increase in the future risk premium, $x_t$. Thus, there is a tendency for gains to follow losses, and vice versa, and this mean-reversion feature effectively compresses the equity process, and prevents both extreme loses and extreme gains over longer periods. The model is simple, yet it can produce a variety of different volatility profiles for the equity process. In fact, the simplicity of the model is deceptive and counter-intuitive volatility profiles can easily arise even for "innocent looking" parameter values.

There are two parameters of interest in the risk premium process: $\alpha$ controlling the degree of mean reversion, and $\sigma_x$ controlling the size of feedbacks. At first sight, we are free to choose these parameters as we please. A larger/smaller $\sigma_x$ causes a larger/smaller correction of the equity premium, while a larger/smaller value of $\alpha$ causes a shorter/longer "corrective" period. While mean-reversion decreases the volatility of the equity process in the short term, excessive mean-reversion ($\sigma_x$ too large and/or $\alpha$ too small) actually increases volatility in the long term. Intuitively, if mean-reversion parameters are set too aggressively, the uncertainty that builds up in the risk premium process can overshadow the short-term volatility reductions and lead to an overall increase in volatility.

In the following we study this phenomenon more formally and give recommendations on the degree of mean-reversion. We are interested in the process of excess equity returns, which we denote by $\tilde{S}_t$. It follows from Section~3 of \cite{jarpre17} that
\begin{equation}
   \tilde{S}_t \equiv S_t \exp\left\{-\int_0^t r_s ds \right\} = S_0\left\{\int_0^t x_s ds - \frac{\sigma_S^2}{2}t + \sigma_S W^S_t\right\},
\end{equation}
and, further, from Section~6 ibid.\ it follows that
\begin{equation} \label{eq:varlogStilde}
   \Var[\log \tilde{S}_t] = \sigma_x^2 \Upsilon(\alpha,t) + \sigma_S^2 t - 2\sigma_x \sigma_S \Theta(\alpha,t),
\end{equation}
where $\Upsilon$ is given by (\ref{eq:upsilondef}) on page \pageref{eq:upsilondef} of the present report, and $\Theta$ is given by
\begin{align} \label{eq:thetadef}
  \Theta(a,t) \equiv \int_0^t \Psi(a,s) ds =
  \begin{cases}
  \frac{t^2}{2} & \mbox{for } a = 0, \\
  \frac{1}{a^2}\left(-1 + at + e^{-at} \right) & \mbox{for } a \neq 0; \\
  \end{cases}
\end{align}
the interested reader is referred to Appendix~A of \cite{jarpre17} for background information on the expression for the variance. We note that the variance expression in (\ref{eq:varlogStilde}) contains both a positive (first) term arising from the stochasticity of the risk premium itself, and a negative (last) term due to the link between innovations and risk premia. The asymptotic rate of increase of these two terms determine whether mean-reversion leads to an increase or a decrease of the long-term variance---compared to a (Black-Scholes) model with constant risk premium, $x_t = \bar{x}$.

For $\alpha >0$, we have from Section~6 of \cite{jarpre17} that the asymptotic (rate of) variance is given by
\begin{equation} \label{eq:varlogStildeAsymp}
   \lim_{t\to\infty}\frac{\Var[\log \tilde{S}_t]}{t} = \frac{\sigma_x^2}{\alpha^2} + \sigma_S^2 - 2\frac{\sigma_x \sigma_S}{\alpha} = \left(\sigma_S -\frac{\sigma_x}{\alpha}\right)^2 = \sigma_S^2\left(1-\frac{1}{\tilde{\alpha}}\right)^2,
\end{equation}
where $\tilde{\alpha}=\alpha/[\sigma_x/\sigma_S]$ is the mean-reversion ratio encountered in Section~\ref{sec:OptEqStrategyNegTradeOff}. It follows from (\ref{eq:varlogStildeAsymp}) that the asymptotic variance is the same, or below, the variance of the Black-Scholes model if and only if $\tilde{\alpha}\geq 1/2$.

\begin{table}[ht]
  \centering
\rowcolors{2}{gray!25}{white}
\bgroup
\def\arraystretch{1.3}
\begin{tabular}{C{1.5cm}C{1.4cm}C{1.4cm}C{1.4cm}C{1.4cm}|C{2cm}C{1.8cm}} \hline
  \rowcolor{gray!50}
 Par.\ set &  $\sigma_S$ & $\sigma_x$  &  $\alpha$  & $\tilde{\alpha}$ &  Asym.\ vol.  & SD($x_\infty$) \\ \hline
     1     &    0.15     &    0        &     0.06   &       $\infty$   &   0.15        &      0         \\
     2     &    0.15     &    0.003    &     0.06   &       3.00       &   0.10        &      0.0087    \\
 {\bf 3}   & {\bf 0.15}  & {\bf 0.007} & {\bf 0.06} & {\bf 1.30}       & {\bf 0.033}   &  {\bf 0.020}   \\
     4     &    0.15     &    0.009    &     0.06   &         1        &     0         &      0.026     \\
     5     &    0.15     &    0.015    &     0.06   &       0.60       &   0.10        &      0.043     \\
     6     &    0.15     &    0.020    &     0.06   &       0.45       &   0.18        &      0.058     \\
     7     &    0.15     &    0.030    &     0.06   &       0.30       &   0.35        &      0.087     \\ \hline
\end{tabular}
\egroup
  \caption{Set of parameters (with varying values of $\sigma_x$) governing the mean-reversion of equity returns. Also shown is the mean-reversion ratio $\tilde{\alpha}=\alpha/[\sigma_x/\sigma_S]$, the asymptotic rate of volatility, $|\sigma_S-\sigma_x/\alpha|$, and the standard deviation in the stationary distribution of the risk premium, $\sigma_x/\sqrt{2\alpha}$.}
  \label{tab:meanrevParSigmax}
\bigskip
\bgroup
\def\arraystretch{1.3}
\begin{tabular}{C{1.5cm}C{1.4cm}C{1.4cm}C{1.4cm}C{1.4cm}|C{2cm}C{1.8cm}} \hline
  \rowcolor{gray!50}
 Par.\ set &  $\sigma_S$ & $\sigma_x$  &  $\alpha$  & $\tilde{\alpha}$ &  Asym.\ vol.  & SD($x_\infty$) \\ \hline
     A     &    0.15     &    0.007    &     0.90   &       19.3       &   0.14        &      0.0052    \\
     B     &    0.15     &    0.007    &     0.14  &       3.00       &   0.10        &      0.013    \\
{\bf C}    & {\bf 0.15}  & {\bf 0.007} & {\bf 0.06} &  {\bf 1.30}      &   {\bf 0.033} & {\bf 0.020}     \\
     D     &    0.15     &    0.007    &     0.047  &         1        &     0         &      0.023     \\
     E     &    0.15     &    0.007    &     0.020  &       0.43       &   0.20        &      0.035     \\
     F     &    0.15     &    0.007    &     0.010  &       0.21       &   0.55        &      0.049     \\
     G     &    0.15     &    0.007    &     0      &        0         &    $\infty$   &   $\infty$     \\ \hline
\end{tabular}
\egroup
  \caption{Sets of parameters (with varying values of $\alpha$) governing the mean-reversion of equity returns. Also shown is the mean-reversion ratio $\tilde{\alpha}=\alpha/[\sigma_x/\sigma_S]$, the asymptotic rate of volatility, $|\sigma_S-\sigma_x/\alpha|$, and the standard deviation in the stationary distribution of the risk premium, $\sigma_x/\sqrt{2\alpha}$.}
  \label{tab:meanrevParAlpha}
\end{table}

The asymptotic volatility, i.e., the square-root of (\ref{eq:varlogStildeAsymp}), is given by $|\sigma_S-\sigma_x/\alpha|$. For $\tilde{\alpha} > 1$, $\sigma_x/\alpha < \sigma_S$ and the effect of mean-reversion can be interpreted as reducing the (local) volatility. Conversely, for $\tilde{\alpha} < 1$, $\sigma_x/\alpha > \sigma_S$ and the volatility reduction due to mean-reversion in a sense "overshoots" the (local) volatility; if the overshoot is large enough mean-reversion in fact increases, rather than reduces, volatility. Also note, that for $\tilde{\alpha}=1$ the two terms are equal and stocks are asymptotically risk free!

Figure~\ref{fig:volProfile} shows the annualized volatility, i.e, $(\Var[\log \tilde{S}_t]/t)^{1/2}$, as a function of time for the parameter sets in Tables~\ref{tab:meanrevParSigmax} and~\ref{tab:meanrevParAlpha}.
The highlighted, third parameter set is inspired by the empirical estimates reported in Table~1 of \cite{munetal04} and it is the same in the two tables. With this parameter set as a starting point, we vary the values of $\sigma_x$ and $\alpha$ in Tables~\ref{tab:meanrevParSigmax} and~\ref{tab:meanrevParAlpha}, respectively. In each case, we see that increased mean-reversion lowers the asymptotic volatility, but only up to a certain point (given by $\tilde{\alpha}=1$). Beyond that point the asymptotic volatility starts to increase again and it eventually diverges as $\tilde{\alpha}$ tends to zero. As mentioned above, this is due to the uncertainty that accumulates in the risk premium process. In the last column of each table, the standard deviation of the stationary distribution of the risk premium process is shown as a measure of this uncertainty.

Clearly, not every volatility profile in Figure~\ref{fig:volProfile} is suitable for modelling equity returns. It is noteworthy, that all parameter values appear reasonable and in fact they all (except set A) are within one standard error of their empirical estimates, cf.~\cite{munetal04}. Thus, some kind of expert judgment is needed in the calibration process. The theoretical foundation for the model seems strongest for $\tilde{\alpha} > 1$, and as a rule of thumb we recommend using parameter sets satisfying this constraint. Parameter sets with $\tilde{\alpha}\leq 1$ might be useful on shorter horizons, but on longer horizons they can lead to counter-intuitive results, as illustrated later.

\begin{figure}[ht]
\begin{center}
\includegraphics[height=7.5cm]{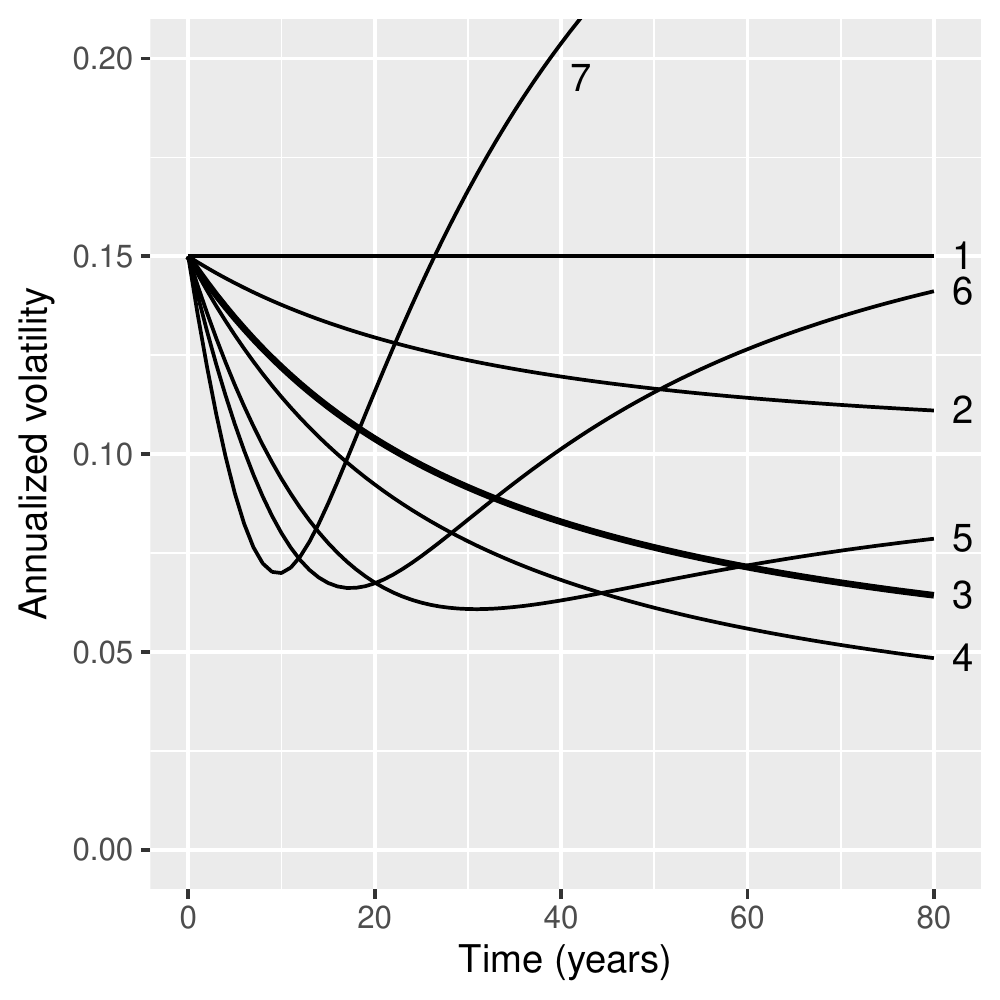}   
\hfill
\includegraphics[height=7.5cm]{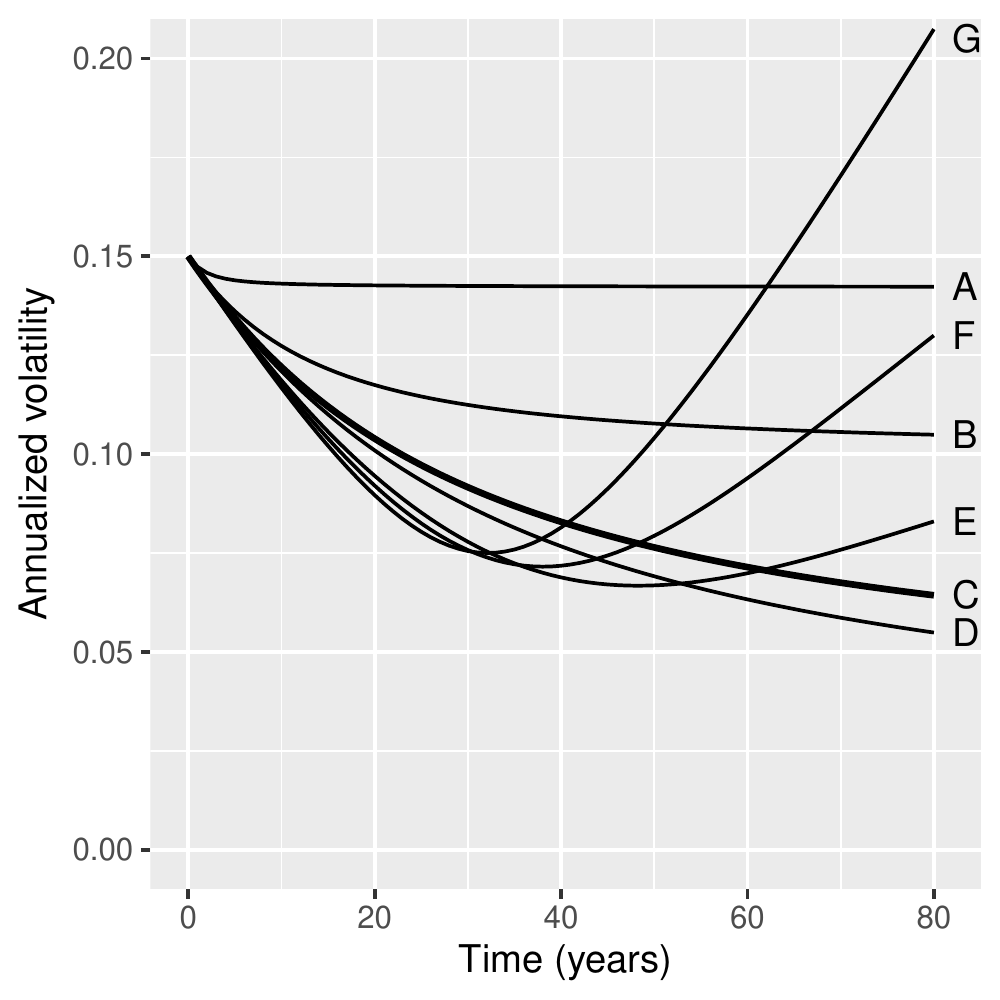}   
\end{center}
\vspace*{-5mm}
\caption{Illustration of the annualized volatility of log excess stock returns, i.e., $(\Var[\log \tilde{S}_t]/t)^{1/2}$ as a function of $t$. The left plot illustrates the parameter sets in Table~\ref{tab:meanrevParSigmax}, where $\sigma_x$ is varied; the right plot illustrates the parameter sets in Table~\ref{tab:meanrevParAlpha}, where $\alpha$ is varied. The emphasized profile (3 and C) is the same in the two plots. It is based on the empirical estimates reported by \cite{munetal04}.}
\label{fig:volProfile}
\end{figure}

\subsubsection{Risk-reward profiles} \label{sec:MVprofileOptimalEquity}
In this section we illustrate the optimal risk-reward trade-off that can be obtained by equity strategies, and we give examples of the optimal strategies. We also compare the optimal strategies with constant strategies; constant strategies provide a natural benchmark due to their optimality in the Black-Scholes model. In particular, it is instructive to see how the horizon and the degree of mean-reversion affect the performance of optimal and constant strategies differently.

We consider the two parameter sets given in Table~\ref{tab:equityPar}, corresponding to a moderate ($\tilde{\alpha}=1.30$) and high ($\tilde{\alpha}=0.60$) degree of mean-reversion. The two sets are the same as sets 3 and 5 of Table~\ref{tab:meanrevParSigmax} with an added assumption on the risk premium in stationarity, $\bar{x}=0.045$. For both parameter sets, the (expected) market price of equity risk in stationarity is $\bar{\lambda}^S=0.30$. This is higher than the market price of interest rate risk used in Section~\ref{sec:NumOptimalRate}, cf.~Table~\ref{tab:ratePar}, but lower than values typically used for equities.\footnote{\cite{munetal04} reports on empirical estimates of $\bar{x}$ and $\sigma_S$ of $0.0648$ and $0.1468$, respectively. This corresponds to a market price of equity risk in stationary of $\bar{\lambda}^S=0.44$, which is more in line with levels typically used. However, we do not wish to overstate the return potential on stocks going forward, and therefore we use more conservative return assumptions in this paper.} Note, that the market price of equity risk is stochastic, with sizeable variations both between paths and within paths. This stochasticity is part of the optimization problem and the main driver of the different (optimal) strategies that arise.

\begin{table}[ht]
  \centering
\rowcolors{2}{gray!25}{white}
\bgroup
\def\arraystretch{1.3}
\begin{tabular}{C{2.6cm}C{1.2cm}C{1.2cm}C{1.2cm}C{1.2cm}C{1.2cm}C{1.2cm}} \hline
  \rowcolor{gray!50}
    Parameter set  &   $\bar{x}$   & $\sigma_S$  &  $\sigma_x$  & $\alpha$  &  $\tilde{\alpha}$  & $\bar{\lambda}^S$ \\ \hline
    Moderate       &    0.045      &   0.015     &   0.007      & 0.06      &     1.30           &  0.30  \\
    High           &    0.045      &   0.015     &   0.015      & 0.06      &     0.60           &  0.30  \\  \hline
\end{tabular}
\egroup
  \caption{Parameters governing the dynamics for equity returns (\ref{eq:stockindex})---(\ref{eq:excessreturn}). Also shown is $\tilde{\alpha}=\alpha/[\sigma_x/\sigma_S]$ as a measure of mean reversion, and
  the market price of equity risk in stationary, $\bar{\lambda}^S=\bar{x}/\sigma_S$. The latter is the continuous-time analogue to the Sharpe ratio, i.e., excess return per unit of volatility for holding equities. }
  \label{tab:equityPar}
\end{table}

It follows from Corollaries~\ref{cor:VTratesonly} and~\ref{cor:VTequityonly}, that we can represent the portfolio value as 
\begin{equation}  \label{eq:VTrateEquityStructure}
   V_T \stackrel{\cal{D}}{=} \frac{V_0}{p_0(T)}Y_T Z_T, 
\end{equation}
where $Y_T$ and $Z_T$ are independent, and determined by the rate and equity strategies, respectively. Further, $Z_T$ is log-normally distributed, $\log Z_T \sim N(\mu_T,\sigma_T^2)$, where $\mu_T$ and $\sigma_T^2$ are given by (\ref{eq:VTmeanequityonly}) and (\ref{eq:VTvarequityonly}), respectively. Due to (\ref{eq:VTrateEquityStructure}), we can study the effect of the rate and equity strategies separately. In Appendix~\ref{app:statistics}, Tables~\ref{tab:equityStatisticsMod} and~\ref{tab:equityStatisticsHigh} show risk and reward statistics for $Z_T$ for a range of optimal equity strategies. 

Figures~\ref{fig:equityMVprofile20Y} and \ref{fig:equityMVprofile40Y} plot $\mu_T$ against $\sigma_T$ on horizons 20 and 40 years, respectively, for the optimal equity strategies. The profiles are constructed by varying $\nu$ from $-\infty$ to $1/2$, both values excluded. For each of the selected values of $\nu$, the corresponding optimal strategy is found by Theorem~\ref{thm:ExtremalOverview}, and $\mu_T$ and $\sigma_T^2$ are then computed by numeric integration of (\ref{eq:VTmeanequityonly}) and (\ref{eq:VTvarequityonly}), respectively, with the optimal strategy inserted.\footnote{In fact, since $\tilde{\alpha} > 1/2$ for the two parameter sets under consideration and since $\nu<1/2$ the optimal strategies are all of exponential form (type I), cf.\ Figure~\ref{fig:domain} of Appendix~\ref{app:overviewExtremal}, and Theorem~\ref{thm:OptimalEquity} also applies.}

Each plot also shows the risk-reward trade-off for constant strategies, $f_s^S=c$. For constant strategies the log-mean and log-variance of $Z_T$ can be calculated explicitly as
\begin{align}
  \mu_T & = \frac{c}{\sigma_S}\left[T\bar{x} + \frac{x_0 - \bar{x}}{\alpha}\left(1-e^{-\alpha T}\right)  \right] - \frac{T}{2}c^2, \\
  \sigma_T^2 & = c^2\left[T \left(\frac{\tilde{\alpha}-1}{\tilde{\alpha}}\right)^2  + 2 \frac{\tilde{\alpha}-1}{\alpha\tilde{\alpha}}\left(1-e^{-\alpha T}\right)
     + \frac{1}{2\alpha\tilde{\alpha}^2}\left(1-e^{-2\alpha T}\right) \right],  \label{eq:constEquityVolExcess}
\end{align}
assuming both $\alpha$ and $\tilde{\alpha}$ are non-zero. The dedicated reader is encouraged to implement the variance formula and compare its behaviour as a function of time to the variance of a constant strategy in a Black-Scholes model, $\sigma_T^2=c^2 T$.

\begin{figure}[ht]
\begin{center}
\includegraphics[height=7.5cm]{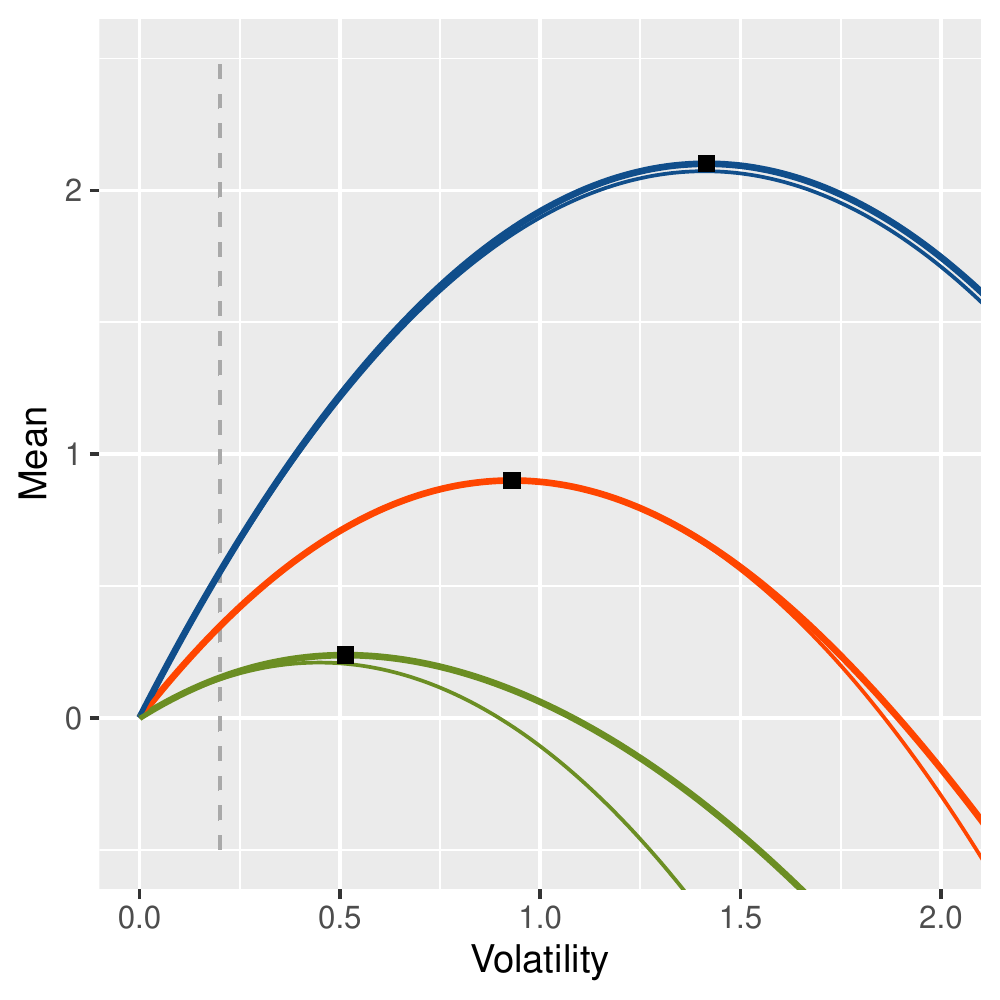}   
\hfill
\includegraphics[height=7.5cm]{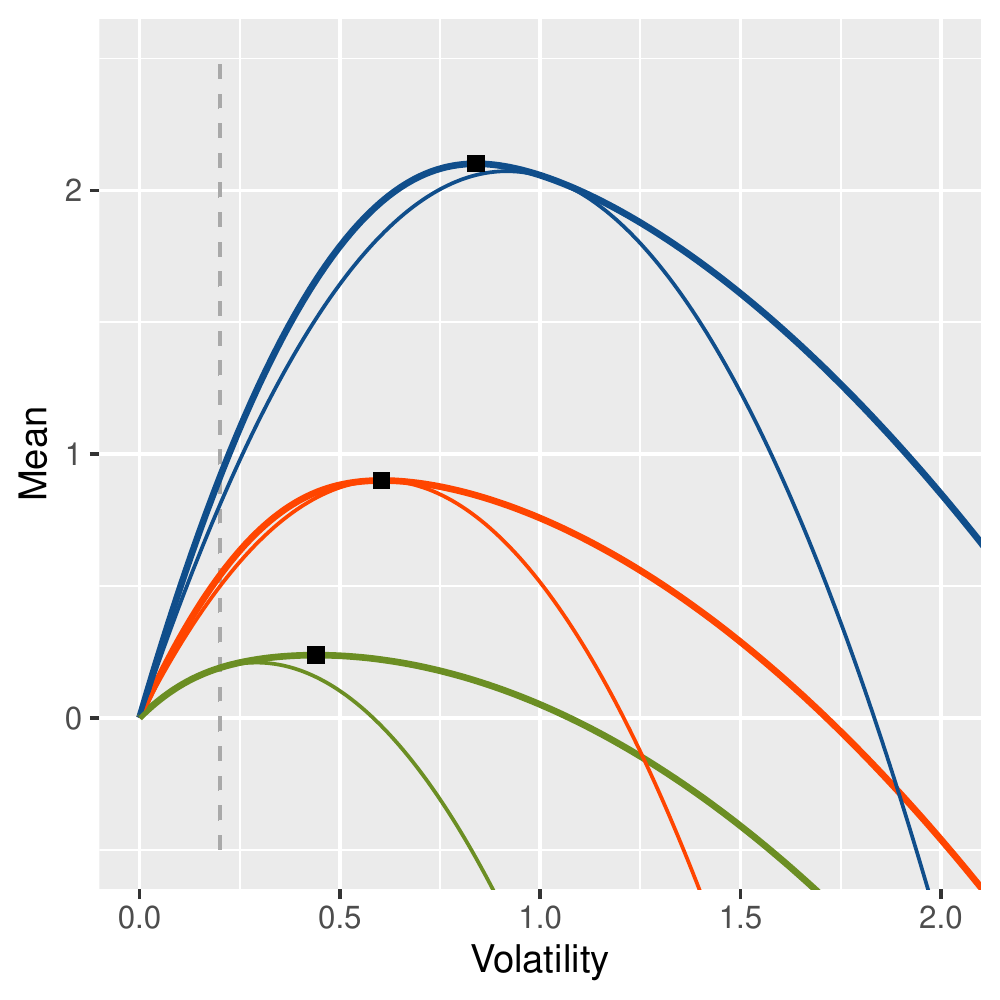}   
\end{center}
\vspace*{-5mm}
\caption{Risk-reward profiles, $(\sigma_T,\mu_T)$, for optimal (thick) and constant (thin) equity strategies on a horizon of $T=20$ years for moderate (left) and high (right) degrees of mean-reversion, cf.\ Table~\ref{tab:equityPar}. The different profiles correspond to $x_0$ of $0.005$ (green), $0.045$ (orange), and $0.085$ (blue). The global optimum is obtained for $\nu=0$ ({\scriptsize $\blacksquare$}). The vertical dashed line marks $\sigma_T=0.2$, the corresponding optimal and constant strategies are illustrated in Figure~\ref{fig:equityStrategy20Y}.}
\label{fig:equityMVprofile20Y}
\end{figure}

\begin{figure}[ht]
\begin{center}
\includegraphics[height=7.5cm]{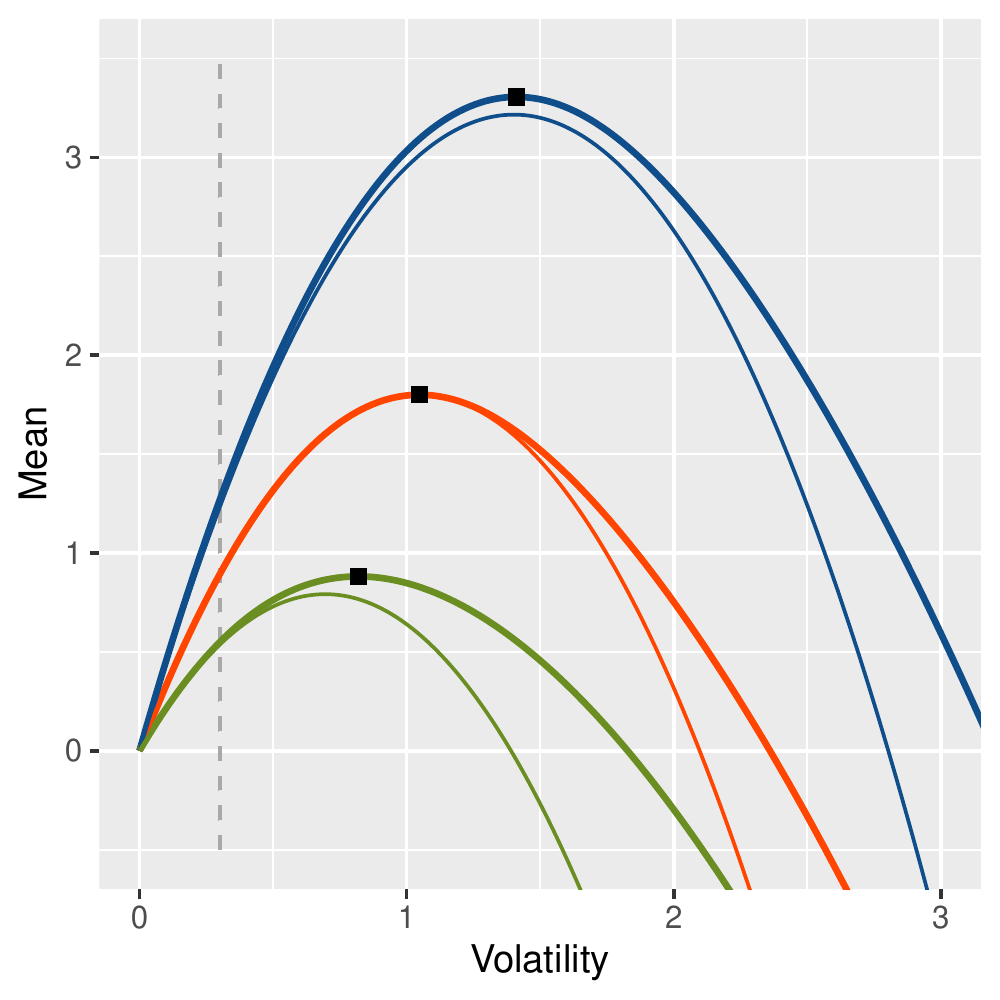}   
\hfill
\includegraphics[height=7.5cm]{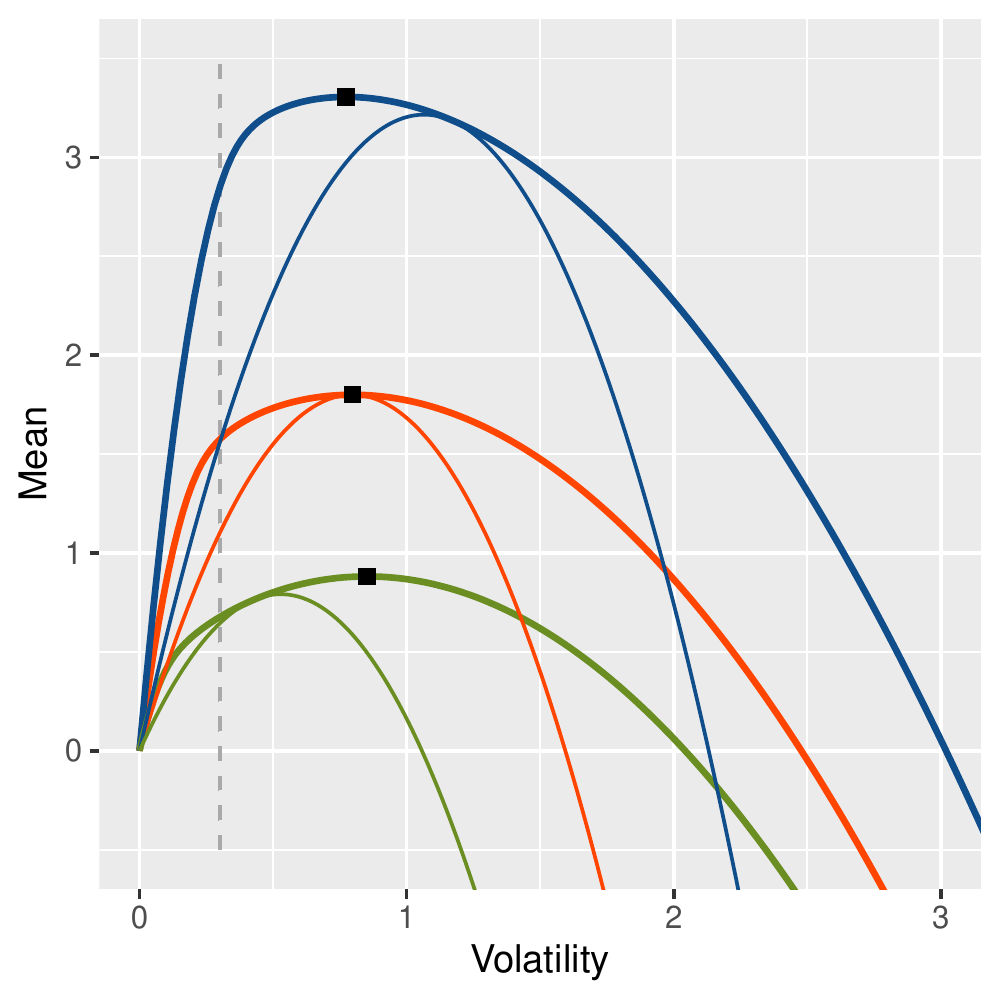}   
\end{center}
\vspace*{-5mm}
\caption{Risk-reward profiles, $(\sigma_T,\mu_T)$, for optimal (thick) and constant (thin) equity strategies on a horizon of $T=40$ years for moderate (left) and high (right) degrees of mean-reversion, cf.\ Table~\ref{tab:equityPar}. The different profiles correspond to $x_0$ of $0.005$ (green), $0.045$ (orange), and $0.085$ (blue). The global optimum is obtained for $\nu=0$ ({\scriptsize $\blacksquare$}). The vertical dashed line marks $\sigma_T=0.3$, the corresponding optimal and constant strategies are illustrated in Figure~\ref{fig:equityStrategy40Y}.}
\label{fig:equityMVprofile40Y}
\end{figure}

We see from Figures~\ref{fig:equityMVprofile20Y} and \ref{fig:equityMVprofile40Y} that for moderate levels of mean-reversion the optimal strategies improve the mean only slightly compared to constant strategies, at least for the strategies of interest, i.e., strategies with $\nu \leq 0$. However, for high levels of mean-reversion substantial improvements are obtained on the 40-year horizon. Indeed, the optimal risk-reward profile in the right plot of Figure~\ref{fig:equityMVprofile40Y} is surprising. The profile is very steep initially meaning that substantial equity gains can be achieved on long horizons with very little risk. Mathematically, the high degree of mean-reversion allows a very effective netting of equity fluctuations---but only on long horizons, and only if the exposure is taken in the right way.

\subsubsection{Illustration of optimal strategies}
In Figures~\ref{fig:equityStrategy20Y} and \ref{fig:equityStrategy40Y} we illustrate the optimal equity strategies with $\sigma_T=0.2$ on 20 years horizon and $\sigma_T=0.3$ on 40 years horizon, respectively. The optimal strategies exploit the dynamics of the risk premium process, and different strategies will therefore be employed depending on the initial value of the risk premium. Therefore, there are three optimal strategies in each plot, corresponding to the three different values of $x_0$. In contrast, we see from (\ref{eq:constEquityVolExcess}) that the volatility of a constant strategy does not depend on $x_0$. Consequently, for given volatility target the same constant strategy will be used regardless of $x_0$; the horizontal, black line in each plot shows this strategy.

The natural scale (left axis) for the strategies is the exposure, $f^S_s$, to the Brownian motion driving equity returns. When positive, the exposure can be interpreted as the (local) volatility due to equity investments. To aid interpretation we also show (right axis) the exposure in terms of the corresponding equity share, i.e., $f^S_s/\sigma_S$. An equity share of, say, 50 pct.\ means that the equity exposure should equal the exposure obtained from investing half the portfolio in equities. Note, since all equity exposures are financed by equivalent short positions in cash, there is no upper limit on the equity exposure/share.

The optimal strategies are generally decreasing over time, except when the initial risk premium is sufficiently low. We also know that for moderate mean-reversion only little is gained from following the optimal strategy compared to a constant strategy, even though the strategies are quite different. However, with high mean-reversion and a long horizon optimal strategies substantially outperform constant strategies. We see from the right plot of Figure~\ref{fig:equityStrategy40Y} that the optimal exposure is initially very high and declines rapidly thereafter. When mean-reversion is high, this profile is very effective at netting equity fluctuations over time, and constant strategies can only obtain the same volatility target by having a much smaller exposure.

\begin{figure}[ht]
\begin{center}
\includegraphics[height=7.5cm]{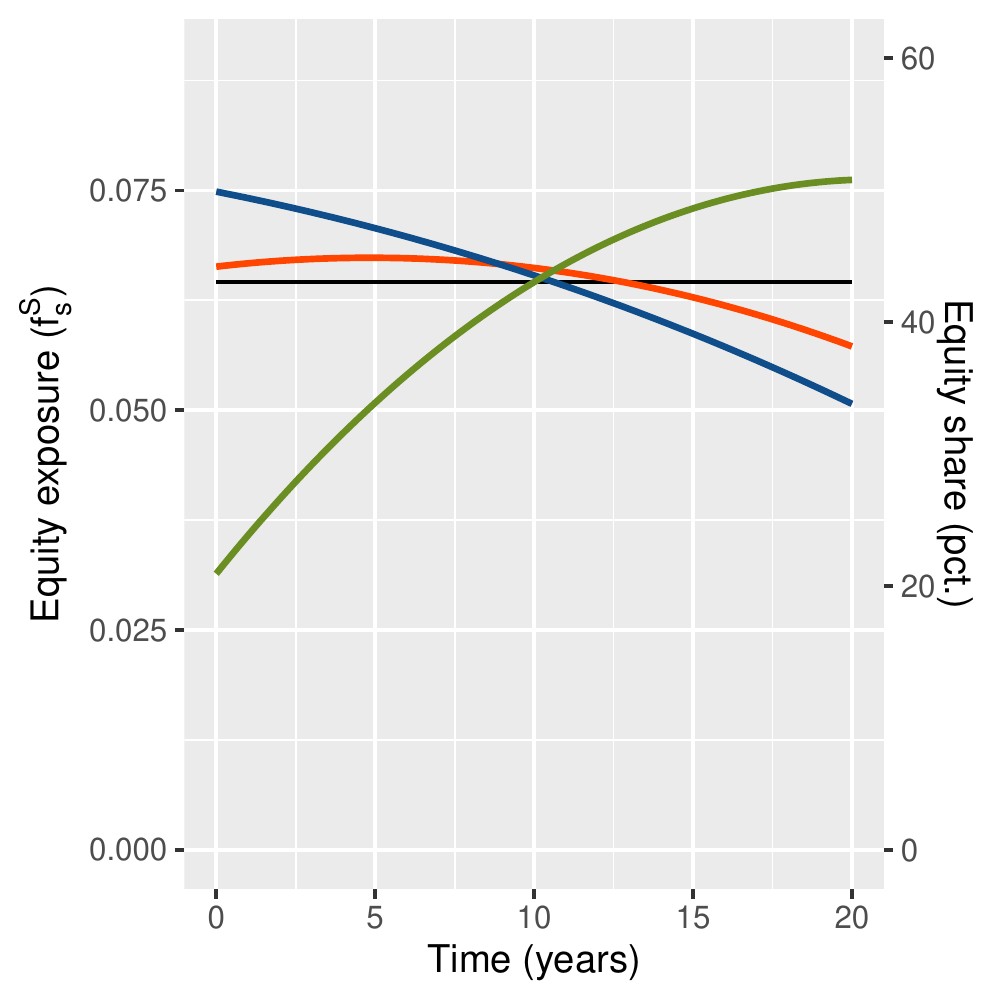}   
\hfill
\includegraphics[height=7.5cm]{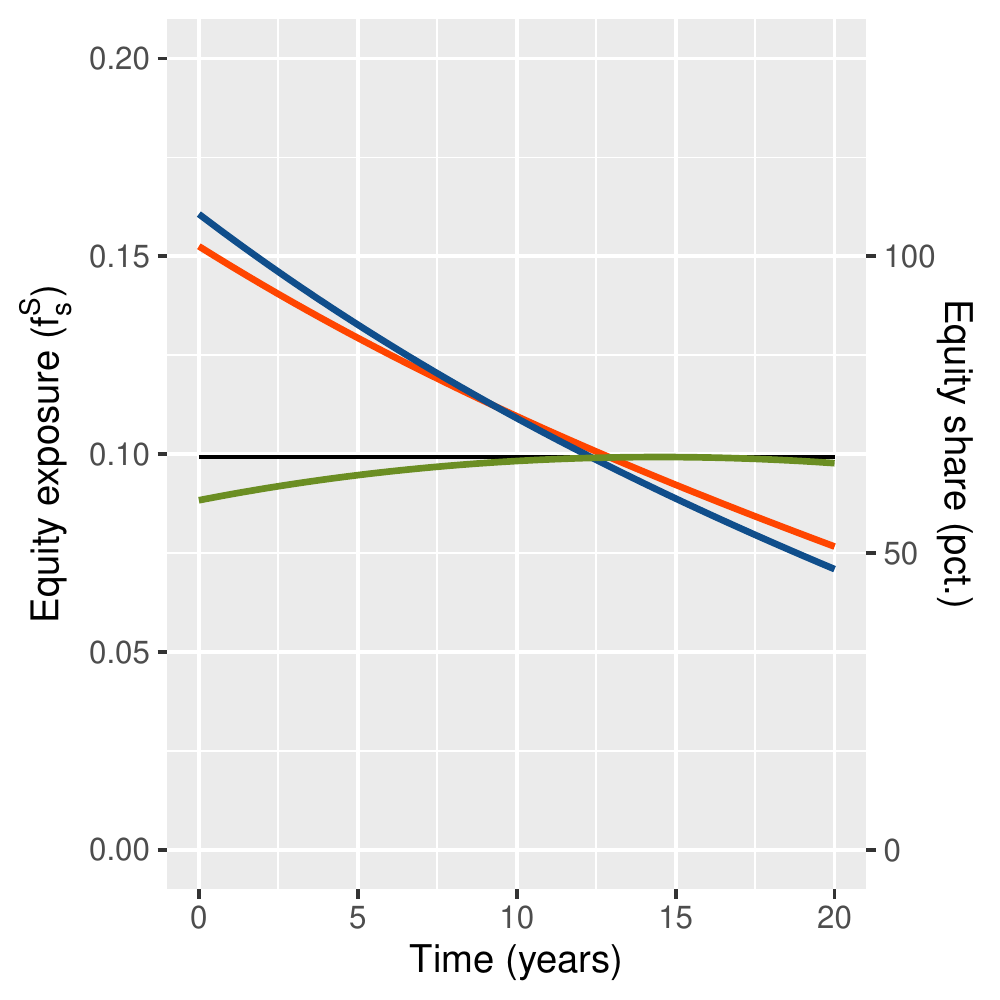}   
\end{center}
\vspace*{-5mm}
\caption{Optimal equity strategies with $\sigma_T=0.2$ for $T=20$ years when the degree of mean-reversion is moderate (left) and high (right), cf.\ Table~\ref{tab:equityPar}. The different profiles correspond to $x_0$ of $0.005$ (green), $0.045$ (orange), and $0.085$ (blue). The horizontal black line shows the constant equity strategy that achieves $\sigma_T=0.2$.}
\label{fig:equityStrategy20Y}
\end{figure}

\begin{figure}[ht]
\begin{center}
\includegraphics[height=7.5cm]{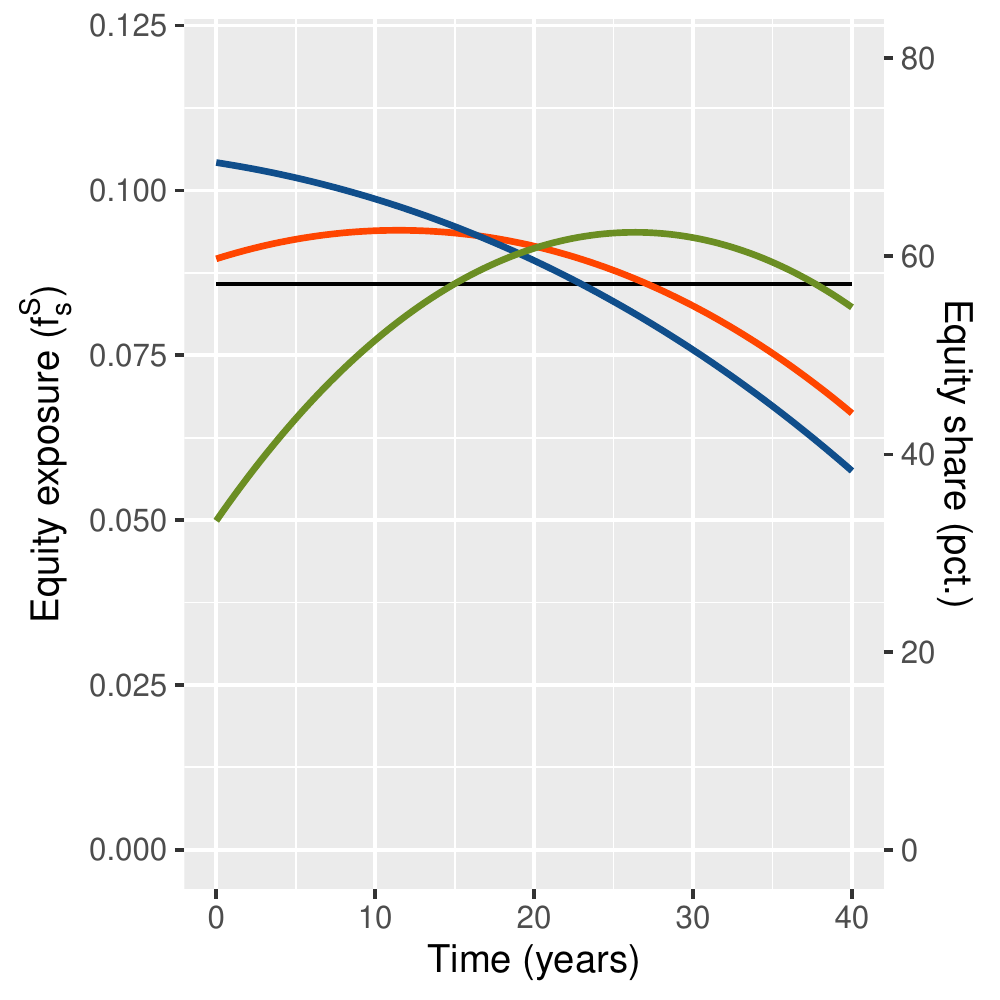}   
\hfill
\includegraphics[height=7.5cm]{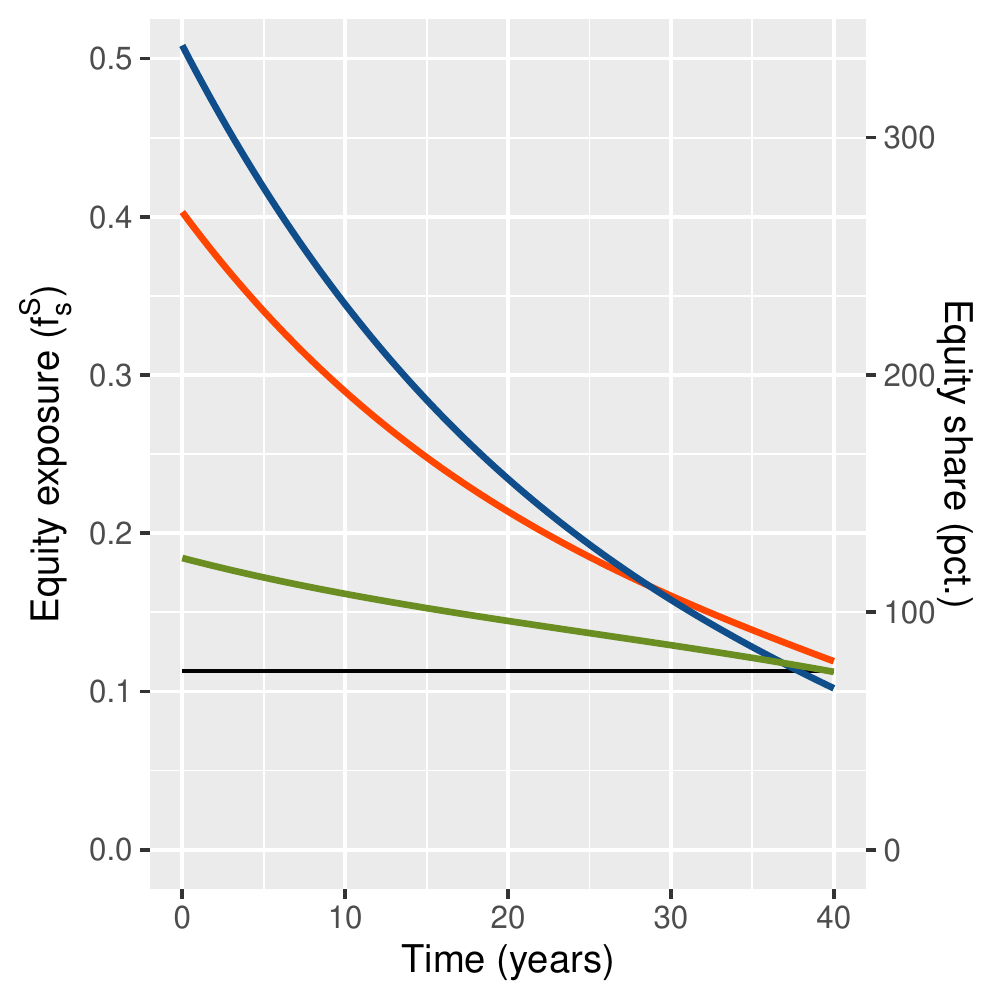}   
\end{center}
\vspace*{-5mm}
\caption{Optimal equity strategies with $\sigma_T=0.3$ for $T=40$ years when the degree of mean-reversion is moderate (left) and high (right), cf.\ Table~\ref{tab:equityPar}. The different profiles correspond to $x_0$ of $0.005$ (green), $0.045$ (orange), and $0.085$ (blue). The horizontal black line shows the constant equity strategy that achieves $\sigma_T=0.3$.}
\label{fig:equityStrategy40Y}
\end{figure}
\clearpage

\subsubsection{Interior wedges of extremal strategies} \label{sec:numWedge}
We conclude the numerical section with an illustration of a surprising phenomenon due to excessive mean-reversion. The risk-reward profiles shown in Section~\ref{sec:MVprofileOptimalRate} have two branches. An upper branch consisting of optimal rate strategies and a lower branch consisting of minimizing rate strategies. From the calculations in Section~\ref{sec:closedFormMeanVar} it follows that there can be no other extremal rate strategies. For extremal equity strategies the situation is considerably more complex.

The left plot of Figure~\ref{fig:equityProfileWedge} shows the risk-reward profile for the entire family of extremal equity strategies ($\nu\neq 1/2$) for the parameter set of Table~\ref{tab:equityPar} with high mean-reversion. We show the profile for a horizon of $T=40$ years and $x_0=\bar{x}$, but other horizons and initial risk premiums give similar plots. The upper branch of the profile is formed by $\nu < 1/2$; this part is also shown in the right plot of Figure~\ref{fig:equityMVprofile40Y} (orange thick line). The remaining extremal strategies, however, do not all line up in a lower branch, as might be expected. Instead, an interior wedge of  strategies is formed, in addition to the lower branch of minimizing strategies. The wedge consists of locally minimizing and maximizing strategies, i.e., strategies where, for given variance target, the mean cannot be improved, either upwards or downwards, in a vicinity of the strategies (in the underlying function space of strategies). A similar wedge does not occur for the parameter set of Table~\ref{tab:equityPar} with moderate mean-reversion.

\begin{figure}[ht]
\begin{center}
\includegraphics[height=7.5cm]{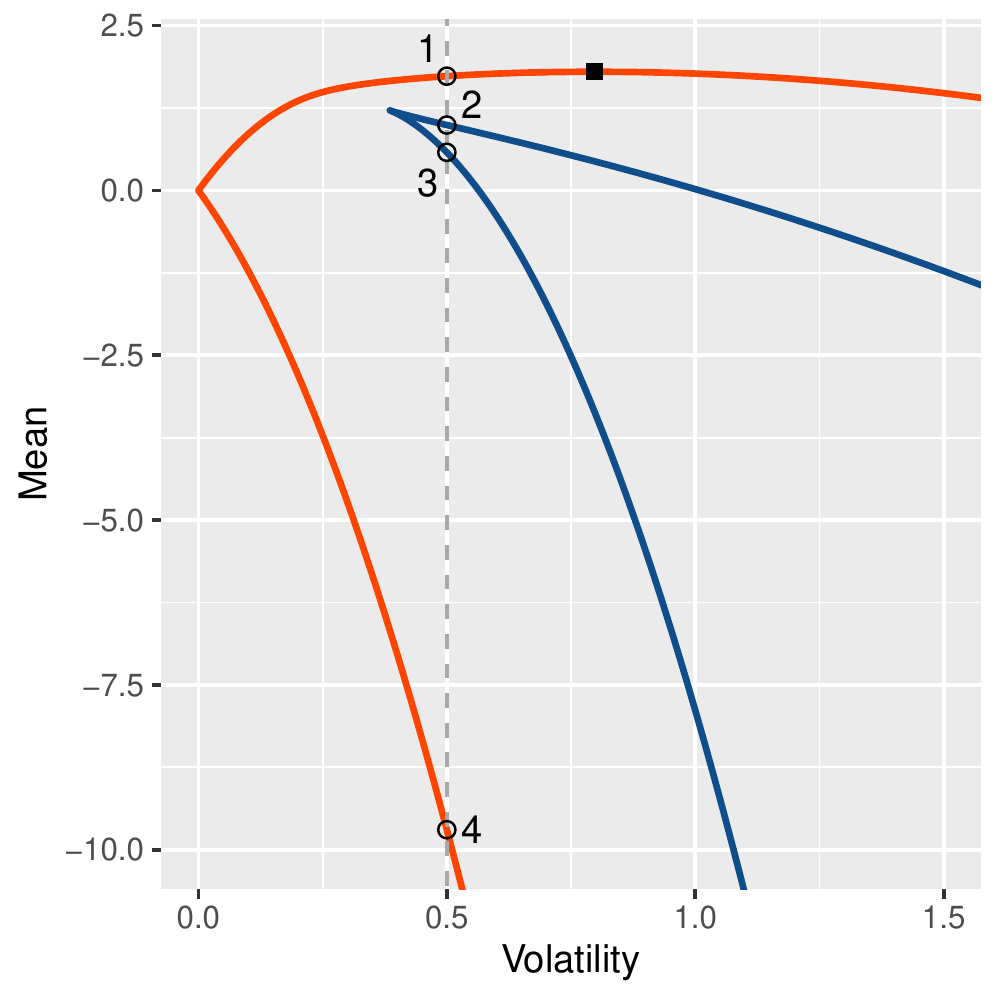}   
\hfill
\includegraphics[height=7.5cm]{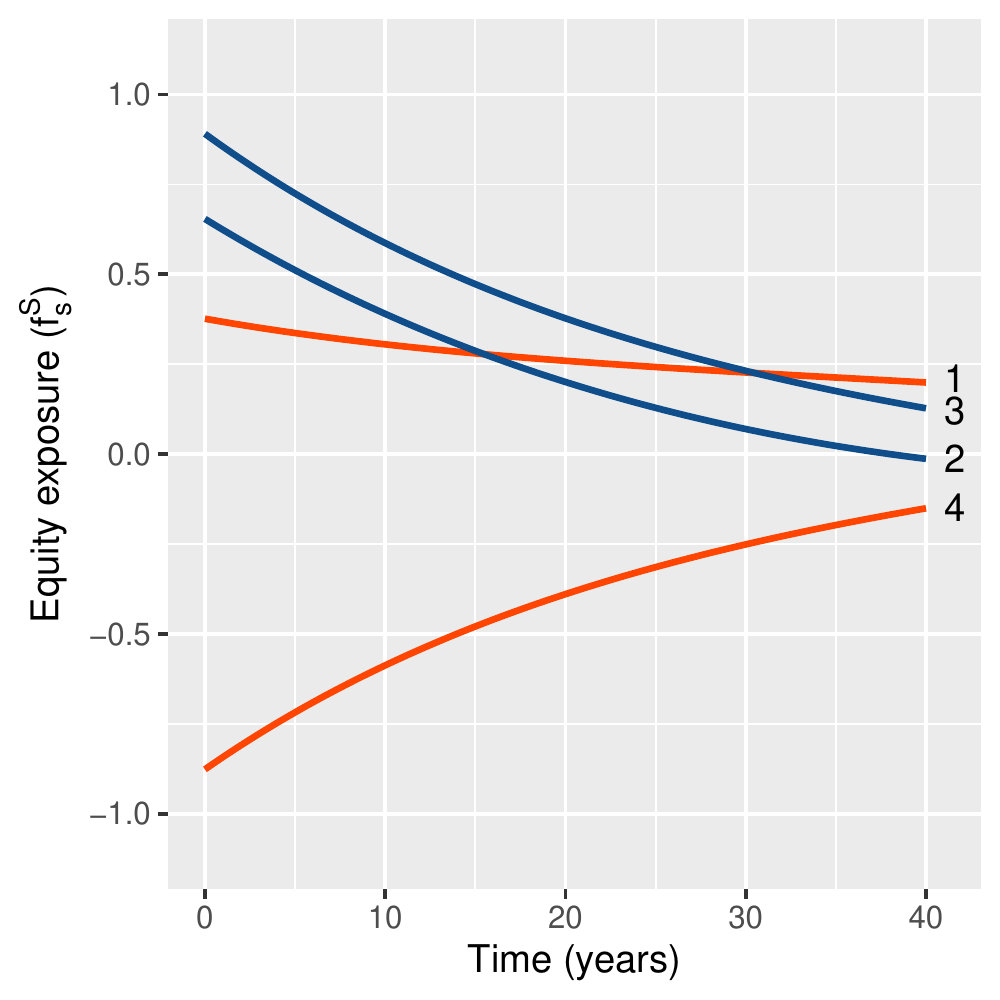}   
\end{center}
\vspace*{-5mm}
\caption{Left plot shows the risk-reward profile, $(\sigma_T,\mu_T)$, for the family of extremal strategies ($\nu\neq 1/2$) on a horizon of $T=40$ years for the parameter set of Table~\ref{tab:equityPar} with high mean-reversion and $x_0=\bar{x}=0.045$. The upper branch corresponds to $\nu<1/2$, with the global maximum attained for $\nu=0$ ({\scriptsize $\blacksquare$}); the interior wedge corresponds to $1/2 < \nu < 15$ (rounded), and the lower branch corresponds to $\nu > 15$. The vertical dashed line marks $\sigma_T=0.5$, the four extremal strategies attaining this level of variability are illustrated in the right plot. The colors distinguish between globally (orange) and locally (blue) minimizing/maximizing strategies.}
\label{fig:equityProfileWedge}
\end{figure}

To investigate further this phenomenon, the right plot of Figure~\ref{fig:equityProfileWedge} shows the four different extremal strategies with $\sigma_T=0.5$. The optimal strategy (1) is similar to the optimal strategy shown in Figure~\ref{fig:equityStrategy40Y} for $\sigma_T=0.3$ (orange thick line). The "drastic" change in exposure over time of the optimal strategy is, however, dwarfed by the radical change in exposure of the minimizing strategy (4). Without a high degree of mean-reversion such a strategy would lead to a very high variance, but high mean-reversion in combination with declining (absolute) exposure all but cancel aggregate return fluctuations.

The locally minimizing strategy (3) is essentially the mirror strategy of the globally minimizing strategy: The same glidepath is used to control the variance, but with positive rather than negative exposure. The reason this strategy has a relatively low mean, is that the extreme leverage results in a large quadratic term that impairs the mean, cf.\ (\ref{eq:VTmeanequityonly}). The locally maximizing strategy (2) uses a dampened version of the same glidepath, but with an overall lower exposure resulting in a lower quadratic term and a higher mean.

In summary, high degrees of mean-reversion implies the existence of very aggressive glidepaths with low aggregate variance. Depending on the sign of the exposure, these glidepaths can be used as building blocks to create either very low, or high, but not optimal, portfolio means. It is the presence of these glidepaths that generates a richer set of extremal strategies than typically seen. Apparently, it requires a high degree of mean-reversion before the interior wedge appears. We conjecture that $\tilde{\alpha}<1$ is a necessary condition, but we leave this as an open research question.

When mean reversion is even higher, e.g., parameter set 7 in Table~\ref{tab:meanrevParSigmax}, a multitude of interior wedges can appear with locally minimizing/maximizing exposure profiles of trigonometric form (not illustrated). These periodic profiles constitute yet another type of building block for obtaining (relatively) low variance even when the absolute exposure is very high. The structural richness of the extremal strategies is intriguing and further insights about the role of mean-reversion can undoubtedly be uncovered, but we leave that for others to explore.

\clearpage
\section{Concluding remarks} \label{sec:Conclusion}
In this report we have derived the mean-variance optimal investment strategies in a two-factor model with stochastic interest rates and mean-reverting equity returns, with the constraint that the strategies are allowed to depend on time only. Such deterministic strategies are closer to practical use than their state-dependent counterparts typically considered in the literature. We have used techniques from calculus of variations to obtain explicit solutions for optimal rate and equity strategies separately, and discussed how to combine these strategies. For equity strategies a complete solution has been provided, while optimal rate strategies have been derived under the simplifying assumption of constant market price of interest rate risk. We have also provided illustrations of the results, and discussed the calibration of the mean-reversion component in detail.

Assuming independent risk factors, the distribution of the portfolio on the horizon can be succinctly represented as
\begin{equation}
   V_T \stackrel{\cal{D}}{=} \frac{V_0}{p_0(T)}Y_T Z_T,
\end{equation}
where $Y_T$ and $Z_T$ are independent, log-normally distributed random variates depending on the rate and equity strategies, respectively. Holding only $T$-bonds and no equity risk corresponds to $Y_T=Z_T=1$ and thereby no risk on the horizon. All rate and equity strategies, or combination of these, deviating from this hedging strategy, entail a non-trivial risk-reward trade-off captured by the distribution of the stochastic multipliers, $Y_T$ and $Z_T$. Appendix~\ref{app:statistics} contains tables with risk and reward statistics that can be used to assess this trade-off. The strategies are parameterized by $\nu$, which can be interpreted as a risk-aversion parameter, with $\nu=0$ corresponding to the strategy with maximal log-mean.

We conclude from Appendix~\ref{app:statistics} that on long horizons, say, 30 years and above there exist equity strategies with substantial upside and very little risk of accumulated losses ($Z_T<1$). When mean-reversion is high the strategies are essentially risk free (Table~\ref{tab:equityStatisticsHigh}), but also for moderate levels of mean-reversion the risk is low for all but the most aggressive strategies (Table~\ref{tab:equityStatisticsMod}). In other words, on long horizons, (optimal) equity strategies can function as return generating overlays with very low risk of performing below the risk-free, hedging strategy.

The optimal equity strategies are typically decreasing over time. This conforms with our intuitive understanding of the effect of mean-reversion. Essentially, "early" exposure is better rewarded than "late" exposure, since early losses are more likely to be (partly) compensated by subsequent excess returns than losses towards the end of the period. However, we have demonstrated numerically, that for moderate levels of mean-reversion constant strategies perform close to optimally. Indeed, under no mean-reversion, i.e., the Black-Scholes model, constant strategies are optimal. Conversely, with high levels of mean-reversion the optimal strategies outperform constant strategies significantly, but they rely on rather extreme glidepaths which may be hard to implement in practice. 


\pagebreak
\appendix
\section{Optimal equity strategies with positive mean-variance  trade-off} \label{app:coefmatch}
In this appendix we derive the (unique) solution to (\ref{eq:OptEquityIntegral}) of Lemma~\ref{lemma:OptEquityIntegral} for $\nu<0$. These strategies correspond to the optimal equity strategies with positive mean-variance trade-off, as discussed in the main text. We assume throughout that $\sigma_x>0$ and $\sigma_S>0$, while $x_0$ and $\bar{x}$ can have any real value.

It follows from Lemma~\ref{lemma:OptEquityODE} that for $\nu<0$ the solution to (\ref{eq:OptEquityIntegral}), if it exists, must be of form
\begin{equation} \label{eq:OptimalEquityAnsatz}
     f_s = b_0 + b_1 e^{c_1 s} + b_2 e^{c_2 s} \quad (0 \leq s \leq T).
\end{equation}
The plan therefore is to rewrite (\ref{eq:OptEquityIntegral}) using this form of $f$, and from this derive the values of the five constants. We first carry out this programme, assuming $\alpha \neq 0$ and $\alpha \neq \frac{1}{2}\frac{\sigma_x}{\sigma_S}$. The two special cases for $\alpha$ are handled separately afterwards. Proof of uniqueness of the stated solutions is given at the end of the appendix. We begin with the following lemma, the conclusion of which we will need throughout.

\begin{lemma} \label{lemma:DistinctRoots}
Assume $\nu < 0$ and $\alpha \neq \frac{1}{2}\frac{\sigma_x}{\sigma_S}$. Let $c_1=\sqrt{-C/A}$, where $A = 1-2\nu$ and $C  = 2\nu(\alpha-\frac{\sigma_x}{\sigma_S})^2 -\alpha^2$.
Then, $c_1 >0$ and $c_1 \neq|\alpha|$. Further, the equation
\begin{equation}
   \left(1+\frac{\sigma_x}{\sigma_S(c-\alpha)}\right)\left(1 -\frac{\sigma_x}{\sigma_S(c+\alpha)}\right) = \frac{1}{2\nu} \label{eq:divergingQuadratic}
\end{equation}
has two distinct solutions given by $c=\pm c_1$.
\end{lemma}
\begin{proof}
In the expression for $C$, at least one of the terms $(\alpha-\frac{\sigma_x}{\sigma_S})^2$ and $\alpha^2$ will be strictly positive, with coefficients $2\nu$ and $-1$, respectively, which are both strictly negative. This implies $C<0$, and since $A>1$ we have $-C/A > 0$, and thereby also $c_1 >0$.

For the second statement, we first note that the assumption $\alpha\neq\frac{1}{2}\frac{\sigma_x}{\sigma_S}$ implies that $\alpha^2 \neq (\alpha-\frac{\sigma_x}{\sigma_S})^2$. Now, consider
\begin{equation}
 c_1^2 = -C/A = p(\nu)\alpha^2 + (1-p(\nu))\left(\alpha-\frac{\sigma_x}{\sigma_S}\right)^2,   \label{eq:gamma2expectation}
\end{equation}
where $p(\nu)=1/(1-2\nu)$ with $0 < p(\nu) < 1$. Equation (\ref{eq:gamma2expectation}) expresses $c_1^2$ as a strict convex combination of two distinct values. In particular, $c_1^2\neq\alpha^2$, which is equivalent to $c_1\neq|\alpha|$.

Considering equation (\ref{eq:divergingQuadratic}) we note that the right-hand side is finite, while the left-hand side diverges for $c$ approaching $\pm\alpha$. In the search for a solution we can therefore assume $c \neq \pm\alpha$. Under this assumption and introducing the short-hand notation $R=\sigma_x/\sigma_S>0$, we can rewrite the left-hand side of (\ref{eq:divergingQuadratic}) as follows
\begin{equation}
   \left(1+\frac{R}{c-\alpha}\right)\left(1 -\frac{R}{c+\alpha}\right) = \frac{c-(\alpha - R)}{c-\alpha}\frac{c+(\alpha - R)}{c+\alpha} = \frac{c^2-(\alpha - R)^2}{c^2-\alpha^2}.
\end{equation}
Equating this to the right-hand side of (\ref{eq:divergingQuadratic}) and rearranging terms, we see that this is equivalent to $Ac^2+C=0$ with solution $c=\pm\sqrt{-C/A}=\pm c_1$. Finally, since $c_1 > 0$ the solutions are distinct.
\end{proof}

The next two lemmas establish the (general) form of $h$ and $\int h$ as defined in Lemma~\ref{lemma:OptEquityIntegral}.

\begin{lemma} \label{lemma:GeneralFormOfh}
Assume $f$ is of form (\ref{eq:OptimalEquityAnsatz}) with $c_1\neq \alpha$, $c_2\neq\alpha$, and $\alpha\neq 0$. For $0\leq u\leq T$,
\begin{equation}
  h_u \equiv f_u - \frac{\sigma_x}{\sigma_S}\int_u^T f_s e^{-\alpha(s-u)}ds =  d_0 + d_1 e^{c_1 u} + d_2 e^{c_2 u} + d_3 e^{\alpha u} ,
\end{equation}
where
\begin{align}
  d_0 & = b_0\left[ 1 - \frac{\sigma_x}{\sigma_S\alpha}\right], \quad d_1 = b_1\left[1 + \frac{\sigma_x}{\sigma_S(c_1-\alpha)}\right], \quad d_2 = b_2\left[1 + \frac{\sigma_x}{\sigma_S(c_2-\alpha)}\right], \\[2mm]
  d_3 & = e^{-\alpha T}\frac{\sigma_x}{\sigma_S}\left[\frac{b_0}{\alpha} - \frac{b_1}{c_1-\alpha}e^{c_1 T} - \frac{b_2}{c_2-\alpha}e^{c_2 T} \right].
\end{align}
\end{lemma}
\begin{proof}
This follows from inserting $f$ and evaluating integrals of the form $\int_u^T e^{k s}ds$ for $k=\alpha$, $c_1-\alpha$, and $c_2-\alpha$. By assumption, all three exponents are different from $0$, and the integral evaluates to $\frac{1}{k}\left[e^{k s}\right]_u^T = \frac{1}{k}(e^{kT} - e^{ku})$. We get,
\begin{align*}
  h_u =\ & b_0 + b_1 e^{c_1 u} + b_2 e^{c_2 u} - \\[2mm]
        & \frac{\sigma_x}{\sigma_S}e^{\alpha u}\left(-\frac{b_0}{\alpha}\left[e^{-\alpha s}\right]_u^T + \frac{b_1}{c_1-\alpha}\left[e^{(c_1-\alpha)s}\right]_u^T + \frac{b_2}{c_2-\alpha}\left[e^{(c_2-\alpha)s}\right]_u^T\right),
\end{align*}
from which the result follows after expanding and collecting terms.
\end{proof}

\begin{lemma} \label{lemma:GeneralFormOfInth}
Assume $f$ is of form (\ref{eq:OptimalEquityAnsatz}) with $c_1\neq \pm\alpha$, $c_2\neq \pm\alpha$, and $\alpha\neq 0$. Let $h_u$ and $d_0$, $d_1$, $d_2$, and $d_3$ be defined as in Lemma~\ref{lemma:GeneralFormOfh}. For $0\leq s\leq T$,
\begin{equation}
    \frac{\sigma_x}{\sigma_S}\int_0^s h_u e^{-\alpha(s-u)}du = a_0 + a_1 e^{c_1 s} + a_2 e^{c_2 s} + a_3 e^{\alpha s} + a_4 e^{-\alpha s},
\end{equation}
where
\begin{align}
  a_0 & = \frac{\sigma_x}{\sigma_S}\frac{d_0}{\alpha}, \quad
  a_1 = \frac{\sigma_x}{\sigma_S}\frac{d_1}{c_1+\alpha}, \quad
  a_2 = \frac{\sigma_x}{\sigma_S}\frac{d_2}{c_2+\alpha}, \quad
  a_3 = \frac{\sigma_x}{\sigma_S}\frac{d_3}{2\alpha},  \\[2mm]
   \quad a_4 & = -\frac{\sigma_x}{\sigma_S}\left[\frac{d_0}{\alpha} + \frac{d_1}{c_1+\alpha} + \frac{d_2}{c_2+\alpha} + \frac{d_3}{2\alpha} \right].
\end{align}
\end{lemma}
\begin{proof}
This follows by the same argument as in the proof of Lemma~\ref{lemma:GeneralFormOfh} using that, by assumption, all exponents are different from $0$. We leave the details to the reader.
\end{proof}

Note, that the assumption of Lemma~\ref{lemma:GeneralFormOfh} is slightly weaker than in Lemma~\ref{lemma:GeneralFormOfInth}. In the former, we need to ensure that the exponents of $f$ are different from $\alpha$, while in the latter we need to ensure that the exponents of $f$ are also different from $-\alpha$. Both assumptions are typically satisfied, except in certain special cases.

\begin{proposition} \label{prop:verifyGenCase}
Assume $\nu<0$, $\alpha\neq 0$, and $\alpha\neq \frac{1}{2}\frac{\sigma_x}{\sigma_S}$. Condition (\ref{eq:OptEquityIntegral}) of Lemma~\ref{lemma:OptEquityIntegral} is satisfied for
\begin{equation} \label{eq:GeneralSolutionFullVersion}
     f_s = b_0 + b_1 e^{c_1 s} + b_2 e^{c_2 s} \quad (0 \leq s \leq T),
\end{equation}
where $c_1 = \sqrt{-C/A} > 0$, $c_2=-\sqrt{-C/A}$, $b_0=-D/C$ with $A = 1-2\nu$, $C  = 2\nu(\alpha-\frac{\sigma_x}{\sigma_S})^2 -\alpha^2$, and $D=\alpha^2 \bar{x}/\sigma_S$. Further,
the exponents satisfy $c_1 = |c_2| \neq |\alpha|$, and
\begin{equation}  \label{eq:b1b2Solution}
   \begin{pmatrix}
    b_1 \\
    b_2
  \end{pmatrix}
  =
  \begin{pmatrix}
    \frac{e^{c_1 T}}{c_1-\alpha} & \frac{e^{c_2 T}}{c_2-\alpha} \\[2mm]
    \frac{\sigma_x}{\sigma_S(c_1+\alpha)-\sigma_x} & \frac{\sigma_x}{\sigma_S(c_2+\alpha)-\sigma_x}
  \end{pmatrix}^{-1}
  \begin{pmatrix}
    \frac{\alpha\bar{x}}{\sigma_S\left(\alpha^2 - 2\nu(\alpha-\sigma_x/\sigma_S)^2 \right)} \\[2mm]
    -\frac{x_0}{\sigma_S} + \frac{\alpha\bar{x}}{\sigma_S}\frac{\alpha-2\nu(\alpha-\sigma_x/\sigma_S)}{\alpha^2-2\nu(\alpha-\sigma_x/\sigma_S)^2}
  \end{pmatrix}.
\end{equation}
\end{proposition}
\begin{proof}
By Lemma~\ref{lemma:DistinctRoots} we have that $c_1\neq \pm\alpha$, and thereby also $c_2\neq \pm\alpha$. The assumptions of Lemmas~\ref{lemma:GeneralFormOfh} and~\ref{lemma:GeneralFormOfInth} are thereby satisfied, and we can use the conclusions of these lemmas to recast the left-hand side of condition (\ref{eq:OptEquityIntegral}) as
\begin{align}
  \xi_s - f_s & + 2\nu h_s - 2\nu \frac{\sigma_x}{\sigma_S}\int_0^s h_u e^{-\alpha(s-u)}du \nonumber \\[2mm]
    =\ & \left[\bar{x} + e^{-\alpha s}(x_0-\bar{x})\right]/\sigma_S + 2\nu\left(d_0 + d_1 e^{c_1 s} + d_2 e^{c_2 s} + d_3 e^{\alpha s} \right) - \nonumber \\[2mm]
     & b_0 - b_1e^{c_1 s} - b_2e^{c_2s} - 2\nu\left(a_0 + a_1 e^{c_1 s} + a_2 e^{c_2 s} + a_3 e^{\alpha s} + a_4 e^{-\alpha s} \right) \nonumber \\[2mm]
    =\ & \left[ \bar{x}/\sigma_S-b_0 + 2\nu(d_0-a_0)\right] + \label{eq:DefiningEquationI}  \\[2mm]
     & \left[-b_1 + 2\nu(d_1-a_1) \right]e^{c_1 s} + \left[-b_2 + 2\nu(d_2-a_2) \right]e^{c_2 s} + \label{eq:DefiningEquationII} \\[2mm]
     & \left[ 2\nu(d_3-a_3) \right]e^{\alpha s} + \left[ (x_0-\bar{x})/\sigma_S - 2\nu a_4 \right]e^{-\alpha s}. \label{eq:DefiningEquationIII}
\end{align}
Condition (\ref{eq:OptEquityIntegral}) states that the expression above should equal $0$ for all $0\leq s \leq T$. Since the exponents in (\ref{eq:DefiningEquationI})--(\ref{eq:DefiningEquationIII}) are all non-zero and distinct, this is satisfied if (and only if) the five expressions in square brackets in (\ref{eq:DefiningEquationI})--(\ref{eq:DefiningEquationIII}) are all $0$. It remains to be shown that this is fulfilled if (and only if) the parameters take the stated values. For the current result we only need the "if"-part, but for the later uniqueness result we also need the "only if"-part; to facilitate this argument we show slightly more than needed below.

By Lemmas~\ref{lemma:GeneralFormOfh} and~\ref{lemma:GeneralFormOfInth}, we have that the expression in (\ref{eq:DefiningEquationI}) is $0$ if and only if
\begin{equation}
   b_0\left[2\nu\left(1- \frac{\sigma_x}{\sigma_S\alpha}\right)^2 - 1\right] = - \frac{\bar{x}}{\sigma_S}, \label{eq:b0defining}
\end{equation}
and that the two expressions in square brackets in (\ref{eq:DefiningEquationII}) are zero if and only if
\begin{align}
   b_1 & = b_1 2\nu\left(1+\frac{\sigma_x}{\sigma_S(c_1-\alpha)}\right)\left(1-\frac{\sigma_x}{\sigma_S(c_1+\alpha)}\right), \label{eq:c1defining} \\[2mm]
   b_2 & = b_2 2\nu\left(1+\frac{\sigma_x}{\sigma_S(c_2-\alpha)}\right)\left(1-\frac{\sigma_x}{\sigma_S(c_2+\alpha)}\right). \label{eq:c2defining}
\end{align}
Further, assuming (\ref{eq:c1defining})--(\ref{eq:c2defining}) hold and using $\alpha\neq \frac{1}{2}\frac{\sigma_x}{\sigma_S}$, the two expressions in square brackets in (\ref{eq:DefiningEquationIII}) are zero if and only if
\begin{equation} \label{eq:b1b2defining}
    \begin{pmatrix}
    \frac{e^{c_1 T}}{c_1-\alpha} & \frac{e^{c_2 T}}{c_2-\alpha} \\[2mm]
    \frac{\sigma_x}{\sigma_S(c_1+\alpha)-\sigma_x} & \frac{\sigma_x}{\sigma_S(c_2+\alpha)-\sigma_x}
  \end{pmatrix}
   \begin{pmatrix}
    b_1 \\
    b_2
  \end{pmatrix}
  =
   \begin{pmatrix}
    \frac{b_0}{\alpha} \\
    -\frac{x_0}{\sigma_S} + \frac{\bar{x}}{\sigma_S} - 2\nu b_0 \frac{\sigma_x}{\sigma_S\alpha}\left(1-\frac{\sigma_x}{\sigma_S\alpha}\right)
  \end{pmatrix}.
\end{equation}
To reach the desired conclusion, we need to verify that (\ref{eq:b0defining})--(\ref{eq:b1b2defining}) all hold. Clearly, $b_0=-D/C$ satisfies (\ref{eq:b0defining}), while (\ref{eq:c1defining}) and (\ref{eq:c2defining}) are satisfied by Lemma~\ref{lemma:DistinctRoots}. Finally, with $b_0=-D/C$, the right-hand side of (\ref{eq:b1b2defining}) is equivalent to the rightmost column of (\ref{eq:b1b2Solution}), and it follows that (\ref{eq:b1b2defining}) is also satisfied.
\end{proof}

Proposition~\ref{prop:verifyGenCase} leaves two cases to be handled separately. These are covered by the following propositions.  The line of argument is the same, but the specific calculations differ slightly from those in the proof of Proposition~\ref{prop:verifyGenCase}.

\begin{proposition} \label{prop:verifyAlphaZero}
Assume $\nu<0$ and $\alpha=0$. Condition (\ref{eq:OptEquityIntegral}) of Lemma~\ref{lemma:OptEquityIntegral} is satisfied for
\begin{equation} \label{eq:GeneralSolutionAlphaZero}
     f_s = b_1 e^{c_1 s} + b_2 e^{c_2 s} \quad (0 \leq s \leq T),
\end{equation}
where $c_1 = \sqrt{-C/A} > 0$, $c_2=-\sqrt{-C/A}$ with $A = 1-2\nu$, and $C  = 2\nu(\frac{\sigma_x}{\sigma_S})^2$. Further,
the exponents satisfy $c_1 = |c_2| \neq 0$, and
\begin{equation}  \label{eq:b1b2SolutionAlphaZero}
   \begin{pmatrix}
    b_1 \\
    b_2
  \end{pmatrix}
  =
  \begin{pmatrix}
    \frac{e^{c_1 T}}{c_1} & \frac{e^{c_2 T}}{c_2} \\[2mm]
    \frac{\sigma_x}{\sigma_S c_1 -\sigma_x} & \frac{\sigma_x}{\sigma_S c_2-\sigma_x}
  \end{pmatrix}^{-1}
  \begin{pmatrix}
    0 \\[2mm]
    -\frac{x_0}{\sigma_S}
  \end{pmatrix}.
\end{equation}
\end{proposition}
\begin{proof} We will only give a sketch of the proof. By computations similar to those of Lemmas~\ref{lemma:GeneralFormOfh} and~\ref{lemma:GeneralFormOfInth} we first obtain expressions for $h$ and $\int h$; the computations are left to the reader. Using these expressions we recast the left-hand side of condition (\ref{eq:OptEquityIntegral}) as
\begin{align}
  \xi_s - f_s & + 2\nu h_s - 2\nu \frac{\sigma_x}{\sigma_S}\int_0^s h_u du \nonumber \\[2mm]
    =\ & x_0/\sigma_S - b_1e^{c_1 s} - b_2e^{c_2s} + 2\nu\left( d_1 e^{c_1 s} + d_2 e^{c_2 s} + d_3 \right) - \nonumber \\[2mm]
     & 2\nu\left(a_1 e^{c_1 s} + a_2 e^{c_2 s} + a_3 s - a_1 - a_2 \right) \nonumber \\[2mm]
    =\ & \left[-b_1 + 2\nu(d_1-a_1) \right]e^{c_1 s} + \left[-b_2 + 2\nu(d_2-a_2) \right]e^{c_2 s} + \label{eq:DefiningEquationAlphaZeroI} \\[2mm]
       & \left[ - 2\nu a_3 \right]s + \left[ x_0/\sigma_S + 2\nu (d_3+a_1+a_2) \right], \label{eq:DefiningEquationAlphaZeroII}
\end{align}
where
\begin{align*}
  d_1 & = b_1\left[1+\frac{\sigma_x}{\sigma_S c_1}\right], d_2 = b_2\left[1+\frac{\sigma_x}{\sigma_S c_2}\right], d_3=-\frac{\sigma_x}{\sigma_S}\left[\frac{b_1}{c_1}e^{c_1 T}+\frac{b_2}{c_2}e^{c_2 T}\right], \\[2mm]
  a_1 & = \frac{\sigma_x}{\sigma_S}\frac{d_1}{c_1}, \quad  a_2 = \frac{\sigma_x}{\sigma_S}\frac{d_2}{c_2}, \quad  a_3 = \frac{\sigma_x}{\sigma_S}d_3.
\end{align*}
As in the proof of Proposition~\ref{prop:verifyGenCase}, it follows that condition (\ref{eq:OptEquityIntegral}) is satisfied if and only if the four terms in square brackets in (\ref{eq:DefiningEquationAlphaZeroI}) and (\ref{eq:DefiningEquationAlphaZeroII}) are zero. This in turn is equivalent to the following set of conditions,
\begin{align}
  b_1 & = b_1 2\nu\left(1+\frac{\sigma_x}{\sigma_Sc_1}\right)\left(1 -\frac{\sigma_x}{\sigma_Sc_1}\right),\ b_2 = b_2 2\nu\left(1+\frac{\sigma_x}{\sigma_Sc_2}\right)\left(1 -\frac{\sigma_x}{\sigma_Sc_2}\right), \label{eq:c1c2definingAlphaZero} \\[2mm]
  0 & = \frac{b_1}{c_1}e^{c_1 T}+\frac{b_2}{c_2}e^{c_2 T}, \quad -\frac{x_0}{\sigma_S} = \frac{b_1\sigma_x}{\sigma_S c_1-\sigma_x} + \frac{b_2\sigma_x}{\sigma_S c_2-\sigma_x}.
  \label{eq:b1b2definingAlphaZero}
\end{align}
Since $\alpha=0 \neq \frac{1}{2}\frac{\sigma_x}{\sigma_S}$, we can use Lemma~\ref{lemma:DistinctRoots} to conclude that (\ref{eq:c1c2definingAlphaZero}) is satisfied, while it follows from (\ref{eq:b1b2SolutionAlphaZero}) that (\ref{eq:b1b2definingAlphaZero}) is satisfied.
\end{proof}

Note that the solution stated in Proposition~\ref{prop:verifyAlphaZero} is in fact the same as in Proposition~\ref{prop:verifyGenCase} with $\alpha=0$; it is only the argument for verifying the solution that differs. This is due to the fact that in both cases the exponents, $c_1$ and $c_2$, differ from $\pm\alpha$. For the other special case, $\alpha=\frac{1}{2}\frac{\sigma_x}{\sigma_S}$, the exponents equal $\pm\alpha$ which change the conditions somewhat.

\begin{proposition} \label{prop:verifyAlphaHalfRatio}
Assume $\nu<0$ and $\alpha=\frac{1}{2}\frac{\sigma_x}{\sigma_S}$. Condition (\ref{eq:OptEquityIntegral}) of Lemma~\ref{lemma:OptEquityIntegral} is satisfied for
\begin{equation} \label{eq:GeneralSolutionHalfRatio}
     f_s = b_0 + b_2 e^{-\alpha s} \quad (0 \leq s \leq T),
\end{equation}
where
\begin{equation}  \label{eq:bSolutionAlphaHalfRatio}
   b_0 =  \frac{\bar{x}}{\sigma_S(1-2\nu)}, \quad
   b_2 = \frac{-\frac{x_0}{\sigma_S}+ \frac{\bar{x}}{\sigma_S}\left(1-\frac{4\nu}{1-2\nu}\left(e^{-\alpha T}-1\right)\right)}{2\nu e^{-2\alpha T}-1}.
\end{equation}
\end{proposition}
\begin{proof} For use in the later uniqueness result, we will prove a bit more than needed to establish the current result. Consider $f$ of the form
\begin{equation} \label{eq:GeneralSolutionHalfRatioExtra}
  f_s = b_0 + b_1 e^{\alpha s} + b_2 e^{-\alpha s} \quad (0 \leq s \leq T).
\end{equation}
By computations similar to those of Lemmas~\ref{lemma:GeneralFormOfh} and~\ref{lemma:GeneralFormOfInth} we recast the left-hand side of condition (\ref{eq:OptEquityIntegral}) as
\begin{align}
  \xi_s - f_s & + 2\nu h_s - 2\nu \frac{\sigma_x}{\sigma_S}\int_0^s h_u e^{-\alpha(s-u)}du \nonumber \\[2mm]
    =\ & \left[\bar{x} + e^{-\alpha s}(x_0-\bar{x})\right]/\sigma_S + 2\nu\left(-b_0 + d_1 e^{\alpha s} + d_2 e^{\alpha s}s \right) - \nonumber \\[2mm]
     & b_0 - b_1e^{\alpha s} - b_2e^{-\alpha s} - 2\nu\left(-2 b_0 + a_1 e^{\alpha s} + a_2 e^{\alpha s}s + a_3 e^{-\alpha s} \right) \nonumber \\[2mm]
    =\ & \left[ \bar{x}/\sigma_S -b_0 + 2\nu b_0 \right] + \left[-b_1 + 2\nu(d_1-a_1) \right]e^{\alpha s} +  \label{eq:DefiningEquationAlphaHalfRatioI}  \\[2mm]
       & \left[ (x_0-\bar{x})/\sigma_S - b_2 - 2\nu a_3 \right]e^{-\alpha s}. \label{eq:DefiningEquationAlphaHalfRatioII}
\end{align}
where
\begin{align*}
  d_1 & = b_1 - \frac{\sigma_x}{\sigma_S} b_1 T + 2b_0e^{-\alpha T} + b_2e^{-2\alpha T}, \quad d_2 = \frac{\sigma_x}{\sigma_S} b_1, \\[2mm]
  a_1 & = d_1 - b_1, \quad a_2=d_2, \quad a_3 = 2b_0 - d_1 + b_1.
\end{align*}
Note that, compared to previous calculations there is apparently missing factors of $\sigma_x/\sigma_S$ in the $a$-coefficients. These, however, are being absorbed by the (also missing) factors proportional to $1/\alpha$.

As in the proof of Proposition~\ref{prop:verifyGenCase}, it follows that condition (\ref{eq:OptEquityIntegral}) is satisfied if and only if the three terms in square brackets in (\ref{eq:DefiningEquationAlphaHalfRatioI}) and (\ref{eq:DefiningEquationAlphaHalfRatioII}) are zero. This in turn is equivalent to the following set of conditions,
\begin{equation}
b_0 = \frac{\bar{x}}{\sigma_S(1-2\nu)}, \quad b_1=0, \quad 4\nu b_0\left(e^{-\alpha T}-1\right) + b_2\left(2\nu e^{-2\alpha T}-1\right) = -\frac{x_0-\bar{x}}{\sigma_S},
\end{equation}
from which the result follows.
\end{proof}

Collectively, Propositions~\ref{prop:verifyGenCase}, \ref{prop:verifyAlphaZero}, and~\ref{prop:verifyAlphaHalfRatio} provide a solution to condition (\ref{eq:OptEquityIntegral}) of Lemma~\ref{lemma:OptEquityIntegral} for $\nu<0$. The proofs of the propositions consist of deriving an equivalent set of conditions characterizing the constants in (\ref{eq:OptimalEquityAnsatz}). Since we know that $f$ is of form (\ref{eq:OptimalEquityAnsatz}), it might seem as if we have also shown uniqueness. However, there is a caveat. In deriving the equivalent set of conditions we make {\em a priori} assumptions on the value of the exponents, $c_1$ and $c_2$, and of the constant $b_0$ (in Proposition~\ref{prop:verifyAlphaZero}). To prove uniqueness we also need to check that condition (\ref{eq:OptEquityIntegral}) cannot be satisfied if $c_1$ and $c_2$ take values initially ruled out by assumption. This amounts to deriving alternative sets of equivalent conditions and showing that these cannot be satisfied. In principle, this is possible, but it is rather laborious. Fortunately, there is a shortcut to this brute force approach.

\begin{proposition} \label{prop:uniqueness}
Assume $\nu < 0$. The solution to (\ref{eq:OptEquityIntegral}) of Lemma~\ref{lemma:OptEquityIntegral} provided by Propositions~\ref{prop:verifyGenCase}, \ref{prop:verifyAlphaZero}, and~\ref{prop:verifyAlphaHalfRatio} is unique.
\end{proposition}
\begin{proof}
For $\nu<0$, it follows from Lemma~\ref{lemma:OptEquityODE} that $f$ is of form (\ref{eq:OptimalEquityAnsatz}) with $b_0=-D/C$. Moreover, it follows from general theory on ordinary differential equations (ODE) that the exponents are given by $\pm\sqrt{-C/A}$, since $-C/A>0$.

Revisiting the proof of Proposition~\ref{prop:verifyGenCase}, the initial assumptions on $c_1$ and $c_2$, and the recasting of condition (\ref{eq:OptEquityIntegral}) are thus justified, not by assumption, but as a consequence of Lemma~\ref{lemma:OptEquityODE}. A fortiori, we can restrict attention to candidate solutions of form (\ref{eq:OptimalEquityAnsatz}) with $c_1=\sqrt{-C/A}$, $c_2=-\sqrt{C/A}$, and $b_0=-D/C$, and for these candidate solutions the derived conditions are valid. In fact, conditions (\ref{eq:b0defining})--(\ref{eq:c2defining}) are automatically satisfied, and a candidate solution is therefore an actual solution if and only if it satisfies condition (\ref{eq:b1b2defining}). Finally, since (\ref{eq:b1b2defining}) has a unique solution, the claimed uniqueness follows, in the case $\alpha\neq 0$ and $\alpha\neq \frac{1}{2}\frac{\sigma_x}{\sigma_S}$. The same argument applies to the proofs of Propositions~\ref{prop:verifyAlphaZero} and~\ref{prop:verifyAlphaHalfRatio} showing uniqueness also for the cases $\alpha=0$, or $\alpha = \frac{1}{2}\frac{\sigma_x}{\sigma_S}$.
\end{proof}


\pagebreak
\section{Overview of extremal equity strategies} \label{app:overviewExtremal}
To allow a full discussion of the optimization problem, we give here an overview of the complete set of extremal equity strategies. Appendix~\ref{app:coefmatch} covers in detail the case of prime interest, namely the strategies with positive mean-variance trade-off. In this appendix we supplement these results with the extremal equity strategies corresponding to a negative mean-variance trade-off, $\nu>0$. We provide only the results and a few intermediate calculations, which can serve as stepping stones for the dedicated reader interested in reproducing the results. We assume throughout that $\sigma_x>0$ and $\sigma_S>0$, while $x_0$ and $\bar{x}$ can have any real value.

From Lemma~\ref{lemma:OptEquityODE} we know that the extremal equity strategies satisfy the nonhomogeneous, second-order differential equation
\begin{equation}\label{eq:OptEquityODEapp}
    A f''_s + C f_s + D = 0 \quad (0 \leq s \leq T),
\end{equation}
where
\begin{equation}
  A  = 1-2\nu, \quad
  C  = 2\nu\left(\alpha-\frac{\sigma_x}{\sigma_S}\right)^2 -\alpha^2,  \quad
  D  = \alpha^2 \frac{\bar{x}}{\sigma_S}.  \label{eq:CoefODEapp}
\end{equation}
There are three different cases, depending on the type of roots to the characteristic polynomial, $Ar^2+C$, of the associated homogeneous differential equation, $A f''_s + C f_s= 0$. The theorem below, which we state without proof, summarizes the results.

\begin{theorem} \label{thm:ExtremalOverview}
Assume independent risk factors ($\rho=0$). For given horizon $T > 0$, the extremal equity strategies are the solutions to condition (\ref{eq:OptEquityIntegral}) of Lemma~\ref{lemma:OptEquityIntegral}. If it exists, the extremal equity strategy, $f$, takes one of three forms.
\begin{description}
  \item[I. Real, distinct roots] If $C<0<A$, or $A<0<C$, then
  \begin{equation}
     f_s = b_0 + b_1e^{c_1 s} + b_2e^{c_2 s}.  \label{eq:solTypeI}
  \end{equation}
  For $\nu=0$, $f_s= \bar{x}/\sigma_S + e^{-\alpha s}(x_0-\bar{x})/\sigma_S$, otherwise $b_0=-D/C$, $c_1=\sqrt{-C/A}$, $c_2=-\sqrt{-C/A}$, while $b_1$ and $b_2$ are given by (\ref{eq:OptimalEquityGenB}) and (\ref{eq:OptimalEquitySpecialB}) of Theorem~\ref{thm:OptimalEquity}.
  \item[II. Complex roots] If $C<0$ and $A<0$, or $C>0$ and $A>0$, then
  \begin{equation}
     f_s = b_0 + b_1 \sin(c s) + b_2 \cos(c s), \label{eq:solTypeII}
  \end{equation}
  where $b_0=-D/C$, and $c=\sqrt{C/A}$. The values of $b_1$ and $b_2$ are given by (\ref{eq:TrigonometricDefb1b2}) of Proposition~\ref{prop:solTrigometric}.
  \item[III. Double roots] If $C=0$ and $A\neq 0$, then
  \begin{equation}
     f_s = b_0 + b_1 s + b_2 s^2, \label{eq:solTypeIII}
  \end{equation}
  where $b_2=-D/(2A)$. The values of $b_0$ and $b_1$ are given by (\ref{eq:QuadraticDefb0b1}) of Proposition~\ref{prop:solQuadratic}. In the special case $\nu=\alpha=0$, $f_s=x_0/\sigma_S$.
\end{description}
\end{theorem}

The theorem covers all situations, except $\nu=1/2$ ($A=0$), where generally there is no corresponding extremal strategy, i.e., no solution to (\ref{eq:OptEquityIntegral}) of Lemma~\ref{lemma:OptEquityIntegral}.\footnote{For $\nu=1/2$, the solution to (\ref{eq:OptEquityODEapp}) is $f_s=-D/C$, i.e., a constant not depending on $x_0$, nor on $T$. However, in general this constant cannot solve (\ref{eq:OptEquityIntegral}), since the initial term of (\ref{eq:OptEquityIntegral}) is $\xi_s=[\bar{x}+e^{-\alpha x}(x_0-\bar{x})]/\sigma_S$, while the remaining terms are independent of $x_0$. Thus, apart from very specific cases, e.g.\ $x_0=\bar{x}=0$, there is no solution to (\ref{eq:OptEquityIntegral}) for $\nu=1/2$.} For given model parameters, there are also other "singular" values of $\nu$ for which there exists no corresponding extremal strategy. However, providing a full description of necessary and sufficient conditions is outside the scope of this paper. Except in special cases, Theorem~\ref{thm:ExtremalOverview} gives the (unique) extremal equity strategy for given model parameters and Lagrange multiplier, $\nu$.

Noting that
\begin{equation}
   C= \left(\frac{\sigma_x}{\sigma_S}\right)^2\left[2\nu(\tilde{\alpha}-1)^2- \tilde{\alpha}^2 \right],
\end{equation}
where $\tilde{\alpha} = \alpha/[\sigma_x/\sigma_S]$, we see that the sign of $C$ depends only on $\nu$ and $\tilde{\alpha}$. This observation makes it possible to visualize the domains corresponding to the three cases, cf.~Figure~\ref{fig:domain}. The solid black line is given by $\tilde{\alpha}=[4\nu\pm\sqrt{8\nu}]/[2(2\nu-1)]$ for $\nu\geq 0$; on this line $C=0$ and the solution is of type III. As seen previously, the strategy is always of type I for $\nu<0$, and this is also the case for $\nu=0$ (except for the special case $\nu=\alpha=0$). For $\nu>0$, the situation is more complex with all three solution types occurring.

\begin{figure}[h]
\begin{center}
\includegraphics[height=7.5cm]{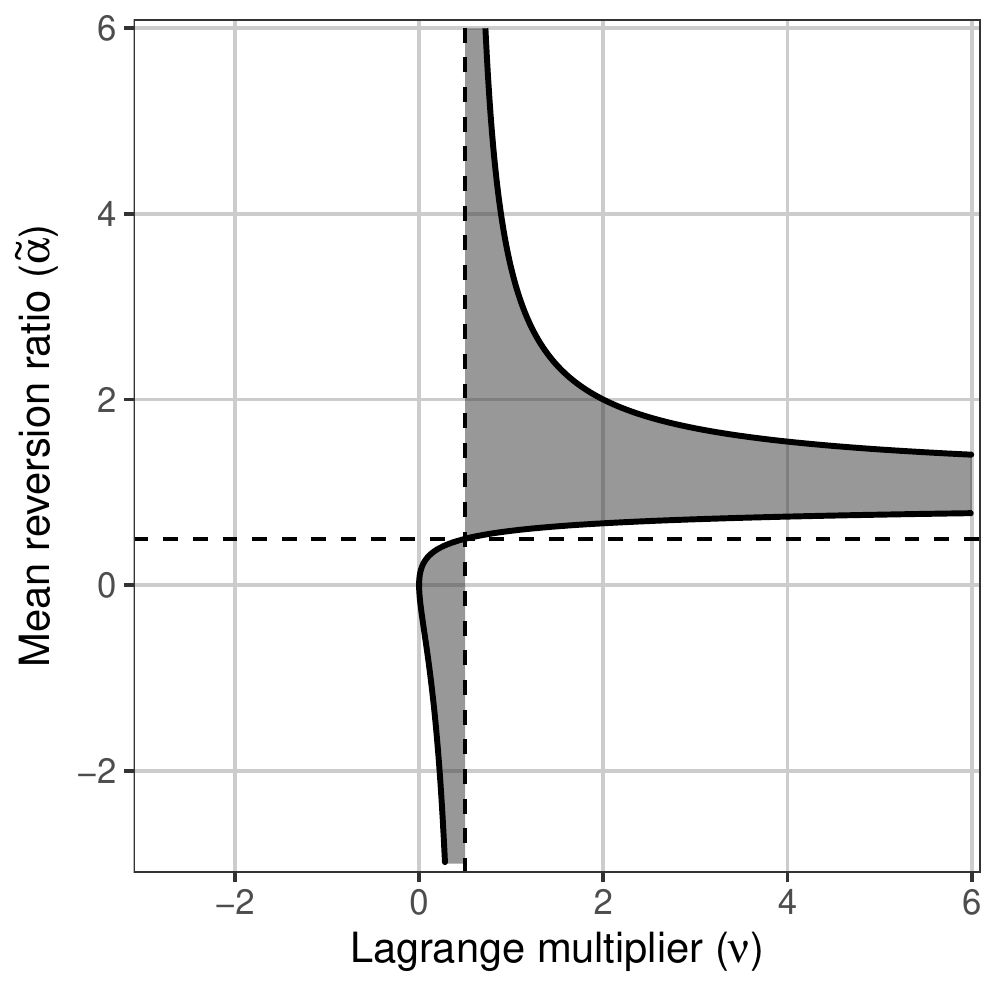}   
\hfill
\includegraphics[height=7.5cm]{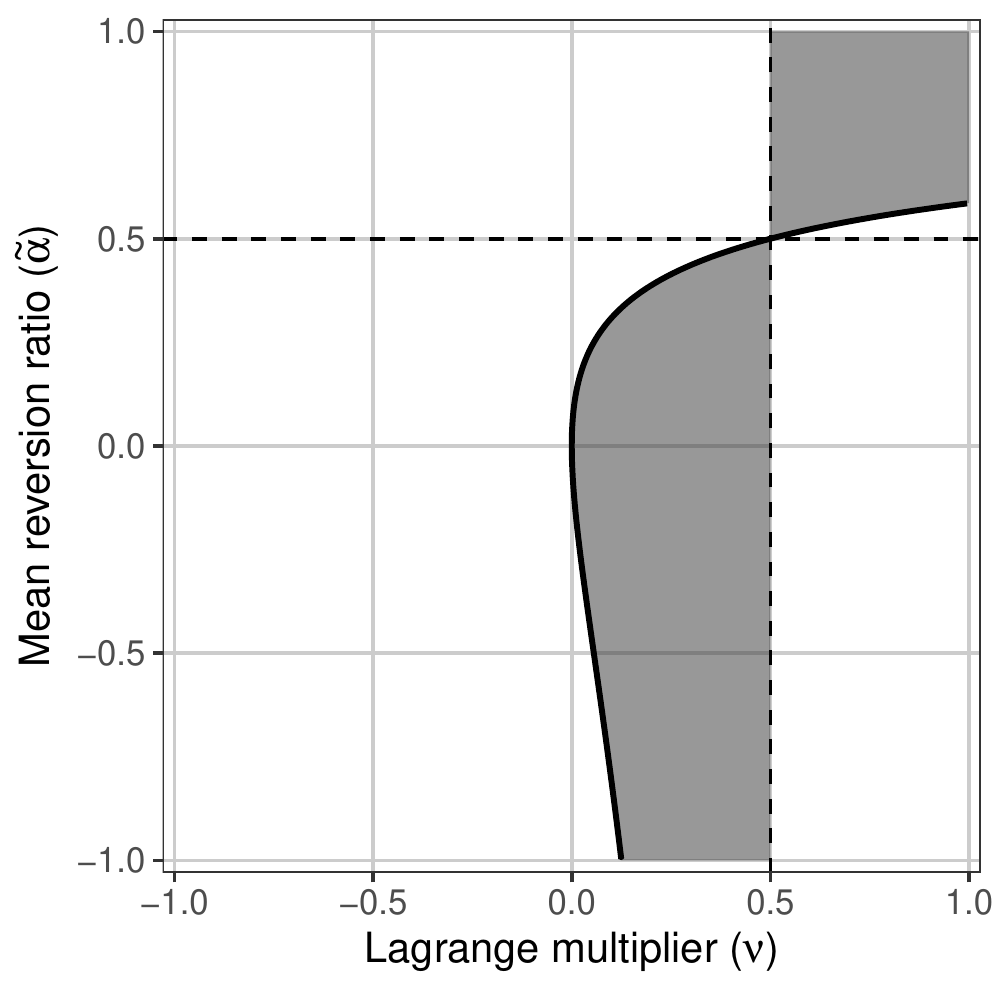}   
\end{center}
\vspace*{-5mm}
\caption{Plots of the parameter domains for solution types I (white area), II (gray area), and III (solid black line). The right plot is a close-up of the left plot. The solution type depends on the signs of $A$ and $C$, which in turn depend on $\nu$ (Lagrange multiplier) and $\tilde{\alpha} = \alpha/[\sigma_x/\sigma_S]$ (mean reversion ratio). The dashed vertical line marks the excluded value $\nu=1/2$; $A$ is positive to the left of this line and negative to the right of this line. The solid black line corresponds to $C=0$ (solution type III), and this line forms the boundary between solutions of type I and type II. The horizontal dashed line marks the special case $\tilde{\alpha}=\frac{1}{2}$ ($\alpha=\frac{1}{2}\frac{\sigma_x}{\sigma_S}$) for which the entire family of extremal strategies is of type I.}
\label{fig:domain}
\end{figure}

By varying $\nu$, but keeping the model parameters fixed, we obtain the family of all extremal strategies under this model. For $\nu<0$, the extremal strategies are optimal with a positive mean-variance trade-off, and for $\nu=0$ we get the global maximum. For $\nu>0$, we find both maximizing strategies, minimizing strategies, and locally extreme strategies, cf.\ Section~\ref{sec:NumOptimalEquity} for further discussion and examples. From Figure~\ref{fig:domain} we see that for $\alpha=\frac{1}{2}\frac{\sigma_x}{\sigma_S}$ (horizontal dashed line), all extremal strategies are of type I, and given by (\ref{eq:OptimalEquitySpecialB}) of Theorem~\ref{thm:OptimalEquity}. This is the only value of $\alpha$ for which all members of the family of extremal strategies are of the same type.

\subsection{Extremal equity strategies of type II and III}
Theorem~\ref{thm:ExtremalOverview} and Figure~\ref{fig:domain} give the full set of extremal strategies and the parameter values for which the different solutions apply. In Appendix~\ref{app:coefmatch}, the exponential solution (type I) was derived in detail for $\nu<0$. We see from Figure~\ref{fig:domain}, that the exponential solution also applies in some cases for $\nu>0$, but we can also have a trigonometric solution (type II), or a quadratic solution (type III). This section contains two propositions deriving the solutions of type II and III.

The method of proof is the same as in Appendix~\ref{app:coefmatch}. Knowing that the extremal strategies solve the ODE (\ref{eq:OptEquityODEapp}), we obtain an ansatz for the form of the solution. Next, we rewrite condition (\ref{eq:OptEquityIntegral}) of Lemma~\ref{lemma:OptEquityIntegral} for this ansatz, and this leads to a set of conditions for the constants appearing in the ansatz. We then verify that the stated strategies satisfy these conditions. This proves that the strategies are indeed extremal strategies. Uniqueness of the extremal strategy, for given $\nu$, could be proved along the same lines as in Proposition~\ref{prop:uniqueness}, but this is left to the reader.

\begin{proposition} \label{prop:solTrigometric}
Let $A$, $C$, and $D$ be given by (\ref{eq:CoefODEapp}). If $C<0$ and $A<0$, or $C>0$ and $A>0$, then condition (\ref{eq:OptEquityIntegral}) of Lemma~\ref{lemma:OptEquityIntegral} is satisfied for
  \begin{equation}
     f_s = b_0 + b_1 \sin(c s) + b_2 \cos(c s), \label{eq:solTrigonometric}
  \end{equation}
where $c=\sqrt{C/A}\neq 0$, $b_0=-D/C$, and
\begin{equation}
  \begin{pmatrix}
    b_1 \\
    b_2
  \end{pmatrix}
  =
  \begin{pmatrix}
    k_1\alpha + k_2 c & k_2\alpha - k_1 c \\[2mm]
    c & \tilde{R}
  \end{pmatrix}^{-1}
  \begin{pmatrix}
    \frac{\alpha\bar{x}}{\sigma_S}[ A^{-1} + \alpha^2 C^{-1}] \\[2mm]
    \frac{x_0}{\sigma_x}[c^2 + \tilde{R}^2]
  + \frac{\alpha\bar{x}}{\sigma_S}[ A^{-1} + \alpha\tilde{R}C^{-1}]
  \end{pmatrix},   \label{eq:TrigonometricDefb1b2}
\end{equation}
with $k_1=\sin(cT)$, $k_2=\cos(cT)$, and $\tilde{R}=\frac{\sigma_x}{\sigma_S}-\alpha$.
\end{proposition}
\begin{proof}
We will only give a sketch of the proof. Note that by assumption $C$ and $A$ are both non-zero and of equal sign. We are thus in the gray area of Figure~\ref{fig:domain}, and we see that this implies that $\nu\neq 0$ and $\alpha \neq \frac{1}{2}\frac{\sigma_x}{\sigma_S}$.

Assume first that $\alpha\neq 0$. By computations similar to those of Lemmas~\ref{lemma:GeneralFormOfh} and~\ref{lemma:GeneralFormOfInth} we obtain expressions for $h$ and $\int h$; the computations are left to the reader. Using these expressions we recast the left-hand side of condition (\ref{eq:OptEquityIntegral}) as
\begin{align}
  \xi_s - f_s & + 2\nu h_s - 2\nu \frac{\sigma_x}{\sigma_S}\int_0^s h_u e^{-\alpha(s-u)}du \nonumber \\[2mm]
    =\ & [\bar{x}+ e^{-\alpha s}(x_0 - \bar{x})]/\sigma_S - b_0 - b_1\sin(cs) - b_2\cos(cs) + \nonumber \\[2mm]
     & 2\nu\left( d_0 + d_1 \sin(cs) + d_2 \cos(cs) + d_3 e^{\alpha s} \right) - \nonumber \\[2mm]
     & 2\nu\left(a_0 + a_1 \sin(cs) + a_2\cos(cs) + a_3 e^{\alpha s} + a_4 e^{-\alpha s} \right) \nonumber \\[2mm]
    =\ & \left[ \bar{x}/\sigma_S-b_0 + 2\nu(d_0-a_0)\right] + \label{eq:TrigDefiningEquationI}  \\[2mm]
     & \left[-b_1 + 2\nu(d_1-a_1) \right]\sin(cs) + \left[-b_2 + 2\nu(d_2-a_2) \right]\cos(cs) + \label{eq:TrigDefiningEquationII} \\[2mm]
     & \left[ 2\nu(d_3-a_3) \right]e^{\alpha s} + \left[ (x_0-\bar{x})/\sigma_S - 2\nu a_4 \right]e^{-\alpha s}. \label{eq:TrigDefiningEquationIII}
\end{align}
where
\begin{align*}
  d_0 & = b_0\left[ 1 - \frac{\sigma_x}{\sigma_S\alpha}\right], \quad
    d_1 = b_1 - \frac{\sigma_x}{\sigma_S}\frac{\alpha b_1 - c b_2}{\alpha^2+c^2},  \quad
    d_2 = b_2 - \frac{\sigma_x}{\sigma_S}\frac{c b_1 + \alpha b_2}{\alpha^2+c^2}, \\[2mm]
  d_3 & = e^{-\alpha T}\frac{\sigma_x}{\sigma_S}\left[\frac{b_0}{\alpha} + k_1 \frac{\alpha b_1 - c b_2}{\alpha^2+c^2} + k_2 \frac{c b_1 + \alpha b_2}{\alpha^2+c^2}  \right], \\[2mm]
  a_0 & = \frac{\sigma_x}{\sigma_S}\frac{d_0}{\alpha}, \quad
  a_1 = \frac{\sigma_x}{\sigma_S}\frac{\alpha d_1 + c d_2}{\alpha^2+c^2}, \quad
  a_2 = \frac{\sigma_x}{\sigma_S}\frac{\alpha d_2-c d_1}{\alpha^2+c^2}, \quad
  a_3 = \frac{\sigma_x}{\sigma_S}\frac{d_3}{2\alpha}, \\[2mm]
  a_4 & = \frac{\sigma_x}{\sigma_S}\left[-\frac{d_0}{\alpha} - \frac{d_3}{2\alpha} + \frac{cd_1-\alpha d_2}{\alpha^2+c^2}\right].
\end{align*}
As in the proof of Proposition~\ref{prop:verifyGenCase}, it follows that condition (\ref{eq:OptEquityIntegral}) is satisfied if and only if the five terms in square brackets in (\ref{eq:TrigDefiningEquationI})--(\ref{eq:TrigDefiningEquationIII}) are zero. Note that these terms are formally equivalent to the ones of Proposition~\ref{prop:verifyGenCase}, but with different values of the constants.

The initial observation that $\nu\neq 0$ and $\alpha \neq \frac{1}{2}\frac{\sigma_x}{\sigma_S}$ implies that the first condition in (\ref{eq:TrigDefiningEquationIII}), $2\nu(d_3-a_3)=0$, is equivalent to $d_3=0$. Using this, we get after a series of simplifications the following, equivalent set of conditions
\begin{align*}
   & b_0 = -\frac{\bar{x}}{\sigma_S}\left[2\nu\left(1-\frac{\sigma_x}{\sigma_S\alpha}\right)^2-1\right]^{-1}, \\[2mm]
   & b_1 = b_1 2 \nu\left[ \left(1-\frac{\sigma_x}{\sigma_S}\frac{\alpha}{\alpha^2+c^2}\right)^2 + \left(\frac{\sigma_x}{\sigma_S}\frac{c}{\alpha^2+c^2}\right)^2 \right], \\[2mm]
   & b_2 = b_2 2 \nu\left[ \left(1-\frac{\sigma_x}{\sigma_S}\frac{\alpha}{\alpha^2+c^2}\right)^2 + \left(\frac{\sigma_x}{\sigma_S}\frac{c}{\alpha^2+c^2}\right)^2 \right], \\[2mm]
   & 0 = \frac{b_0}{\alpha} + k_1\frac{\alpha b_1 - c b_2}{\alpha^2+c^2} + k_2\frac{c b_1 + \alpha b_2}{\alpha^2+c^2}, \\[2mm]
   & 0 = \frac{x_0-\bar{x}}{\sigma_S} + 2\nu\frac{\sigma_x}{\sigma_S}\left[\frac{b_0}{\alpha}\left(1-\frac{\sigma_x}{\sigma_S\alpha}\right) - \frac{b_1 c}{\alpha^2+c^2} + b_2\frac{\alpha - \sigma_x/\sigma_S}{\alpha^2+c^2}\right].
\end{align*}
In the second and third of the above conditions, the term in square brackets is identical and we therefore in reality have only four conditions (corresponding to the number of constants in $f$). Inserting $c^2=C/A$, this term reduces to $1/(2\nu)$ and the two conditions are thereby satisfied. It is also straightforward to verify that the remaining conditions are satisfied for $b_0=-D/C$, and $(b_1,b_2)$ given by (\ref{eq:TrigonometricDefb1b2}). This concludes the proof for $\alpha\neq 0$.

For $\alpha=0$, we get by similar, but somewhat simpler calculations, that $f$ of form (\ref{eq:solTrigonometric}) solves (\ref{eq:OptEquityIntegral}) of Lemma~\ref{lemma:OptEquityIntegral} with $c=\sqrt{C/A}$, $b_0=0$ and
\begin{equation}
\begin{pmatrix}
    b_1 \\
    b_2
\end{pmatrix}
  =
  \begin{pmatrix}
    \cos(c T)c & -\sin(c T)c \\[2mm]
    c & \sigma_x/\sigma_S
  \end{pmatrix}^{-1}
  \begin{pmatrix}
    0 \\[2mm]
    x_0[c^2 + (\sigma_x/\sigma_S)^2]/\sigma_x
  \end{pmatrix}.
\end{equation}
Finally, we note that the solution for $\alpha=0$ is in fact also covered by the general solution, i.e., it is also of the form $b_0=-D/C$ and $(b_1,b_2)$ given by (\ref{eq:TrigonometricDefb1b2}) evaluated at $\alpha=0$. Hence, when stating the result we do not need to distinguish between the two cases.
\end{proof}

Finally, we give the proposition covering strategies of type III. These, quadratic, solutions constitute the boundary between the exponential solutions (type I) and the trigonometric solutions (type II). In the proof, we see that we obtain the same formal set of conditions as seen before, but with different values of the constants. It should be noted that $C=0$ implies $\alpha \neq \sigma_x/\sigma_S$, and (\ref{eq:QuadraticDefb0b1}) of Proposition~\ref{prop:solQuadratic} is therefore well-defined.

\begin{proposition} \label{prop:solQuadratic}
Let $A$, $C$, and $D$ be given by (\ref{eq:CoefODEapp}). If $C=0$ and $A \neq 0$, then condition (\ref{eq:OptEquityIntegral}) of Lemma~\ref{lemma:OptEquityIntegral} is satisfied for
  \begin{equation}
     f_s = b_0 + b_1 s + b_2 s^2, \label{eq:solQuadratic}
  \end{equation}
where $b_2=-D/(2A)$, and
\begin{equation}
  \begin{pmatrix}
    b_0 \\
    b_1
  \end{pmatrix}
  =
  \begin{pmatrix}
    \alpha  & T \alpha +1 \\[2mm]
    \sigma_x/\sigma_S & \sigma_x/(\sigma_x-\sigma_S\alpha)
  \end{pmatrix}^{-1}
  \begin{pmatrix}
    \frac{\bar{x}\alpha}{\sigma_S(1-2\nu)}\left[\frac{1}{2}T^2\alpha^2 + T\alpha + 1\right] \\[2mm]
    \frac{x_0}{\sigma_x}\tilde{R}
  + \frac{\bar{x}\alpha}{\sigma_S(1-2\nu)}\frac{\sigma_x}{\sigma_S}\tilde{R}^{-1}
  \end{pmatrix},   \label{eq:QuadraticDefb0b1}
\end{equation}
with $\tilde{R}=\frac{\sigma_x}{\sigma_S}-\alpha$.
\end{proposition}
\begin{proof}
We will only give a sketch of the proof. Assume first that $\alpha\neq 0$. By computations similar to those of Lemmas~\ref{lemma:GeneralFormOfh} and~\ref{lemma:GeneralFormOfInth} we obtain expressions for $h$ and $\int h$; the computations are left to the reader. Using these expressions we recast the left-hand side of condition (\ref{eq:OptEquityIntegral}) as
\begin{align}
  \xi_s - f_s & + 2\nu h_s - 2\nu \frac{\sigma_x}{\sigma_S}\int_0^s h_u e^{-\alpha(s-u)}du \nonumber \\[2mm]
    =\ & [\bar{x}+ e^{-\alpha s}(x_0 - \bar{x})]/\sigma_S - b_0 - b_1 s - b_2 s^2 + \nonumber \\[2mm]
     & 2\nu\left( d_0 + d_1 s + d_2 s^2 + d_3 e^{\alpha s} \right) - \nonumber \\[2mm]
     & 2\nu\left(a_0 + a_1 s + a_2 s^2 + a_3 e^{\alpha s} + a_4 e^{-\alpha s} \right) \nonumber \\[2mm]
    =\ & \left[ \bar{x}/\sigma_S-b_0 + 2\nu(d_0-a_0)\right] + \label{eq:QuadDefiningEquationI}  \\[2mm]
     & \left[-b_1 + 2\nu(d_1-a_1) \right]s + \left[-b_2 + 2\nu(d_2-a_2) \right]s^2 + \label{eq:QuadDefiningEquationII} \\[2mm]
     & \left[ 2\nu(d_3-a_3) \right]e^{\alpha s} + \left[ (x_0-\bar{x})/\sigma_S - 2\nu a_4 \right]e^{-\alpha s}. \label{eq:QuadDefiningEquationIII}
\end{align}
where
\begin{align*}
  d_0 & = b_0\left[ 1 - \frac{\sigma_x}{\sigma_S\alpha}\right] - \frac{\sigma_x}{\sigma_S}\frac{b_1\alpha + 2b_2}{\alpha^3}, \quad
  d_1 = b_1\left[ 1 - \frac{\sigma_x}{\sigma_S\alpha}\right] - \frac{\sigma_x}{\sigma_S}\frac{2 b_2}{\alpha^2},  \\[2mm]
  d_2 & = b_2\left[ 1 - \frac{\sigma_x}{\sigma_S\alpha}\right], \quad
  d_3 = e^{-\alpha T}\frac{\sigma_x}{\sigma_S\alpha}\left[b_0 + b_1\left(T+\frac{1}{\alpha}\right) + b_2\left(T^2 + \frac{2T}{\alpha} + \frac{2}{\alpha^2}\right) \right], \\[2mm]
  a_0 & = \frac{\sigma_x}{\sigma_S}\left[\frac{d_0}{\alpha} - \frac{d_1}{\alpha^2} + \frac{2d_2}{\alpha^3}\right], \quad
  a_1 = \frac{\sigma_x}{\sigma_S}\left[\frac{d_1}{\alpha} - \frac{2d_2}{\alpha^2}\right] , \quad
  a_2 = \frac{\sigma_x}{\sigma_S}\frac{d_2}{\alpha}, \quad
  a_3 = \frac{\sigma_x}{\sigma_S}\frac{d_3}{2\alpha}, \\[2mm]
  a_4 & = \frac{\sigma_x}{\sigma_S}\left[-\frac{d_0}{\alpha} - \frac{d_3}{2\alpha} + \frac{d_1}{\alpha^2} - \frac{2d_2}{\alpha^3}\right].
\end{align*}
As in the proof of Proposition~\ref{prop:verifyGenCase}, it follows that condition (\ref{eq:OptEquityIntegral}) is satisfied if and only if the five terms in square brackets in (\ref{eq:QuadDefiningEquationI})--(\ref{eq:QuadDefiningEquationIII}) are zero. It is easy to verify that $C=0$ implies that the two terms in square brackets in (\ref{eq:QuadDefiningEquationI})
are both zero. This reduces the number of conditions to three, corresponding to the number of constants in (\ref{eq:solQuadratic}).

By assumption $\nu\neq 0$ (as a consequence of $\alpha\neq 0$) and $\alpha \neq \frac{1}{2}\frac{\sigma_x}{\sigma_S}$ (as a consequence of $\nu\neq 1/2$), cf.\ Figure~\ref{fig:domain}. This implies that the first condition in (\ref{eq:TrigDefiningEquationIII}), $2\nu(d_3-a_3)=0$, is equivalent to $d_3=0$. Using this, we get after a series of simplifications the following, equivalent set of conditions
\begin{align*}
   & 0 = b_0 + b_1\left(T+\frac{1}{\alpha}\right) + b_2\left(T^2 + \frac{2T}{\alpha} + \frac{2}{\alpha^2}\right), \\[2mm]
   & 0 = \frac{x_0-\bar{x}}{\sigma_S}+2\nu\frac{\sigma_x}{\sigma_S\alpha}\left[b_0\left(1-\frac{\sigma_x}{\sigma_S\alpha}\right) - \frac{b_1}{\alpha} + \frac{2b_2}{\alpha^2}\left(1-\frac{\sigma_x}{\sigma_S\alpha}\right)\right], \\[2mm]
   & 0 = \frac{x_0}{\sigma_S}- b_0 + 2\nu\left[b_0\left(1-\frac{\sigma_x}{\sigma_S\alpha}\right)-\frac{\sigma_x}{\sigma_S\alpha}\frac{b_1}{\alpha} - \frac{\sigma_x}{\sigma_S\alpha}\frac{2b_2}{\alpha^2}\right].
\end{align*}
Adding the last two equations and using $C=0$, we get $b_2=-\alpha^2\bar{x}/[\sigma_S2(1-2\nu)]=-D/(2A)$. Substituting $b_2$ with $-D/(2A)$, the first two equations above become
\begin{align*}
    b_0+ b_1\left(T+\frac{1}{\alpha}\right) & = \frac{\bar{x}\alpha^2}{\sigma_S 2(1-2\nu)}\left(T^2 + \frac{2T}{\alpha} + \frac{2}{\alpha^2}\right), \\[2mm]
    b_0\left[2\nu\tilde{R}+\alpha\right] + b_1 2\nu\frac{\sigma_x}{\sigma_S\alpha} & = \frac{x_0\alpha}{\sigma_S} + \frac{\bar{x}2\nu}{\sigma_S(1-2\nu)}\frac{\sigma_x}{\sigma_S},
\end{align*}
where we have used the identify $2\nu\frac{\sigma_x}{\sigma_S\alpha}(1-\frac{\sigma_x}{\sigma_S\alpha})=2\nu(1-\frac{\sigma_x}{\sigma_S\alpha})-1$, due to $C=0$. Using again $C=0$, we rewrite terms containing $2\nu$, e.g., $2\nu\tilde{R}+\alpha=\alpha\frac{\sigma_x}{\sigma_S}\tilde{R}^{-1}$, to get
\begin{align*}
    b_0\alpha + b_1\left(T\alpha+1\right) & = \frac{\bar{x}\alpha}{\sigma_S(1-2\nu)}\left(T^2\alpha^2/2 + T\alpha+1\right), \\[2mm]
    b_0\frac{\sigma_x}{\sigma_S} + b_1\frac{\sigma_x}{\sigma_x-\sigma_S\alpha} & = \frac{x_0}{\sigma_x}\tilde{R} + \frac{\bar{x}\alpha}{\sigma_S(1-2\nu)}\frac{\sigma_x}{\sigma_S}\tilde{R}^{-1},
\end{align*}
which is satisfied for $(b_0, b_1)$ given by (\ref{eq:QuadraticDefb0b1}). This concludes the proof for $\alpha\neq 0$.

Assume $\alpha=0$. Since $C=0$ implies $\nu=0$, we know that (\ref{eq:OptEquityIntegral}) of Lemma~\ref{lemma:OptEquityIntegral} is satisfied for $f_s$ equal to the expected market price of equity risk, cf.\ Section~\ref{sec:OptEqStrategyPosTradeOff}. Thus, $f_s=x_0/\sigma_S$ solves (\ref{eq:OptEquityIntegral}). Since this is also the solution we obtain from (\ref{eq:solQuadratic}) with $b_2=-D/(2A)$ and $(b_0, b_1)$ given by (\ref{eq:QuadraticDefb0b1}) the proposition is proved.
\end{proof}


\section{Proofs of results in the main text}  \label{app:proofs}
\subsection*{Proof of Theorem~\ref{thm:VTrep}}
\begin{proof}
First, we decompose the exponent of the right-hand side of (\ref{eq:Vintrep}) into the contributions from the two risk sources and a mixed term,
\begin{equation} \label{eq:Vdecomp}
  V_T = V_0\exp\left\{I^r_T + I^S_T - \rho\int_0^T f^r_s f^S_s ds \right\},
\end{equation}
where
\begin{align}
  I^r_T & = \int_0^T r_s ds + \int_0^T f^r_s \lambda^r_s ds - \frac{1}{2}\int_0^T \left(f^r_s\right)^2 ds + \int_0^T f^r_s dW^r_s,  \\[2mm]
  I^S_T & = \int_0^T f^S_s \lambda^S_s ds - \frac{1}{2}\int_0^T \left(f^S_s\right)^2 ds + \int_0^T f^S_s dW^S_s.
\end{align}
Next, we use (\ref{eq:lambdar}) and (\ref{eq:rsol}) to obtain an integral representation of $\lambda^r_s$
\begin{align}
  \lambda^r_s & = \frac{1}{\sigma_r}\left[(\kappa\bar{r}-ab) + (a-\kappa)r_s \right] \nonumber \\[2mm]
              & = \frac{1}{\sigma_r}\left[(\kappa\bar{r}-ab) + (a-\kappa)\left(\bar{r} + e^{-\kappa s}(r_0 - \bar{r}) + \sigma_r \int_0^s e^{-\kappa(s-u)}dW^r_u \right) \right] \nonumber \\[2mm]
              & = \frac{1}{\sigma_r}\left[a(\bar{r}-b)+e^{-\kappa s}(a-\kappa)(r_0 - \bar{r})\right] + (a-\kappa)\int_0^s e^{-\kappa(s-u)}dW^r_u.  \label{eq:lambdarint}
\end{align}
Substituting $\lambda^r_s$ with (\ref{eq:lambdarint}) in the expression for $I^r_T$, and using (\ref{eq:intrsol}), we get
\begin{align}
  I^r_T & = m^0_T + m^r_T + \int_0^T \left[\sigma_r\Psi(\kappa,T-u) + f^r_u\right] dW^r_u + (a-\kappa)\int_0^T \int_0^s f^r_s e^{-\kappa(s-u)}dW^r_u ds. \nonumber \\[2mm]
        & = m^0_T + m^r_T + \int_0^T \left[\sigma_r\Psi(\kappa,T-u) + f^r_u + (a-\kappa)\int_u^T f^r_s e^{-\kappa(s-u)}ds\right] dW^r_u \nonumber \\[2mm]
        & = m^0_T + m^r_T + \int_0^T h^r_u dW^r_u, \label{eq:IrTrep}
\end{align}
where the second equality follows by changing the order of integration and combining the stochastic integrals.

Recalling that $\lambda^S_s = x_s/\sigma_S$, cf.~(\ref{eq:lambdaS}), and using the integral representation (\ref{eq:xsol}), we get
\begin{align}
  I^S_T  & =  \frac{1}{\sigma_S}\int_0^T f^S_s\left(\bar{x} + e^{-\alpha s}(x_0 - \bar{x}) - \sigma_x \int_0^s e^{-\alpha(s-u)}dW^S_u \right)ds - \nonumber \\[2mm]
         &   \quad\quad \frac{1}{2}\int_0^T \left(f^S_s\right)^2 ds + \int_0^T f^S_u dW^S_u. \nonumber \\[2mm]
         &=  m^S_T - \frac{\sigma_x}{\sigma_S}\int_0^T \int_u^T f^S_s e^{-\alpha(s-u)}ds dW^S_u + \int_0^T f^S_u dW^S_u \nonumber \\[2mm]
         & =  m^S_T + \int_0^T h^S_u dW^S_u. \label{eq:ISTrep}
\end{align}
The representation (\ref{eq:Vrep}) now follows from inserting (\ref{eq:IrTrep}) and (\ref{eq:ISTrep}) into (\ref{eq:Vdecomp}). The final claims regarding log-normality follows directly from (\ref{eq:Vrep}), and the variance formula follows from $\sigma_T^2 = \langle \int_0^T h_u^r dW^r_u + \int_0^T h^S_u dW^S_u, \int_0^T h_u^r dW^r_u + \int_0^T h^S_u dW^S_u\rangle$.
\end{proof}

\subsection*{Proof of Theorem~\ref{thm:OptimalBond}}
\begin{proof}
The problem takes the form of an isoperimetric problem with a constraint and it can be solve be calculus of variations, see e.g.\ Chapter 4 of \cite{wei74}. For ease of notation we will skip the superscript $r$ throughout.

We want to find the extremizing function of the functional
\begin{equation}
   I[f] = \int_0^T L(s,f,f') ds,
\end{equation}
with Lagrangian $L(s,f,f')= \lambda f_s - \frac{1}{2}f_s^2$, subject to the constraint that
\begin{equation}
   J[f] = \int_0^T \left(g_s + f_s\right)^2 ds
\end{equation}
possesses a given prescribed value. The constraint is handled by introducing a Lagrange multiplier. We thus introduce the functional
\begin{equation}
   I^*[f] = \int_0^T L^*(s,f,f') ds,
\end{equation}
with Lagrangian $L^*(s,f,f')= \lambda f_s - \frac{1}{2}f_s^2 + \nu \left(g_s + f_s\right)^2$. The constant $\nu$ is the undetermined multiplier whose value remains to be determined by the (variance) constraint.

Since $L^*$ does not depend on $f'$  the associated Euler-Lagrange equation satisfied by the extremizing function is simply
\begin{equation}
   \frac{\partial L^*}{\partial f} = \lambda - f_s + 2\nu\left(g_s + f_s\right) = 0,
\end{equation}
with solution given by (\ref{eq:OptimalBond}). Further, since the second variation equals $2\nu-1$, it follows that the extremizing function is a maximum for $\nu<1/2$ (and a minimum for $\nu>1/2$).
\end{proof}

\subsection*{Proof of Lemma~\ref{lemma:OptEquityIntegral}}  
\begin{proof}
We assume that $f$ is the extremizing function of (\ref{eq:VTmeanequityonly}) for given value of (\ref{eq:VTvarequityonly}), and we introduce a family of comparison functions with respect to which we carry out the extremization. Consider the two-parameter family of functions
\begin{equation}
  F(s) = f_s + \epsilon_1 \eta_1(s) + \epsilon_2 \eta_2(s)
\end{equation}
in which $\eta_1$ and $\eta_2$ are arbitrary continuous functions. Replace $f_s$ by $F(s)$ in (\ref{eq:VTmeanequityonly}) and (\ref{eq:VTvarequityonly}) to form
\begin{align}
  I(\epsilon_1, \epsilon_2) & = \int_0^T \xi_s F(s) - \frac{1}{2}F^2(s) ds, \\[2mm]
  J(\epsilon_1, \epsilon_2) & = \int_0^T \left[F(u) - \frac{\sigma_x}{\sigma_S}\int_u^T F(s) e^{-\alpha(s-u)}ds \right]^2 du.
\end{align}
The parameters $\epsilon_1$ and $\epsilon_2$ are not independent, since $J$ must maintain the prescribed value, $c$, say. By assumption, $J(0,0)=c$ and $I$ is maximized for $\epsilon_1=\epsilon_2=0$ with
respect to $\epsilon_1$ and $\epsilon_2$ satisfying $J(\epsilon_1,\epsilon_2)=c$.

The above construction turns the problem into one of finding conditions for an extremal point of $I$ under the constraint $J=c$. This is solved by a standard application of Lagrange multipliers. Introduce the function
\begin{equation}
  I^*(\epsilon_1,\epsilon_2) = I(\epsilon_1,\epsilon_2) + \nu J(\epsilon_1,\epsilon_2),
\end{equation}
where $\nu$ is the undetermined multiplier to be determined by the (variance) constraint. We must have
\begin{equation}\label{eq:Istarzero}
  \frac{\partial I^*}{\partial \epsilon_1}(0,0) = \frac{\partial I^*}{\partial \epsilon_2}(0,0) = 0,
\end{equation}
from which the integral characterization will follow.

Differentiation of $I$ and $J$ with respect to $\epsilon_i$ yields
\begin{align*}
  \frac{\partial I}{\partial \epsilon_i} & = \int_0^T \xi_s \eta_i(s) - F(s)\eta_i(s) ds = \int_0^T \left[\xi_s - F(s)\right]\eta_i(s) ds, \\[2mm]
  \frac{\partial J}{\partial \epsilon_i} & = \int_0^T 2\left[F(u) - \frac{\sigma_x}{\sigma_S}\int_u^T F(s) e^{-\alpha(s-u)}ds \right]\left(\eta_i(u)- \frac{\sigma_x}{\sigma_S}\int_u^T \eta_i(s) e^{-\alpha(s-u)}ds \right) du,
\end{align*}
and by combining these expressions and recalling the definition of $h$ in (\ref{eq:hEquityDef}) we get
\begin{align}
  \frac{\partial I^*}{\partial \epsilon_i}(0,0) & = \frac{\partial I}{\partial \epsilon_i}(0,0) + \nu \frac{\partial J}{\partial \epsilon_i}(0,0)  \nonumber \\[2mm]
  & = \int_0^T \left[\xi_s - f_s\right]\eta_i(s) ds  + 2\nu \int_0^T h_u \left[\eta_i(u)- \frac{\sigma_x}{\sigma_S}\int_u^T \eta_i(s) e^{-\alpha(s-u)}ds \right] du \nonumber \\[2mm]
  & = \int_0^T \left[\xi_s - f_s + 2\nu h_s\right]\eta_i(s) ds - 2\nu\frac{\sigma_x}{\sigma_S} \int_0^T \int_0^s h_u\eta_i(s) e^{-\alpha(s-u)}  du ds \nonumber \\[2mm]
  & = \int_0^T \left[\xi_s - f_s + 2\nu h_s - 2\nu\frac{\sigma_x}{\sigma_S}\int_0^s h_u e^{-\alpha(s-u)} du   \right]\eta_i(s) ds. \label{eq:Istar}
\end{align}
Because of (\ref{eq:Istarzero}) the expression in (\ref{eq:Istar}) equals 0 for $i=1,2$. Finally, since $\eta_i$ is an arbitrary continuous function it follows from the fundamental lemma of calculus of variations that the expression in square brackets in (\ref{eq:Istar}) is identically zero, i.e., the optimal $f$ satisfies (\ref{eq:OptEquityIntegral}) as claimed.
\end{proof}

\subsection*{Proof of Lemma~\ref{lemma:OptEquityODE}}  
\begin{proof}
For given $T$ and $\nu$, let $G(s)$ denote the left-hand side of (\ref{eq:OptEquityIntegral}). It follows from Lemma~\ref{lemma:OptEquityIntegral} that $G(s)=0$ for $0\leq s\leq T$ and thereby also
\begin{equation}
   G(s) = G'(s) = G''(s) = 0 \quad (0\leq s \leq T).
\end{equation}
The proof consists of carrying out these differentiations and simplifying the expressions.

By swapping the roles of $u$ and $s$ in the definition of $h$ given in (\ref{eq:hEquityDef}), we have
\begin{equation}\label{eq:hAlternative}
  h_s = f_s - \frac{\sigma_x}{\sigma_S}\int_s^T f_u e^{-\alpha(u-s)}du \quad (0\leq s \leq T),
\end{equation}
and thereby
\begin{align}
  h'_s  & = f'_s - \alpha\frac{\sigma_x}{\sigma_S}\int_s^T e^{-\alpha(u-s)}f_udu + \frac{\sigma_x}{\sigma_S}f_s.
\end{align}
Further, we have
\begin{align}
\frac{d}{ds}\left(-2\nu\frac{\sigma_x}{\sigma_S}\int_0^s h_u e^{-\alpha(s-u)}du\right) & = \alpha 2\nu\frac{\sigma_x}{\sigma_S}\int_0^s h_u e^{-\alpha(s-u)}du - 2\nu\frac{\sigma_x}{\sigma_S}h_s, \nonumber \\[2mm]
     & = \alpha\left(\xi_s - f_s + 2\nu h_s \right) - 2\nu\frac{\sigma_x}{\sigma_S}h_s,
\end{align}
where we have used $G(s)=0$ in the last equality. Using the above expressions we get
\begin{align}
  G'(s) & = \xi'_s - f'_s + 2\nu h'_s + \frac{d}{ds}\left(-2\nu\frac{\sigma_x}{\sigma_S}\int_0^s h_u e^{-\alpha(s-u)}du\right) \nonumber \\[2mm]
  & = \xi'_s - f'_s + 2\nu\left(f'_s - \alpha\frac{\sigma_x}{\sigma_S}\int_s^T e^{-\alpha(u-s)}f_udu + \frac{\sigma_x}{\sigma_S}f_s \right)
      + \alpha\left(\xi_s - f_s + 2\nu h_s \right) - 2\nu\frac{\sigma_x}{\sigma_S}h_s  \nonumber \\[2mm]
    & = \xi'_s + \alpha\xi_s + (2\nu-1)f'_s + \alpha(2\nu-1)f_s + 2\nu\frac{\sigma_x}{\sigma_S}\left(\frac{\sigma_x}{\sigma_S}-2\alpha\right)\int_s^T e^{-\alpha(u-s)}f_udu \nonumber \\[2mm]
    & = \frac{\alpha}{\sigma_S}\bar{x} + (2\nu-1)f'_s + \alpha(2\nu-1)f_s + C_0 \int_s^T e^{-\alpha(u-s)}f_udu,  \label{eq:Gprime}
\end{align}
where the last equality uses the definition of $\xi$ in (\ref{eq:meanEquityRiskPremium}), and we let $C_0=2\nu\frac{\sigma_x}{\sigma_S}(\frac{\sigma_x}{\sigma_S}-2\alpha)$. Differentiating $G$ once again yields
\begin{align}
  G''(s) & = (2\nu-1)f''_s + \alpha(2\nu-1)f'_s + C_0\left[\alpha\int_s^T e^{-\alpha(u-s)}f_udu - f_s\right]   \nonumber \\[2mm]
         & = (2\nu-1)f''_s + \alpha(2\nu-1)f'_s - \alpha\left[\frac{\alpha}{\sigma_S}\bar{x} + (2\nu-1)f'_s + \alpha(2\nu-1)f_s \right] - C_0 f_s \nonumber \\[2mm]
         & = (2\nu-1)f''_s -  \left[ C_0 + \alpha^2(2\nu-1) \right] f_s - \frac{\alpha^2}{\sigma_S}\bar{x}, \label{eq:Gprime2}
\end{align}
where, in the second equality, we have used the expression for $G'(s)$ from (\ref{eq:Gprime}) and the fact that $G'(s)=0$ to replace the integral. Since $G''(s)=0$ and thereby also $-G''(s)=0$, (\ref{eq:OptEquityODE}) follows from (\ref{eq:Gprime2}) after realizing $C_0 + \alpha^2(2\nu-1) = 2\nu(\alpha - \sigma_x/\sigma_S)^2 - \alpha^2$.
\end{proof}

\subsection*{Proof of Theorem~\ref{thm:OptimalEquity}}
\begin{proof}
It follows from Propositions~\ref{prop:verifyGenCase}, \ref{prop:verifyAlphaZero}, and~\ref{prop:verifyAlphaHalfRatio} that the stated strategy satisfies (\ref{eq:OptEquityIntegral}) of Lemma~\ref{lemma:OptEquityIntegral} for $\nu<0$, i.e., that it is an optimal strategy with positive mean-variance trade-off. Further, it follows from Proposition~\ref{prop:uniqueness} that there are no other solutions to (\ref{eq:OptEquityIntegral}) for $\nu<0$.
\end{proof}

\subsection*{Proof of Theorem~\ref{thm:OptimalJoint}}
\begin{proof}
Assume that $(f^r,f^S)$ is a pair of functions extremizing (\ref{eq:VTmeanIndependent}) for given value of (\ref{eq:VTvarIndependent}), say, $c$. We introduce a family of comparison functions with respect to which we carry out the extremization. Consider the four-parameter family of functions
\begin{align}
  F^r(s) & = f^r_s + \epsilon^r_1\eta^r_1(s) + \epsilon^r_2\eta^r_2(s), \\[2mm]
  F^S(s) & = f^S_s + \epsilon^S_1\eta^S_1(s) + \epsilon^S_2\eta^S_2(s),
\end{align}
in which the $\eta$'s are arbitrary continuous functions. Define the four functions
\begin{align}
  I^r(\epsilon^r_1,\epsilon^r_2) & = \int_0^T \lambda^r F^r(s) - \frac{1}{2} \left(F^r(s)\right)^2ds, \\[2mm]
  I^S(\epsilon^S_1,\epsilon^S_2) & = \int_0^T \xi_s F^S(s) - \frac{1}{2}\left(F^S(s)\right)^2 ds      \\[2mm]
  J^r(\epsilon^r_1,\epsilon^r_2) & = \int_0^T \left[\sigma_r\Psi(\kappa,T-u) + F^r(u)  \right]^2 du ,   \\[2mm]
  J^S(\epsilon^S_1,\epsilon^S_2) & = \int_0^T \left[F^S(u) - \frac{\sigma_x}{\sigma_S}\int_u^T F^S(s) e^{-\alpha(s-u)}ds \right]^2 du.
\end{align}
Further, let $I$ and $J$ denote the value of (\ref{eq:VTmeanIndependent}) and (\ref{eq:VTvarIndependent}) with $f^r$ and $f^S$ replaced by, respectively, $F^r$ and $F^S$. Then
\begin{align}
  I(\epsilon^r_1,\epsilon^r_2,\epsilon^S_1,\epsilon^S_2)  & = I^r(\epsilon^r_1,\epsilon^r_2) + I^S(\epsilon^S_1,\epsilon^S_2), \\[2mm]
  J(\epsilon^r_1,\epsilon^r_2,\epsilon^S_1,\epsilon^S_2)  & = J^r(\epsilon^r_1,\epsilon^r_2) + J^S(\epsilon^S_1,\epsilon^S_2).
\end{align}
By assumption, $J(0,0,0,0)=c$ and $I$ is (locally) extremized for $\epsilon^r_1 = \epsilon^r_2 = \epsilon^S_1 =\epsilon^S_2 = 0$ with respect to the $\epsilon$'s satisfying $J(\epsilon^r_1,\epsilon^r_2,\epsilon^S_1,\epsilon^S_2)=0$. The extremity of $(f^r,f^S)$ is thus equivalent to $(0,0,0,0)$ being a stationary point of $I$ under the constraint $J=c$ (for all $\eta$'s).
This is solved by a standard application of Lagrange multipliers. Introduce the function
\begin{equation}
   I^*(\epsilon^r_1,\epsilon^r_2,\epsilon^S_1,\epsilon^S_2) =  I(\epsilon^r_1,\epsilon^r_2,\epsilon^S_1,\epsilon^S_2) + \nu J(\epsilon^r_1,\epsilon^r_2,\epsilon^S_1,\epsilon^S_2),
\end{equation}
where $\nu$ is a Lagrange multiplier. We must have
\begin{align}
  \frac{\partial I^*}{\partial \epsilon^r_1}(0,0,0,0) & = \frac{\partial I^r}{\partial \epsilon^r_1}(0,0) + \nu \frac{\partial J^r}{\partial \epsilon^r_1}(0,0)=0, \label{eq:jointIstar1} \\[2mm]
  \frac{\partial I^*}{\partial \epsilon^r_2}(0,0,0,0) & = \frac{\partial I^r}{\partial \epsilon^r_2}(0,0) + \nu \frac{\partial J^r}{\partial \epsilon^r_2}(0,0)=0, \label{eq:jointIstar2} \\[2mm]
  \frac{\partial I^*}{\partial \epsilon^S_1}(0,0,0,0) & = \frac{\partial I^S}{\partial \epsilon^S_1}(0,0) + \nu \frac{\partial J^S}{\partial \epsilon^S_1}(0,0)=0, \label{eq:jointIstar3} \\[2mm]
  \frac{\partial I^*}{\partial \epsilon^S_2}(0,0,0,0) & = \frac{\partial I^S}{\partial \epsilon^S_2}(0,0) + \nu \frac{\partial J^S}{\partial \epsilon^S_2}(0,0)=0. \label{eq:jointIstar4}
\end{align}
Now, we realize that (\ref{eq:jointIstar1})--(\ref{eq:jointIstar2}) are the conditions for $f^r$ being an extremal rate strategy, while (\ref{eq:jointIstar3})--(\ref{eq:jointIstar4}) are the conditions for $f^S$ being an extremal equity strategy, cf.\ the proofs of Theorem~\ref{thm:OptimalBond} and Lemma~\ref{lemma:OptEquityIntegral}. We conclude that $(f^r,f^S)$ is an extremal pair of strategies for Problem~\ref{problem:IndependentRiskFactors} if and only if $f^r$ and $f^S$ are extremal strategies for, respectively, Problems~\ref{problem:RatesOnly} and \ref{problem:EquitiesOnly}, with the same $\nu$.
\end{proof}

\pagebreak
\section{Risk and reward statistics}  \label{app:statistics}
\subsection{Optimal rate strategies}

\begin{table}[ht]
  \centering
\rowcolors{2}{gray!25}{white}
\bgroup
\def\arraystretch{1.3}
\begin{tabular}{C{0.8cm}|L{3.2cm}C{1.1cm}C{1.1cm}C{1.1cm}C{1.1cm}C{1.1cm}C{1.1cm}C{1.1cm}} \hline
  \rowcolor{gray!50}
    $T$     &   \hphantom{median} $\nu$                &   -10    &   -2     & -1    &  -1/2  & -1/4   & -1/16 & 0 \\ \hline
  10 &  median$\left(Y_T\right)$       &   1.019  &   1.076  &   1.120  &   1.165  &   1.198  &   1.223  &   1.226  \\
     &  $\gP(Y_T<1)$                   &   0.267  &   0.283  &   0.297  &   0.316  &   0.335  &   0.361  &   0.375  \\
     &  $\E\left[1-Y_T|Y_T<1\right]$   &   0.018  &   0.074  &   0.121  &   0.176  &   0.227  &   0.291  &   0.322  \\
     &  $\E\left[1-Y_T\right]^+$       &   0.005  &   0.021  &   0.036  &   0.056  &   0.076  &   0.105  &   0.121  \\ \hline
  20 &  median$\left(Y_T\right)$       &   1.033  &   1.132  &   1.211  &   1.295  &   1.359  &   1.406  &   1.412  \\
     &  $\gP(Y_T<1)$                   &   0.209  &   0.227  &   0.244  &   0.267  &   0.290  &   0.322  &   0.339  \\
     &  $\E\left[1-Y_T|Y_T<1\right]$   &   0.022  &   0.089  &   0.144  &   0.208  &   0.268  &   0.340  &   0.374  \\
     &  $\E\left[1-Y_T\right]^+$       &   0.005  &   0.020  &   0.035  &   0.055  &   0.078  &   0.110  &   0.127  \\ \hline
  30 &  median$\left(Y_T\right)$       &   1.044  &   1.181  &   1.293  &   1.415  &   1.509  &   1.579  &   1.588  \\
     &  $\gP(Y_T<1)$                   &   0.174  &   0.193  &   0.211  &   0.235  &   0.261  &   0.297  &   0.315  \\
     &  $\E\left[1-Y_T|Y_T<1\right]$   &   0.024  &   0.097  &   0.157  &   0.226  &   0.290  &   0.368  &   0.404  \\
     &  $\E\left[1-Y_T\right]^+$       &   0.004  &   0.019  &   0.033  &   0.053  &   0.076  &   0.109  &   0.127  \\ \hline
  40 &  median$\left(Y_T\right)$       &   1.054  &   1.228  &   1.372  &   1.533  &   1.660  &   1.756  &   1.768  \\
     &  $\gP(Y_T<1)$                   &   0.149  &   0.168  &   0.187  &   0.212  &   0.238  &   0.277  &   0.297  \\
     &  $\E\left[1-Y_T|Y_T<1\right]$   &   0.026  &   0.104  &   0.167  &   0.240  &   0.307  &   0.387  &   0.424  \\
     &  $\E\left[1-Y_T\right]^+$       &   0.004  &   0.017  &   0.031  &   0.051  &   0.073  &   0.107  &   0.126  \\ \hline
  50 &  median$\left(Y_T\right)$       &   1.065  &   1.274  &   1.453  &   1.656  &   1.819  &   1.944  &   1.960  \\
     &  $\gP(Y_T<1)$                   &   0.129  &   0.148  &   0.167  &   0.192  &   0.220  &   0.260  &   0.281  \\
     &  $\E\left[1-Y_T|Y_T<1\right]$   &   0.027  &   0.108  &   0.174  &   0.250  &   0.319  &   0.402  &   0.441  \\
     &  $\E\left[1-Y_T\right]^+$       &   0.003  &   0.016  &   0.029  &   0.048  &   0.070  &   0.104  &   0.124  \\ \hline
  60 &  median$\left(Y_T\right)$       &   1.075  &   1.321  &   1.537  &   1.787  &   1.990  &   2.148  &   2.168  \\
     &  $\gP(Y_T<1)$                   &   0.112  &   0.131  &   0.150  &   0.175  &   0.203  &   0.245  &   0.267  \\
     &  $\E\left[1-Y_T|Y_T<1\right]$   &   0.028  &   0.113  &   0.181  &   0.259  &   0.330  &   0.415  &   0.454  \\
     &  $\E\left[1-Y_T\right]^+$       &   0.003  &   0.015  &   0.027  &   0.045  &   0.067  &   0.102  &   0.121  \\ \hline
\end{tabular}
\egroup
  \caption{Risk and reward statistics for the stochastic multiplier $Y_T$ in (\ref{eq:VTrateFactorStructure}) for optimal rate strategies on horizons from 10 to 60 years for a {\em moderate} market price of interest rate risk, cf.\ Table~\ref{tab:ratePar} of Section~\ref{sec:NumOptimalRate}.}
  \label{tab:rateStatisticsMod}
\end{table}

\begin{table}[ht]
  \centering
\rowcolors{2}{gray!25}{white}
\bgroup
\def\arraystretch{1.3}
\begin{tabular}{C{0.8cm}|L{3.2cm}C{1.1cm}C{1.1cm}C{1.1cm}C{1.1cm}C{1.1cm}C{1.1cm}C{1.1cm}} \hline
  \rowcolor{gray!50}
    $T$     &   \hphantom{median} $\nu$                &   -10    &   -2     & -1    &  -1/2  & -1/4   & -1/16 & 0 \\ \hline
  10 &  median$\left(Y_T\right)$       &   1.004  &   1.014  &   1.022  &   1.030  &   1.035  &   1.039  &   1.040  \\
     &  $\gP(Y_T<1)$                   &   0.393  &   0.401  &   0.408  &   0.417  &   0.426  &   0.438  &   0.445  \\
     &  $\E\left[1-Y_T|Y_T<1\right]$   &   0.009  &   0.039  &   0.063  &   0.094  &   0.123  &   0.160  &   0.178  \\
     &  $\E\left[1-Y_T\right]^+$       &   0.004  &   0.015  &   0.026  &   0.039  &   0.052  &   0.070  &   0.079  \\ \hline
  20 &  median$\left(Y_T\right)$       &   1.005  &   1.019  &   1.030  &   1.041  &   1.049  &   1.054  &   1.055  \\
     &  $\gP(Y_T<1)$                   &   0.375  &   0.384  &   0.393  &   0.403  &   0.414  &   0.428  &   0.435  \\
     &  $\E\left[1-Y_T|Y_T<1\right]$   &   0.011  &   0.044  &   0.072  &   0.107  &   0.140  &   0.182  &   0.202  \\
     &  $\E\left[1-Y_T\right]^+$       &   0.004  &   0.017  &   0.028  &   0.043  &   0.058  &   0.078  &   0.088  \\ \hline
  30 &  median$\left(Y_T\right)$       &   1.006  &   1.022  &   1.035  &   1.047  &   1.056  &   1.062  &   1.063  \\
     &  $\gP(Y_T<1)$                   &   0.366  &   0.376  &   0.385  &   0.397  &   0.408  &   0.423  &   0.431  \\
     &  $\E\left[1-Y_T|Y_T<1\right]$   &   0.011  &   0.047  &   0.077  &   0.113  &   0.147  &   0.192  &   0.213  \\
     &  $\E\left[1-Y_T\right]^+$       &   0.004  &   0.018  &   0.030  &   0.045  &   0.060  &   0.081  &   0.092  \\ \hline
  40 &  median$\left(Y_T\right)$       &   1.006  &   1.024  &   1.038  &   1.051  &   1.061  &   1.068  &   1.069  \\
     &  $\gP(Y_T<1)$                   &   0.361  &   0.372  &   0.381  &   0.392  &   0.404  &   0.420  &   0.428  \\
     &  $\E\left[1-Y_T|Y_T<1\right]$   &   0.012  &   0.048  &   0.079  &   0.116  &   0.152  &   0.198  &   0.219  \\
     &  $\E\left[1-Y_T\right]^+$       &   0.004  &   0.018  &   0.030  &   0.046  &   0.061  &   0.083  &   0.094  \\ \hline
  50 &  median$\left(Y_T\right)$       &   1.007  &   1.026  &   1.040  &   1.054  &   1.065  &   1.072  &   1.073  \\
     &  $\gP(Y_T<1)$                   &   0.357  &   0.368  &   0.377  &   0.389  &   0.401  &   0.417  &   0.425  \\
     &  $\E\left[1-Y_T|Y_T<1\right]$   &   0.012  &   0.050  &   0.081  &   0.119  &   0.156  &   0.202  &   0.225  \\
     &  $\E\left[1-Y_T\right]^+$       &   0.004  &   0.018  &   0.031  &   0.046  &   0.063  &   0.084  &   0.096  \\ \hline
  60 &  median$\left(Y_T\right)$       &   1.007  &   1.027  &   1.042  &   1.057  &   1.068  &   1.076  &   1.077  \\
     &  $\gP(Y_T<1)$                   &   0.353  &   0.364  &   0.374  &   0.386  &   0.398  &   0.415  &   0.423  \\
     &  $\E\left[1-Y_T|Y_T<1\right]$   &   0.012  &   0.051  &   0.083  &   0.122  &   0.159  &   0.206  &   0.229  \\
     &  $\E\left[1-Y_T\right]^+$       &   0.004  &   0.018  &   0.031  &   0.047  &   0.063  &   0.086  &   0.097  \\ \hline
\end{tabular}
\egroup
  \caption{Risk and reward statistics for the stochastic multiplier $Y_T$ in (\ref{eq:VTrateFactorStructure}) for optimal rate strategies on horizons from 10 to 60 years for a {\em low} market price of interest rate risk, cf.\ Table~\ref{tab:ratePar} of Section~\ref{sec:NumOptimalRate}.}
  \label{tab:rateStatisticsLow}
\end{table}
\clearpage

\subsection{Optimal equity strategies}

\begin{table}[ht]
  \centering
\rowcolors{2}{gray!25}{white}
\bgroup
\def\arraystretch{1.3}
\begin{tabular}{C{0.8cm}|L{3.2cm}C{1.1cm}C{1.1cm}C{1.1cm}C{1.1cm}C{1.1cm}C{1.1cm}C{1.1cm}} \hline
  \rowcolor{gray!50}
    $T$     &   \hphantom{median} $\nu$                &   -10    &   -2     & -1    &  -1/2  & -1/4   & -1/16 & 0 \\ \hline
  10 &  median$\left(Z_T\right)$       &   1.063  &   1.237  &   1.355  &   1.460  &   1.525  &   1.564  &   1.568  \\
     &  $\gP(Z_T<1)$                   &   0.130  &   0.157  &   0.180  &   0.207  &   0.233  &   0.265  &   0.280  \\
     &  $\E\left[1-Z_T|Z_T<1\right]$   &   0.027  &   0.101  &   0.156  &   0.213  &   0.262  &   0.315  &   0.338  \\
     &  $\E\left[1-Z_T\right]^+$       &   0.003  &   0.016  &   0.028  &   0.044  &   0.061  &   0.083  &   0.095  \\ \hline
  20 &  median$\left(Z_T\right)$       &   1.176  &   1.668  &   1.983  &   2.238  &   2.378  &   2.452  &   2.460  \\
     &  $\gP(Z_T<1)$                   &   0.032  &   0.054  &   0.074  &   0.100  &   0.124  &   0.153  &   0.166  \\
     &  $\E\left[1-Z_T|Z_T<1\right]$   &   0.033  &   0.120  &   0.178  &   0.234  &   0.279  &   0.324  &   0.343  \\
     &  $\E\left[1-Z_T\right]^+$       &   0.001  &   0.006  &   0.013  &   0.023  &   0.034  &   0.050  &   0.057  \\ \hline
  30 &  median$\left(Z_T\right)$       &   1.355  &   2.386  &   3.019  &   3.491  &   3.730  &   3.847  &   3.857  \\
     &  $\gP(Z_T<1)$                   &   0.005  &   0.015  &   0.027  &   0.043  &   0.059  &   0.080  &   0.089  \\
     &  $\E\left[1-Z_T|Z_T<1\right]$   &   0.037  &   0.126  &   0.182  &   0.234  &   0.272  &   0.311  &   0.327  \\
     &  $\E\left[1-Z_T\right]^+$       &   0.000  &   0.002  &   0.005  &   0.010  &   0.016  &   0.025  &   0.029  \\ \hline
  40 &  median$\left(Z_T\right)$       &   1.617  &   3.539  &   4.683  &   5.480  &   5.857  &   6.034  &   6.050  \\
     &  $\gP(Z_T<1)$                   &   0.001  &   0.004  &   0.009  &   0.016  &   0.025  &   0.037  &   0.043  \\
     &  $\E\left[1-Z_T|Z_T<1\right]$   &   0.038  &   0.127  &   0.180  &   0.226  &   0.260  &   0.293  &   0.307  \\
     &  $\E\left[1-Z_T\right]^+$       &   0.000  &   0.000  &   0.002  &   0.004  &   0.007  &   0.011  &   0.013  \\ \hline
  50 &  median$\left(Z_T\right)$       &   1.991  &   5.366  &   7.322  &   8.614  &   9.198  &   9.465  &   9.488  \\
     &  $\gP(Z_T<1)$                   &   0.000  &   0.001  &   0.002  &   0.006  &   0.010  &   0.016  &   0.019  \\
     &  $\E\left[1-Z_T|Z_T<1\right]$   &   0.039  &   0.126  &   0.174  &   0.216  &   0.246  &   0.276  &   0.288  \\
     &  $\E\left[1-Z_T\right]^+$       &   0.000  &   0.000  &   0.000  &   0.001  &   0.002  &   0.004  &   0.005  \\ \hline
  60 &  median$\left(Z_T\right)$       &   2.515  &   8.237  &   11.49  &   13.54  &   14.44  &   14.85  &   14.88  \\
     &  $\gP(Z_T<1)$                   &   0.000  &   0.000  &   0.001  &   0.002  &   0.003  &   0.006  &   0.007  \\
     &  $\E\left[1-Z_T|Z_T<1\right]$   &   0.040  &   0.123  &   0.168  &   0.206  &   0.233  &   0.260  &   0.271  \\
     &  $\E\left[1-Z_T\right]^+$       &   0.000  &   0.000  &   0.000  &   0.000  &   0.001  &   0.002  &   0.002  \\ \hline
\end{tabular}
\egroup
  \caption{Risk and reward statistics for the stochastic multiplier $Z_T$ in (\ref{eq:VTrateEquityStructure}) for optimal equity strategies on horizons from 10 to 60 years for a {\em moderate} degree of mean-reversion, cf.\ Table~\ref{tab:equityPar} of Section~\ref{sec:NumOptimalEquity}. The initial risk premium is equal to the stationary mean, $x_0=\bar{x}=0.045$.}
  \label{tab:equityStatisticsMod}
\end{table}

\begin{table}[ht]
  \centering
\rowcolors{2}{gray!25}{white}
\bgroup
\def\arraystretch{1.3}
\begin{tabular}{C{0.8cm}|L{3.2cm}C{1.1cm}C{1.1cm}C{1.1cm}C{1.1cm}C{1.1cm}C{1.1cm}C{1.1cm}} \hline
  \rowcolor{gray!50}
    $T$     &   \hphantom{median} $\nu$                &   -10    &   -2     & -1    &  -1/2  & -1/4   & -1/16 & 0 \\ \hline
  10 &  median$\left(Z_T\right)$       &   1.102  &   1.329  &   1.439  &   1.514  &   1.549  &   1.567  &   1.568  \\
     &  $\gP(Z_T<1)$                   &   0.075  &   0.110  &   0.136  &   0.165  &   0.188  &   0.213  &   0.224  \\
     &  $\E\left[1-Z_T|Z_T<1\right]$   &   0.029  &   0.102  &   0.146  &   0.188  &   0.219  &   0.251  &   0.264  \\
     &  $\E\left[1-Z_T\right]^+$       &   0.002  &   0.011  &   0.020  &   0.031  &   0.041  &   0.053  &   0.059  \\ \hline
  20 &  median$\left(Z_T\right)$       &   1.448  &   2.102  &   2.292  &   2.391  &   2.435  &   2.457  &   2.460  \\
     &  $\gP(Z_T<1)$                   &   0.002  &   0.011  &   0.020  &   0.032  &   0.044  &   0.060  &   0.068  \\
     &  $\E\left[1-Z_T|Z_T<1\right]$   &   0.035  &   0.100  &   0.131  &   0.158  &   0.179  &   0.203  &   0.214  \\
     &  $\E\left[1-Z_T\right]^+$       &   0.000  &   0.001  &   0.003  &   0.005  &   0.008  &   0.012  &   0.015  \\ \hline
  30 &  median$\left(Z_T\right)$       &   2.289  &   3.319  &   3.550  &   3.704  &   3.794  &   3.851  &   3.857  \\
     &  $\gP(Z_T<1)$                   &   0.000  &   0.000  &   0.001  &   0.002  &   0.006  &   0.015  &   0.021  \\
     &  $\E\left[1-Z_T|Z_T<1\right]$   &   0.033  &   0.076  &   0.097  &   0.122  &   0.148  &   0.183  &   0.201  \\
     &  $\E\left[1-Z_T\right]^+$       &   0.000  &   0.000  &   0.000  &   0.000  &   0.001  &   0.003  &   0.004  \\ \hline
  40 &  median$\left(Z_T\right)$       &   3.719  &   4.885  &   5.287  &   5.643  &   5.878  &   6.032  &   6.050  \\
     &  $\gP(Z_T<1)$                   &   0.000  &   0.000  &   0.000  &   0.000  &   0.002  &   0.007  &   0.012  \\
     &  $\E\left[1-Z_T|Z_T<1\right]$   &   0.025  &   0.054  &   0.077  &   0.111  &   0.147  &   0.196  &   0.219  \\
     &  $\E\left[1-Z_T\right]^+$       &   0.000  &   0.000  &   0.000  &   0.000  &   0.000  &   0.001  &   0.003  \\ \hline
  50 &  median$\left(Z_T\right)$       &   5.372  &   6.833  &   7.696  &   8.526  &   9.085  &   9.447  &   9.488  \\
     &  $\gP(Z_T<1)$                   &   0.000  &   0.000  &   0.000  &   0.000  &   0.001  &   0.005  &   0.009  \\
     &  $\E\left[1-Z_T|Z_T<1\right]$   &   0.017  &   0.043  &   0.074  &   0.119  &   0.163  &   0.220  &   0.246  \\
     &  $\E\left[1-Z_T\right]^+$       &   0.000  &   0.000  &   0.000  &   0.000  &   0.000  &   0.001  &   0.002  \\ \hline
  60 &  median$\left(Z_T\right)$       &   6.907  &   9.295  &   11.07  &   12.83  &   14.03  &   14.80  &   14.88  \\
     &  $\gP(Z_T<1)$                   &   0.000  &   0.000  &   0.000  &   0.000  &   0.001  &   0.004  &   0.008  \\
     &  $\E\left[1-Z_T|Z_T<1\right]$   &   0.011  &   0.042  &   0.081  &   0.134  &   0.184  &   0.245  &   0.273  \\
     &  $\E\left[1-Z_T\right]^+$       &   0.000  &   0.000  &   0.000  &   0.000  &   0.000  &   0.001  &   0.002  \\ \hline
\end{tabular}
\egroup
  \caption{Risk and reward statistics for the stochastic multiplier $Z_T$ in (\ref{eq:VTrateEquityStructure}) for optimal equity strategies on horizons from 10 to 60 years for a {\em high} degree of mean-reversion, cf.\ Table~\ref{tab:equityPar} of Section~\ref{sec:NumOptimalEquity}. The initial risk premium is equal to the stationary mean, $x_0=\bar{x}=0.045$.}
  \label{tab:equityStatisticsHigh}
\end{table}

\clearpage
\pagebreak
\addcontentsline{toc}{section}{References}
\bibliographystyle{plainnat}
\bibliography{references}
\end{document}